\documentclass[11pt,reqno]{amsart}

\usepackage{fullpage}
\usepackage[OT1]{fontenc}

\usepackage{amsthm,amsmath,amsfonts,amssymb}
\usepackage[authoryear]{natbib}
\usepackage[colorlinks,citecolor=blue,urlcolor=blue]{hyperref}
\usepackage{graphicx}
\usepackage{placeins,pdflscape}
\usepackage{algorithm}
\usepackage{algpseudocode}
\usepackage{booktabs}

\theoremstyle{plain}
\newtheorem{proposition}{Proposition}[section]
\newtheorem{corollary}[proposition]{Corollary}

\newcommand{\bs}{\boldsymbol}

\title[Amortized Bayesian Boundary Detection]
{Amortized Bayesian Disease Mapping and Boundary Detection on Heterogeneous Spatial Graphs}

\author{Luca Aiello\textsuperscript{$\ast$,1} and Sudipto Banerjee\textsuperscript{$\ast$,2}}

\thanks{\textsuperscript{$\ast$}Department of Biostatistics, University of California, Los Angeles.}
\thanks{\textsuperscript{1}\href{mailto:laiello@g.ucla.edu}{laiello@g.ucla.edu} \textsuperscript{2}
\href{mailto:sudipto@ucla.edu}{sudipto@ucla.edu}.}


\begin{document}

\begin{abstract}
Spatial disease maps help public-health researchers identify geographic
inequalities, but standard Bayesian smoothing can obscure localized
disparities when neighboring communities have sharply different socioeconomic
or behavioral profiles. Analysts therefore need to determine where smoothing
should be interrupted and repeat that analysis as maps, adjacency structures,
and outcomes change. We develop a covariate-informed Bayesian boundary model
and an amortized posterior approximation trained across heterogeneous areal
graphs. The model distinguishes local interruptions in smoothing from broader
residual spatial dependence; the trained approximation handles maps with
different numbers of regions. Simulations examine posterior calibration,
boundary-probability recovery, replicated-data behavior, and
MCSE-controlled agreement with prior-matched MCMC.
In contrast to traditional approaches that analyze these data separately, we
demonstrate the effectiveness of using a single trained deep learning network
to analyze respiratory hospitalizations in Greater Glasgow; lung cancer
incidence in California; and tracheal, bronchial, and lung cancer mortality
in South Korea, comprising 58 to 241 regions. Selected boundary density is
greatest in Glasgow and lowest in South Korea despite substantial residual
spatial dependence in both, showing that local interruption and broader spatial
persistence need not vary together. Across all three applications, edge-level
boundary probabilities agree substantially with dataset-specific analyses,
although posterior spread and thresholded boundary sets differ. These results
support reusable Bayesian boundary analysis across the evaluated disease-map
class and identify the validation needed before deployment to new applications.
\end{abstract}

\maketitle

\section{Introduction}
\label{sec:introduction}

Disease mapping is an applied tool for characterizing geographic variation in incidence, mortality, and hospitalization and for identifying local patterns that warrant scientific or public-health investigation \citep{koch2005cartographies,wakefield2007disease,lawson2016handbook,lawson2018bayesian}. Because counts from small areas can be unstable, practitioners commonly use Bayesian spatial models that borrow information across neighboring regions through conditional autoregressive (CAR) or related Gaussian Markov random field priors \citep{besag1974spatial,besag1991bayesian,leroux2000estimation}. This borrowing improves precision when disease risk varies gradually, but it can conceal localized disparities when neighboring communities lie on opposite sides of sharp socioeconomic, environmental, administrative, or behavioral contrasts. Therefore, analyzing disease maps must identify both where smoothing is useful and where it is insufficiently supported.

We consider three health applications concerning spatial disparities or boundary detection on disease maps. Respiratory hospitalizations are analyzed across 134 areas in Greater Glasgow, where neighboring communities can differ substantially in income deprivation. Lung cancer incidence is examined across the 58 counties of California, where smoking prevalence varies geographically. Tracheal, bronchial, and lung cancer mortality is studied across 241 South Korean municipalities, again using local smoking-prevalence contrasts to inform possible interruptions in smoothing. The applications differ in graph size, topology, covariate structure, and spatial heterogeneity, but pose the same substantive question: which neighboring areas can reasonably share spatial information, and which borders show evidence that this borrowing should be interrupted? Answering this question helps applied analysts distinguish broad spatial persistence from localized departures that global smoothing may obscure.

This article addresses a growing operational challenge in such applications: public (global) health organizations increasingly demand near-instantaneous disease mapping and boundary detection, particularly to interface spatial analyses with large language models (LLMs). Conventional workflows require fitting a separate Markov chain Monte Carlo (MCMC) sampler whenever outcomes, geographic resolutions, or adjacency graphs change. This computational bottleneck motivates an amortized "train-once and apply widely" paradigm. By using a single pre-trained network, researchers can analyze disparate, highly heterogeneous geographic data without model retraining or per-dataset MCMC sampling. We therefore ask whether posterior inference for a spatial boundary model can be learned across simulated graphs and transferred, with appropriate validation, to new disease maps.

Although boundary detection in areal disease mapping has a rich methodological literature, frameworks have not been extended to amortized learning settings. Early approaches identified boundaries through large posterior differences in fitted risks across neighboring areas \citep{lu2005bayesian}, while later formulations introduced edge-specific indicators or neighbor-specific weights to weaken spatial dependence across selected borders \citep{lu2007bayesian,ma2007bayesian,ma2010hierarchical}. More structured approaches relate boundary formation to observed dissimilarities between neighboring areas, reducing spatial smoothing when adjacent regions differ sufficiently in relevant covariates \citep{lee2012boundary,lee2014bayesian,rushworth2017adaptive,lee2021improved,gao2023spatial,wu2025assessing}. We follow this covariate-informed perspective but couple it with a directed acyclic graph autoregressive (DAGAR) prior \citep{datta2019spatial}. The resulting model separates two components of spatial variation: a boundary parameter determines where edges are removed because of covariate dissimilarity, while a residual dependence parameter controls spatial persistence over the retained graph. Posterior inference quantifies both local boundary evidence and residual spatial dependence.

To make this boundary model reusable across changing spatial settings, we investigate \emph{amortized Bayesian inference} (ABI). ABI trains a posterior approximator on $\text{(parameter, data)}$ pairs simulated from the generative model and reuses it for new datasets within a specified deployment regime \citep{radev2020bayesflow,sainsbury2024likelihood,zammit2025neural}. The intended benefit is a common inferential workflow across maps whose numbers of regions and adjacency structures differ. The resulting procedure remains an approximation of posterior inference under the stated statistical model and therefore requires explicit validation before practical deployment.

This transfer requires an input representation that is not tied to a fixed
number of regions. We represent each map as an unordered collection of
node-specific, graph-aware summaries constructed from observed counts,
offsets, covariates, and adjacency information. A permutation-invariant
SetTransformer \citep{zaheer2017deep,lee2019set} maps this variable-size
collection to a fixed-dimensional representation, and a conditional
normalizing flow
\citep{rezende2015variational,papamakarios2017masked,durkan2019neural,
papamakarios2021normalizing}
uses that representation to approximate the posterior distribution of the
model parameters. Training over maps that vary in size, topology, covariate
surface, boundary configuration, and residual dependence therefore yields a
single posterior operator over a class of spatial graphs rather than a
dataset-specific approximation.

Reliability is especially important when a learned approximation is intended for repeated use. We therefore evaluate parameter recovery, interval calibration, simulation-based calibration, parameter-posterior replicated-data diagnostics, and edge-level boundary probabilities on held-out maps. We also compare the learned posterior with a prior-matched MCMC implementation,
and use ablation experiments to assess the graph-aware representation. These evaluations characterize its accuracy across the evaluated regime and the aspects of inference for which material discrepancies remain.

The empirical analyses illustrate how the framework answers the same applied question in three distinct settings. Greater Glasgow has the largest selected boundary density, California has a sparser and more uncertain boundary pattern, and South Korea has the smallest boundary density despite a high posterior median for residual spatial dependence. These findings identify where a globally smooth representation is least compatible with the fitted boundary model and provide geographically specific results for substantive follow-up. Held-out simulations establish calibration and boundary recovery within the training regime, while agreement with dataset-specific Bayesian analyses supports stable transfer to empirical graphs with atypical topology and covariate smoothness.

Our intended contribution is summarized as follows. %
First, we formulate the practical disease-mapping question of where spatial borrowing should be interrupted while retaining a separate representation of broader residual dependence. Second, we develop an amortized posterior approximation that allows this analysis to be reused across areal graphs with different numbers of regions and adjacency structures. Third, we evaluate the resulting workflow through truth-based simulation, comparison with dataset-specific Bayesian computation, and three substantive health applications with distinct graph structures and boundary regimes. Together, these components provide applied statisticians with a transferable framework and an explicit validation strategy for assessing its use on new maps.

The paper proceeds as follows. Section~\ref{sec:applications} introduces the motivating data and scientific questions. Section~\ref{sec:model} develops the boundary model and the amortized inference framework. Section~\ref{sec:validation} presents the validation study; Section~\ref{sec:real_data} presents the data analyses and comparisons with dataset-specific benchmarks; and Section~\ref{sec:discussion} concludes.

\section{Spatial health applications and scientific questions}\label{sec:applications}

Together, these applications span distinct outcomes, spatial resolutions, graph
structures, and boundary-driving covariates. Their comparison asks whether
broad residual dependence and local interruptions in borrowing vary together
or constitute distinct spatial features.

\subsection{Respiratory disease in Greater Glasgow}
\label{sec:glasgow_data}

The first application considers respiratory-disease hospitalizations in 2010
for the 134 Intermediate Zones north of the River Clyde in the Greater Glasgow
and Clyde health board. The analysis files reproduce the
\texttt{respiratory}\allowbreak\texttt{data} object in the
\texttt{CARBayes}\allowbreak\texttt{data} package
\citep{CARBayes2013,CARBayesdata2022}. For each area,
the data contain the observed number of hospitalizations and an expected count
obtained by indirect standardization using Scotland-wide respiratory
hospitalization rates. The corresponding observed-to-expected ratio is the 
standardized hospitalization ratio (SHR). The areal graph contains 360 adjacent pairs, 
averaging 5.37 neighbors per area.

Socioeconomic deprivation is an important source of spatial heterogeneity in
this setting. We use the percentage of residents classified as income deprived
as the covariate defining dissimilarity between neighboring areas, following
the localized-smoothing analysis of \citet{lee2012boundary}. The resulting
application therefore asks whether a spatial risk surface that is broadly
correlated across Glasgow also contains localized interruptions in smoothing
between neighboring communities with different deprivation profiles.

\subsection{Lung cancer incidence in California}
\label{sec:california_data}

The second application considers lung cancer incidence during 2012--2016 across the 58 California counties. We use the processed county-level data analyzed by \citet{gao2023spatial}, with incidence counts originating from the Surveillance, Epidemiology, and End Results (SEER) Program of the National Cancer Institute \citep{seer}. Expected counts were obtained by indirect standardization using statewide incidence rates and county populations across 38 age--sex strata (19 age groups and two sexes). The corresponding observed-to-expected ratio is the standardized incidence ratio (SIR). The county adjacency graph contains 139 neighboring pairs, with an average of 4.79 neighbors per county.

We use county-level adult smoking prevalence for 2014--2016 from
\emph{California Tobacco Facts and Figures 2018}
\citep{californiaTobacco2018} to define dissimilarity between neighboring
counties. Smoking is a major risk factor for lung cancer and varies
substantially across California, making contrasts in smoking prevalence a
scientifically interpretable source of information about where spatial
smoothing may be inappropriate. Relative to Glasgow, the California map is
smaller and has a different graph topology and covariate distribution, allowing
us to examine the same boundary-detection question under a distinct spatial
configuration.

\subsection{Tracheal, bronchial, and lung cancer mortality in South Korea}
\label{sec:korea_data}

The third application considers deaths from malignant neoplasms of the
trachea, bronchus, and lung across South Korean municipalities.
Observed deaths were pooled over 2015--2019. Expected counts were obtained by indirect age--sex standardization using annual national cause-specific mortality rates and municipal midyear populations from the Korean Statistical Information Service \citep{kosis2026}. The corresponding observed-to-expected ratio is the standardized mortality ratio (SMR). We retain the largest connected component of the graph constructed from a 2019 municipal boundary layer \citep{koreaBoundaries2019}, comprising 241 municipalities and 626 neighboring pairs, with an average of 5.20 neighbors per municipality.

We use age-standardized current smoking prevalence from the Korean Community
Health-Related Factors Database \citep{kdca2026}, averaged over 2015--2019,
to define dissimilarity between neighboring municipalities. This application 
deals with 241 areas, providing a substantially larger
observed graph while remaining within the 40--300 area training range. It
therefore asks whether the same learned boundary mechanism can identify
localized smoking-associated discontinuities in tracheal, bronchial, and lung
cancer mortality without dataset-specific retraining.

Fig.~\ref{fig:data} summarizes the three motivating applications by displaying
the observed-to-expected outcome ratios together with the corresponding
boundary-driving covariates. The maps illustrate the different spatial scales
of the applications and the local covariate contrasts motivating localized
interruptions in smoothing.

\begin{figure}[t]
\centering
\includegraphics[height=0.26\textwidth,keepaspectratio]{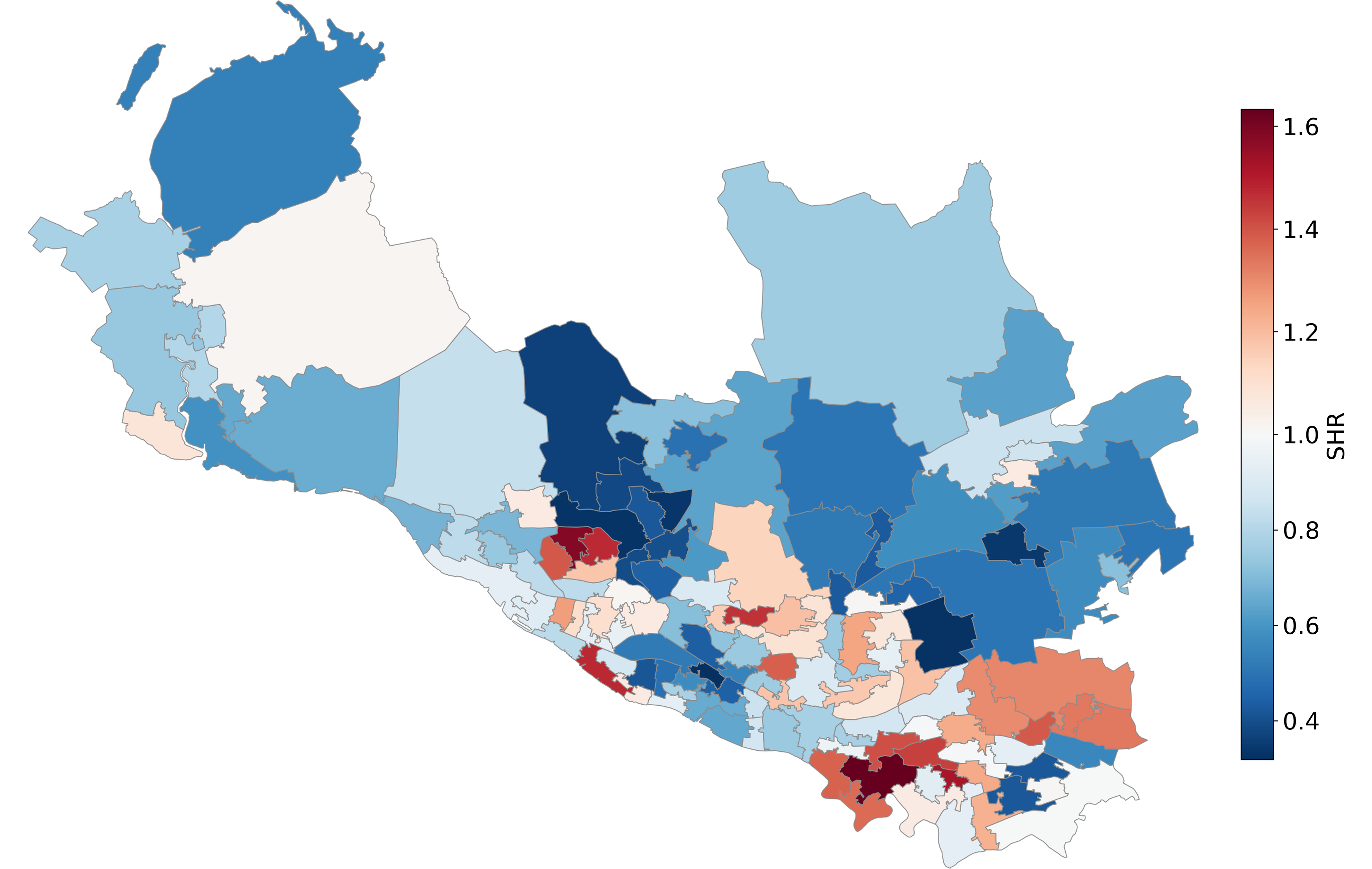}\hfill
\includegraphics[height=0.26\textwidth,keepaspectratio]{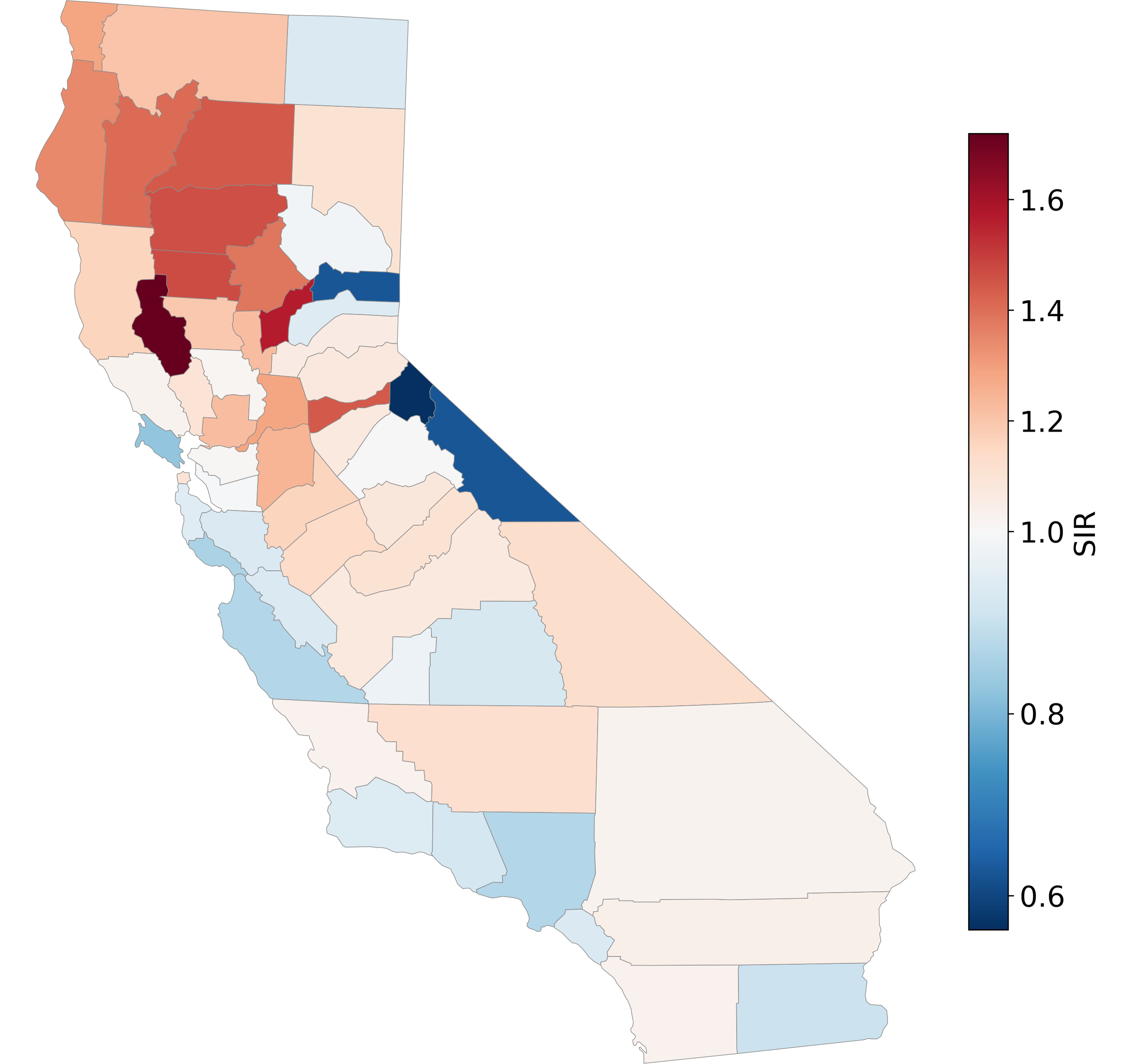}\hfill
\includegraphics[height=0.26\textwidth,keepaspectratio]{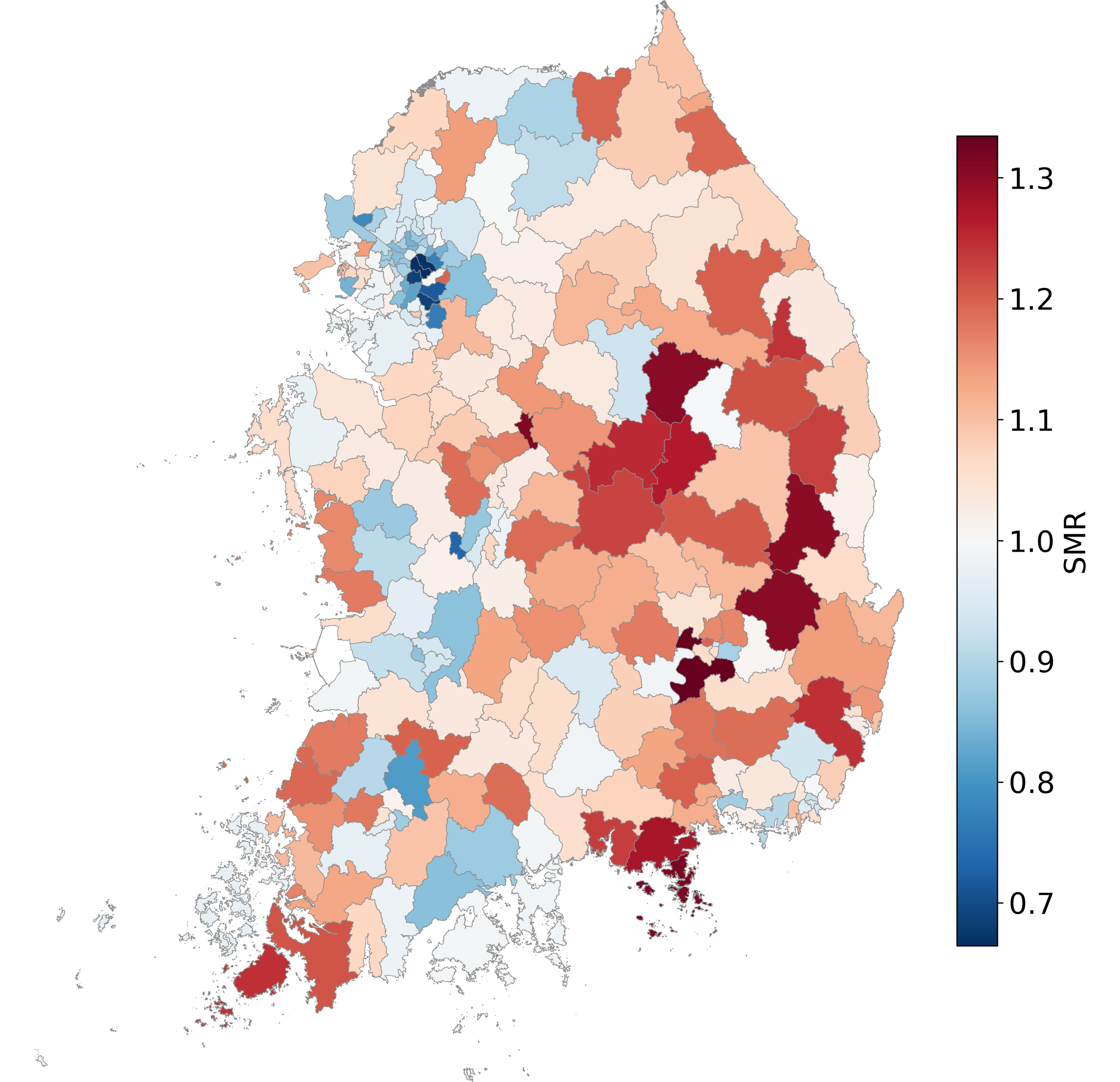}\\
\includegraphics[height=0.26\textwidth,keepaspectratio]{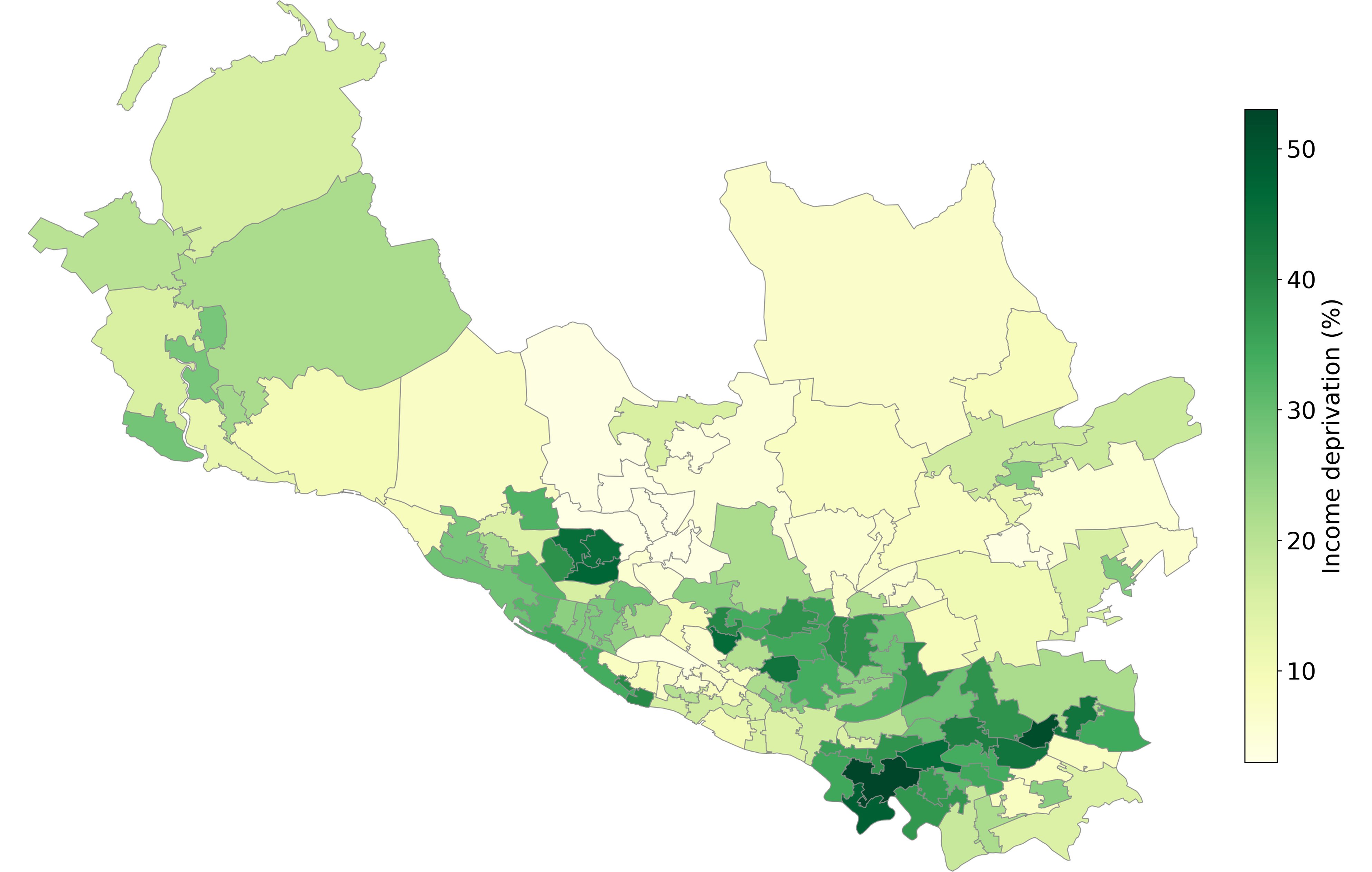}\hfill
\includegraphics[height=0.26\textwidth,keepaspectratio]{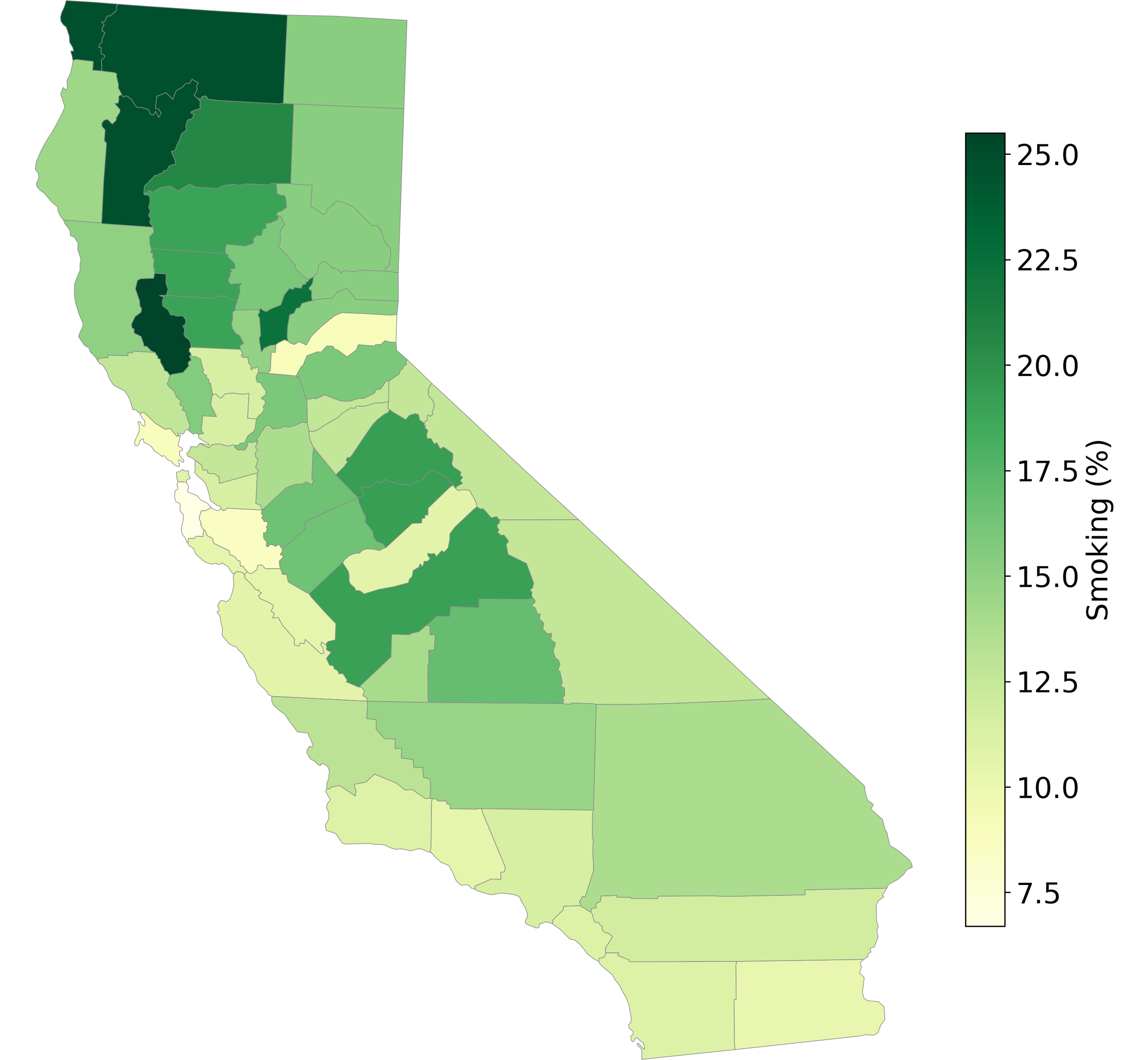}\hfill
\includegraphics[height=0.26\textwidth,keepaspectratio]{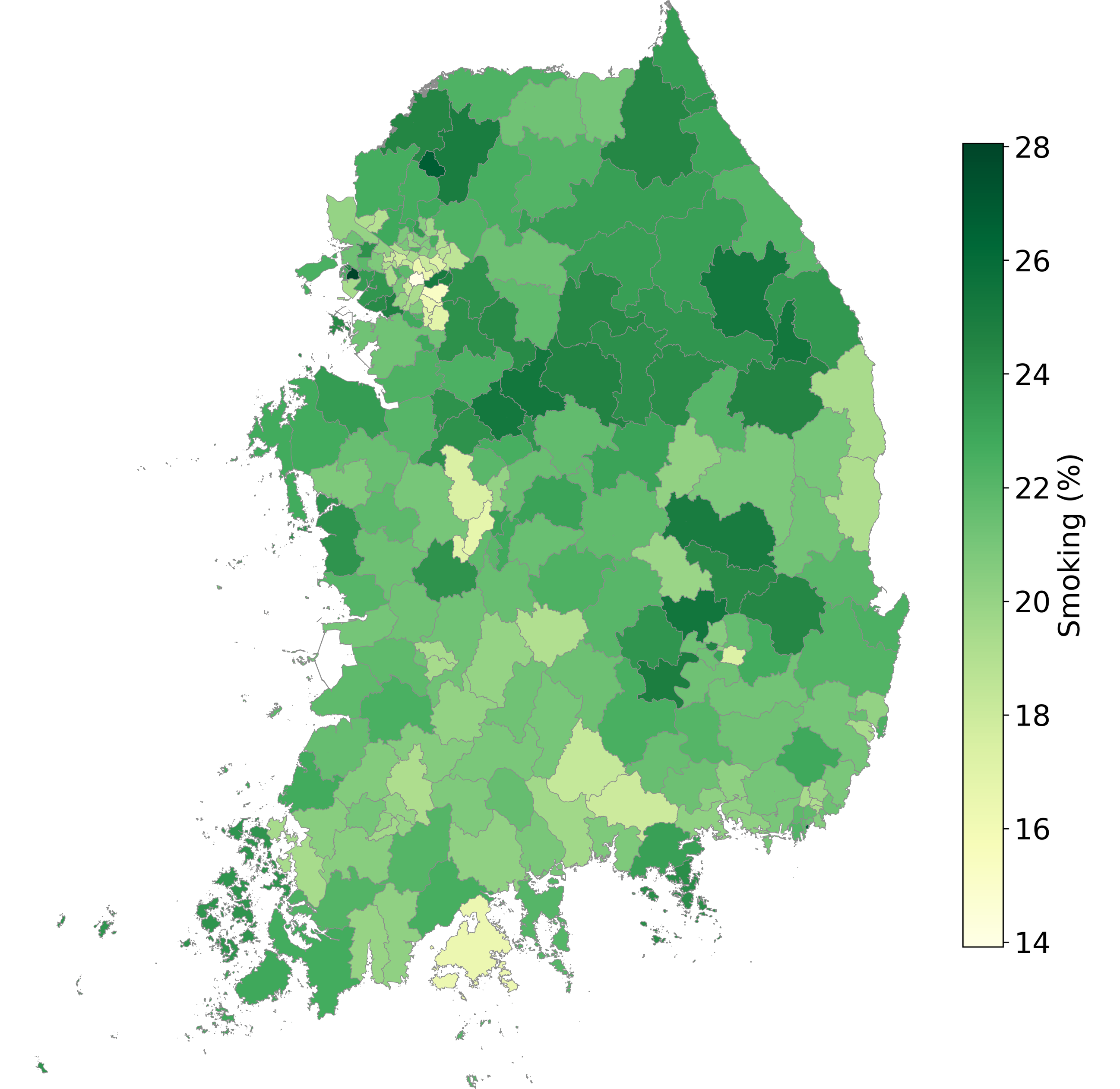}
\caption{Health outcomes and boundary-driving covariates in Greater Glasgow, California, and South Korea (left to right). Top: observed-to-expected ratios (SHR, SIR, and SMR, respectively). Bottom: income deprivation in Glasgow and smoking prevalence in California and South Korea (\%).}
\label{fig:data}
\end{figure}

The health outcomes of substantive interest are geographic patterns in hospitalization, cancer incidence, and mortality, summarized by the SHR in Greater Glasgow, SIR in California, and SMR in South Korea. Each ratio compares the observed count with that expected under the corresponding reference rates: values above one indicate more hospitalizations, cases, or deaths than expected, whereas values below one indicate fewer. Our analysis investigates the spatial structure underlying these outcome patterns, identifying where neighboring areas can share information and where the evidence supports interruptions in smoothing. For inference, we model the observed counts using expected counts as offsets, thereby accounting for differences in the precision of the observed ratios.

Across applications, residual dependence and boundary formation vary
independently. Glasgow combines strong persistence with widespread
interruptions, California shows a sparser and more uncertain pattern, and
South Korea combines strong persistence with localized boundaries on the
largest graph. The Glasgow--South Korea contrast is especially informative:
similar residual persistence accompanies markedly different local borrowing
structures. Thus, strong overall dependence does not imply that smoothing is
appropriate across every border. A common framework makes these regimes
comparable.

Fig.~\ref{fig:application_structure} examines neighboring similarity and its relationship with covariate differences. We summarize crude log risk as $r_i = \log(y_i+0.5)-\log(e_i)$. The top panels compare absolute differences between adjacent areas and an equally sized random sample of non-neighboring pairs. Adjacent contrasts are smaller on average: $0.346$ versus $0.477$ in Greater Glasgow, $0.180$ versus $0.235$ in California, and $0.105$ versus $0.140$ in South Korea. Together with positive Moran's $I$ values ($0.415$, $0.283$, and $0.466$, respectively), this indicates that adjacent areas tend to have similar risks, supporting spatial borrowing.

The bottom panels examine whether neighboring similarity weakens as deprivation or smoking prevalence differs more sharply across a border. Adjacent pairs are grouped from the lowest (Q1) to the highest (Q5) covariate-dissimilarity quintile. Mean contrasts in Q5 are $2.980$, $1.300$, and $1.380$ times those in Q1 for Glasgow, California, and South Korea, respectively. The increase is clearest in Glasgow and is not monotone in the other applications. Vertical bars describe variation across edges within each quintile, not uncertainty in the means. These patterns motivate spatial borrowing that can be relaxed at highly dissimilar borders; individual boundaries require model-based assessment.

\begin{figure}[t]
\centering
\includegraphics[width=\textwidth]{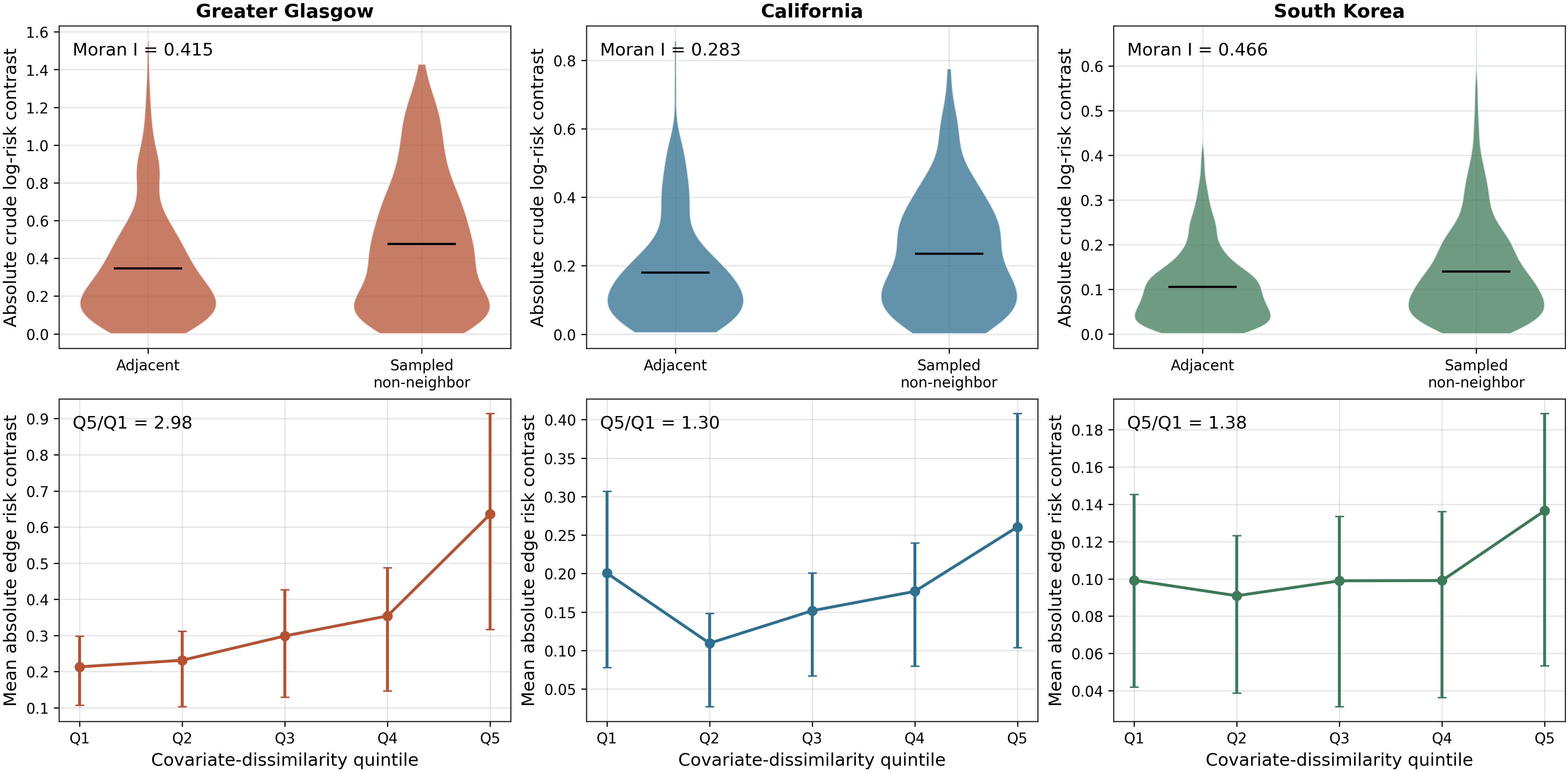}
\caption{Spatial similarity and covariate contrasts in Greater Glasgow, California, and South Korea (left to right). Top: crude log-risk differences for neighboring and sampled non-neighboring areas, with mean bars. Bottom: neighboring differences by covariate-dissimilarity quintile, with means and interquartile ranges.}
\label{fig:application_structure}
\end{figure}

The three applications share the same observed-data structure, i.e., counts,
expected counts, one boundary-driving covariate, and an areal graph, and all
fall within the 40--300 area training range. Their contiguity graphs are sparser
and their boundary-driving covariates
spatially smoother than typical training configurations, providing a meaningful
test of transfer beyond the simulator's central regime. Supplementary
Section~S4.2 documents these departures, while
Sections~\ref{sec:validation} and~\ref{sec:bayesian_benchmarks} evaluate whether
the resulting transfer remains reliable.

\subsection{Scientific and inferential questions}
\label{sec:questions}

The three applications motivate two substantive questions and a broader
inferential challenge. First, where do the disease maps provide evidence that
spatial smoothing should be interrupted between neighboring areas, and how is
this boundary structure related to geographic differences in income deprivation
or smoking prevalence? Second, how do such localized discontinuities coexist
with residual spatial dependence, and does the balance between boundary
formation and spatial persistence differ across the Glasgow, California, and
South Korean applications?

Beyond these application-specific questions, the central inferential challenge
is whether the Bayesian analysis itself can be made reusable across
heterogeneous spatial graphs. In a conventional workflow, each new disease map
would require a separate MCMC analysis, despite the fact that the underlying
statistical model and inferential targets remain unchanged. We therefore ask
whether a single posterior inference mechanism can instead be learned once
over a sufficiently rich distribution of simulated graphs and subsequently
reused to generate posterior draws for new maps that differ in dimension,
topology, and covariate structure. Such a procedure is useful only if this
transfer preserves the inferential content of the original Bayesian analysis:
posterior uncertainty must remain well calibrated, edge-level boundary evidence
must remain reliable, and the resulting posterior summaries should agree
closely with those obtained from dataset-specific MCMC.
These questions motivate the development of a boundary-detection model and a reusable
amortized posterior framework.

\section{Boundary detection with reusable Bayesian inference}
\label{sec:model}

The data and objectives in Section~\ref{sec:applications} require an analysis
that distinguishes localized interruptions in smoothing from broader residual
dependence across graphs with different sizes and adjacency structures. We
first specify a covariate-informed boundary model and then construct an
amortized posterior approximation across heterogeneous areal graphs.

\subsection{Covariate-informed spatial boundary model}
\label{sec:boundary_model}

Let $i=1,\ldots,N$ index the areal units, $y_i$ denote the observed disease
count, $e_i$ the expected count, and $x_i$ an area-level covariate used to
characterize boundary information. We specify
\begin{equation}\label{eq:likelihood}
y_i\mid w_i,\beta_0 \sim \operatorname{Pois}(\lambda_i),
\qquad
\log\lambda_i=\log e_i+\beta_0+w_i ,
\end{equation}
so that $\exp(\beta_0+w_i)$ is the relative risk for area $i$, $\beta_0$ is an overall intercept, and $w_i$
captures residual spatial variation. Let $\bs{A}=(a_{ij})$ denote the observed geographic adjacency matrix and define
the covariate dissimilarity between neighboring areas by
$z_{ij}=|x_i-x_j|$. Following the covariate-informed localized-smoothing
perspective of \citet{lee2012boundary}, we construct an effective adjacency
matrix $\bs{A}^*=(a^*_{ij})$ as
\begin{equation}\label{eq:boundary}
a^*_{ij}
=
a_{ij}\,
\mathbb{I}\!\left(z_{ij}\eta\leq\log 2\right).
\end{equation}
Thus, geographic neighbors are prevented from smoothing across their shared
border when their covariate dissimilarity is sufficiently large. Larger values
of the boundary parameter $\eta$ remove more dissimilar edges and therefore
support stronger local discontinuities. Because $\eta$ is inferred from the
data, each observed edge has posterior boundary probability
\begin{equation}\label{eq:boundary_probability}
p_{ij}
=
\Pr(a^*_{ij}=0\mid \mathcal{D})
=
\Pr(z_{ij}\eta>\log 2\mid \mathcal{D}),
\qquad a_{ij}=1,
\end{equation}
where $\mathcal{D}$ comprises the observed counts, offsets, covariates, and
geographic graph; $p_{ij}$ is the edge-level quantity used for boundary
ranking and selection.

Here, a boundary denotes removal of an edge from the adjacency structure used to 
construct the DAGAR prior; it does not directly assert that the disease-risk contrast 
across that border exceeds a specified magnitude. For a given dataset, posterior 
boundary probabilities are non-decreasing functions of covariate dissimilarity, because 
every edge shares the same posterior distribution of $\eta$. The outcomes inform this 
common distribution, while the covariate dissimilarities determine the ordering of candidate edges.

Residual spatial dependence over $\bs{A}^*$ is modeled using a DAGAR prior \citep{datta2019spatial},
\begin{equation}\label{eq:dagar}
\bs{w}\mid\sigma_w^2,\rho,\bs{A}^*
\sim
\mathcal{N}\!\left(\bs{0},\sigma_w^2 \bs{Q}(\rho;\bs{A}^*)^{-1}\right),
\end{equation}
where $\sigma_w^2>0$ controls residual spatial variability and
$\rho\in(0,1)$ controls spatial dependence over retained neighboring pairs.
The DAGAR precision $\bs{Q}(\rho;\bs{A}^*)$ is constructed from a fixed ordering of the
areal units; its complete specification is given in Supplementary
Section~S1.1.

The modified graph $\bs{A}^*$ and the residual dependence parameter
$\rho$ play complementary roles: $\eta$ determines which geographic edges are
interrupted according to covariate dissimilarity, whereas $\rho$ controls
residual smoothing over the retained graph. This construction is related to
the adjacency-modeling framework of \citet{aiello2023detecting}, but here the
objective is boundary detection. Unlike the localized CAR formulation of
\citet{lee2012boundary}, in which residual spatial dependence is fixed close to
one, we estimate $\rho$ jointly with $\eta$, allowing local interruptions in
borrowing and broader spatial persistence to vary separately.

Let $\bs{\theta}=(\beta_0,\sigma_w^2,\eta,\rho)$. We assign the priors $\beta_0\sim \mathcal{N}(0,\sigma_\beta^2)$ and $\sigma_w^2\sim\left|\mathcal{N}(0,0.5)\right|$ to the intercept and residual variance, respectively. For the boundary and residual dependence parameters, we use $\eta\sim\mathcal{U}(0,M)$ and $\rho\sim\mathcal{U}(0,1)$. Here, $M=\log 2/Z_{0.5}$, with $Z_{0.5}=\operatorname{median}\{z_{ij}:a_{ij}=1\}$ assumed to be positive, adapts $\eta$ to the observed dissimilarity scale and makes only edges with above-median dissimilarity removable. Further details on the prior specification and graph-specific scaling are provided in Supplementary Section~S1.2.

\subsection{Amortized posterior inference across varying-size graphs}
\label{sec:abi}

Posterior computation for the model in \eqref{eq:likelihood}--\eqref{eq:dagar}, 
with the parameter priors specified in Section~\ref{sec:boundary_model}, 
can be performed separately for each dataset but must be repeated whenever the
disease map changes. We instead use ABI to learn a reusable approximation of the
posterior of $\bs{\theta}$ from simulated parameter--data pairs
\citep{radev2020bayesflow,sainsbury2024likelihood,
zammit2025neural}. Training data are generated from $\mathcal{D}^{(g)}\sim p(\mathcal{D}\mid\bs{\theta}^{(g)})$ 
with $\bs{\theta}^{(g)}\sim p(\bs{\theta})$ where $g=1,\ldots,G$ and
$\mathcal{D}$ comprises the observable counts, offsets, covariates, and geographic
adjacency matrix. Graph size, topology, covariate structure, boundary
configuration, and residual dependence vary across simulated replicates, so the learned
posterior approximation is not tied to a single spatial layout. The specific
training distribution defining this deployment regime is described in
Section~\ref{sec:simulation_design}.

A central challenge is that the dimension of $\mathcal{D}$ changes with the number of
areas $N$. We therefore represent each dataset as an unordered collection
$\mathcal{S}(\mathcal{D})=\{\bs{s}_1,\ldots,\bs{s}_N\}$
of node-specific, graph-aware features constructed entirely from observable
quantities. These features summarize the local count and offset information,
residual agreement among geographic neighbors, covariate dissimilarities across
edges, and graph-level spatial dependence. Counts and offsets mainly inform 
$\beta_0$ and $\sigma_w^2$; neighborhood
roughness and graph-level autocorrelation inform $\sigma_w^2$ and $\rho$;
covariate dissimilarities and neighboring contrasts inform $\eta$.

The complete feature construction is given in Supplementary Section~S2.1. Importantly, neither the latent spatial effects $\bs{w}$ nor the filtered graph $\bs{A}^*$ is supplied to the network, since these quantities are unavailable in empirical applications. A permutation-invariant summary network maps the variable-size input to a
fixed-dimensional representation,
\begin{equation}
\bs{z}=h_{\bs{\psi}}\left(\mathcal{S}(\mathcal{D})\right).
\label{eq:summary}
\end{equation}
We use a SetTransformer for the mapping in \eqref{eq:summary}, allowing the same network
to process maps with different numbers of regions. Its output is invariant to permutations 
of the supplied node-feature rows. This encoder property does not remove the
ordering dependence of the DAGAR prior: the implemented simulator and analyses
use a fixed identity ordering of the supplied area indices. Conditional on
$\bs{z}$, a normalizing flow 
approximates the posterior through
\begin{equation}
q_{\bs{\phi}}(\bs{\theta}\mid\bs{z})
=
q_{\bs{\phi}}\!\left(
\bs{\theta}\mid
h_{\bs{\psi}}\left(\mathcal{S}(\mathcal{D})\right)
\right).
\label{eq:flow}
\end{equation}
The summary and inference networks \eqref{eq:summary}--\eqref{eq:flow} are trained jointly by minimizing,
\begin{equation}
(\widehat{\bs{\phi}},\widehat{\bs{\psi}})
=
\arg\min_{(\bs{\phi},\bs{\psi})}
-
\mathbb{E}_{p(\bs{\theta},\mathcal{D})}
\left[
\log q_{\bs{\phi}}\!\left(
\bs{\theta}\mid
h_{\bs{\psi}}\left(\mathcal{S}(\mathcal{D})\right)
\right)
\right].
\label{eq:loss}
\end{equation}

The SetTransformer handles variation in graph size by mapping each
node-feature set to a fixed-dimensional representation, while the conditional
flow performs posterior inference in the common four-parameter space.
To make the probabilistic construction of
$q_{\bs{\phi}}(\bs{\theta}\mid\bs{z})$ explicit, let
$f_{\bs{\phi}}(\cdot;\bs{z})$ denote the complete invertible transformation from
the parameter support to the Gaussian reference space, including any fixed
one-to-one support transformations together with the learned flow
transformation. Thus,
\begin{equation*}
\bs{u}
=
f_{\bs{\phi}}(\bs{\theta};\bs{z}),
\qquad
\bs{\theta}
=
f_{\bs{\phi}}^{-1}(\bs{u};\bs{z}),
\qquad
\bs{u}\sim\mathcal N(\bs{0},\bs{I}),
\end{equation*}
which induces
\begin{equation*}
q_{\bs{\phi}}(\bs{\theta}\mid\bs{z})
=
p_{\bs{u}}
(f_{\bs{\phi}}(\bs{\theta};\bs{z}))
\left|
\det
\frac{\partial f_{\bs{\phi}}(\bs{\theta};\bs{z})}
{\partial\bs{\theta}}
\right|.
\end{equation*}
For each fixed dataset and representation $\bs{z}$, this transformation is
bijective, so the change-of-variables formula implies that
$q_{\bs{\phi}}(\cdot\mid\bs{z})$ integrates to one over the corresponding
parameter support.
Thus, the learned representation affects the information on which posterior
inference is conditioned, but not the validity of the conditional flow as a
probability density. The corresponding normalization argument is given in
Supplementary Section~S2.2.

Let $\mathcal L(\bs{\phi},\bs{\psi})$ denote the minimization objective in
\eqref{eq:loss}. Then
\begin{align}\label{eq:loss_kl}
\mathcal L(\bs{\phi},\bs{\psi})
=
H(\bs{\theta}\mid\mathcal D)
+
\mathbb E_{p(\mathcal D)}
\left[
\operatorname{KL}
\left\{
p(\bs{\theta}\mid\mathcal D)
\,\Vert\,
q_{\bs{\phi}}(\bs{\theta}\mid\bs{z})
\right\}
\right],
\end{align}
where $H(\bs{\theta}\mid\mathcal D) = -\mathbb E_{p(\bs{\theta},\mathcal D)}
[\log p(\bs{\theta}\mid\mathcal D)]$ is the conditional entropy, and 
$\operatorname{KL}\{r\,\Vert\,s\}$ denotes the Kullback--Leibler 
divergence between densities $r$ and $s$ on a common parameter
space $\Theta$. Because $H(\bs{\theta}\mid\mathcal D)$ is independent 
of $(\bs{\phi},\bs{\psi})$, minimizing $\mathcal L(\bs{\phi},\bs{\psi})$ is equivalent
to minimizing the second term on the right hand side of \eqref{eq:loss_kl}. 
More details in Supplementary Section~S2.2.

The presence of the summary network introduces a distinct approximation
question. For a fixed representation map $h_{\bs{\psi}}$, the conditional
log-density objective is minimized, over unrestricted conditional densities,
by $p(\bs{\theta}\mid\bs{z})$. This conditional posterior coincides with the
full-data posterior $p(\bs{\theta}\mid\mathcal D)$ when the representation is
posterior-sufficient, that is, $\bs{\theta}\perp\!\!\!\perp\mathcal D\mid\bs{z}$. 
Joint training therefore encourages the representation to retain information
relevant to posterior inference, but exact sufficiency of the learned
representation is not assumed here. This distinction is consistent with
\citet{radev2020bayesflow}, who establish exact posterior recovery under
perfect convergence and sufficient summaries existence, while
identifying information loss through the summary network as a distinct source
of approximation error in practice.

The KL expectation term in \eqref{eq:loss_kl} contains two conceptually
distinct components. Proposition~\ref{prop:abi_decomposition} separates
information lost through the learned representation from conditional-density
approximation error. Its proof and the resulting exact-recovery condition are
given in Supplementary Section~S2.2. All distributions in the proposition are induced by the joint
prior-predictive law
$p(\bs{\theta},\mathcal D)
=
p(\bs{\theta})p(\mathcal D\mid\bs{\theta})$,
with
$\bs{z}=h_{\bs{\psi}}(\mathcal S(\mathcal D))$.
Here $p(\mathcal D)$ denotes the prior-predictive marginal and
$p(\bs{z})$ its induced distribution under the learned representation.
The dependence of $p(\bs{z})$ and
$p(\bs{\theta}\mid\bs{z})$ on $\bs{\psi}$ is suppressed for notational
simplicity.
\begin{proposition}[Approximation-error decomposition]
\label{prop:abi_decomposition}
Let $\bs{z}=h_{\bs{\psi}}(\mathcal S(\mathcal D))$ be a deterministic
representation of $\mathcal D$, and let
$q_{\bs{\phi}}(\bs{\theta}\mid\bs{z})$ be a conditional density. Assume that the
conditional densities below are defined with respect to a common dominating
measure and that the displayed expected KL divergences are finite.
Then
\begin{align*}
\mathbb E_{p(\mathcal D)}
\left[
\operatorname{KL}
\left\{
p(\bs{\theta}\mid\mathcal D)
\,\Vert\,
q_{\bs{\phi}}(\bs{\theta}\mid\bs{z})
\right\}
\right]
=
I(\bs{\theta};\mathcal D\mid\bs{z}) +
\mathbb E_{p(\bs{z})}
\left[
\operatorname{KL}
\left\{
p(\bs{\theta}\mid\bs{z})
\,\Vert\,
q_{\bs{\phi}}(\bs{\theta}\mid\bs{z})
\right\}
\right],
\end{align*}
where
$I(\bs{\theta};\mathcal D\mid\bs{z})
=
\mathbb E_{p(\mathcal D)}
\left[
\operatorname{KL}
\left\{
p(\bs{\theta}\mid\mathcal D)
\,\Vert\,
p(\bs{\theta}\mid\bs{z})
\right\}
\right]$
is the conditional mutual information between $\bs{\theta}$ and $\mathcal D$
given $\bs{z}$ under the induced joint distribution.
\end{proposition}
The two terms on the right hand side in Proposition~\ref{prop:abi_decomposition} have distinct
inferential interpretations. The first, $I(\bs{\theta};\mathcal D\mid\bs{z})$, measures posterior-relevant information in the full data that is not retained by the learned representation and vanishes under posterior sufficiency. The second measures the conditional-density approximation error of the normalizing flow given that representation. Thus, even an exact conditional approximation to $p(\bs{\theta}\mid\bs{z})$ need not equal the full-data posterior if the representation discards parameter-relevant information. Finite simulation and optimization additionally determine how closely the objective is attained.

For a new observed map $\mathcal{D}_{\mathrm{obs}}$, we first construct the same
observable graph-aware features used during training and compute 
$\bs{z}_{\mathrm{obs}} = h_{\widehat{\bs{\psi}}}\left(\mathcal{S}(\mathcal{D}_{\mathrm{obs}})\right)$.
Posterior draws are then generated from
$q_{\widehat{\bs{\phi}}}
(\bs{\theta}\mid\bs{z}_{\mathrm{obs}})$, and boundary probabilities are 
estimated by averaging \eqref{eq:boundary} over the posterior draws of $\eta$. 
We refer to this entire amortized framework as
ABI-DAGAR.

\subsection{Deployment regime and validation strategy}
\label{sec:deployment}

The amortized draws are samples from
\begin{equation*}
q_{\widehat{\bs{\phi}}}
(\bs{\theta}\mid
h_{\widehat{\bs{\psi}}}(\mathcal S(\mathcal D))),
\end{equation*}
which we use as an approximation to
$p(\bs{\theta}\mid\mathcal D)$.
Proposition~\ref{prop:abi_decomposition} separates the representation and
conditional-density components of this approximation, so its adequacy must be
evaluated over the intended deployment regime. Moreover, the trained network
is intended for datasets resembling those generated by its training
distribution; substantially different graph structures, generated outcomes, or 
covariate behavior may require renewed validation or retraining.

Validation is therefore part of the inferential workflow.
Section~\ref{sec:validation} evaluates posterior calibration,
boundary-probability quality, replicated-data behavior, agreement with
prior-matched DAGAR MCMC, and the graph-aware representation before empirical
deployment. Fig.~\ref{fig:abi_workflow} summarizes this train
–validate–reuse workflow; the fitted SetTransformer and conditional normalizing flow are
reused unchanged across all three applications.

\begin{figure}[t]
\centering
\includegraphics[width=\textwidth]{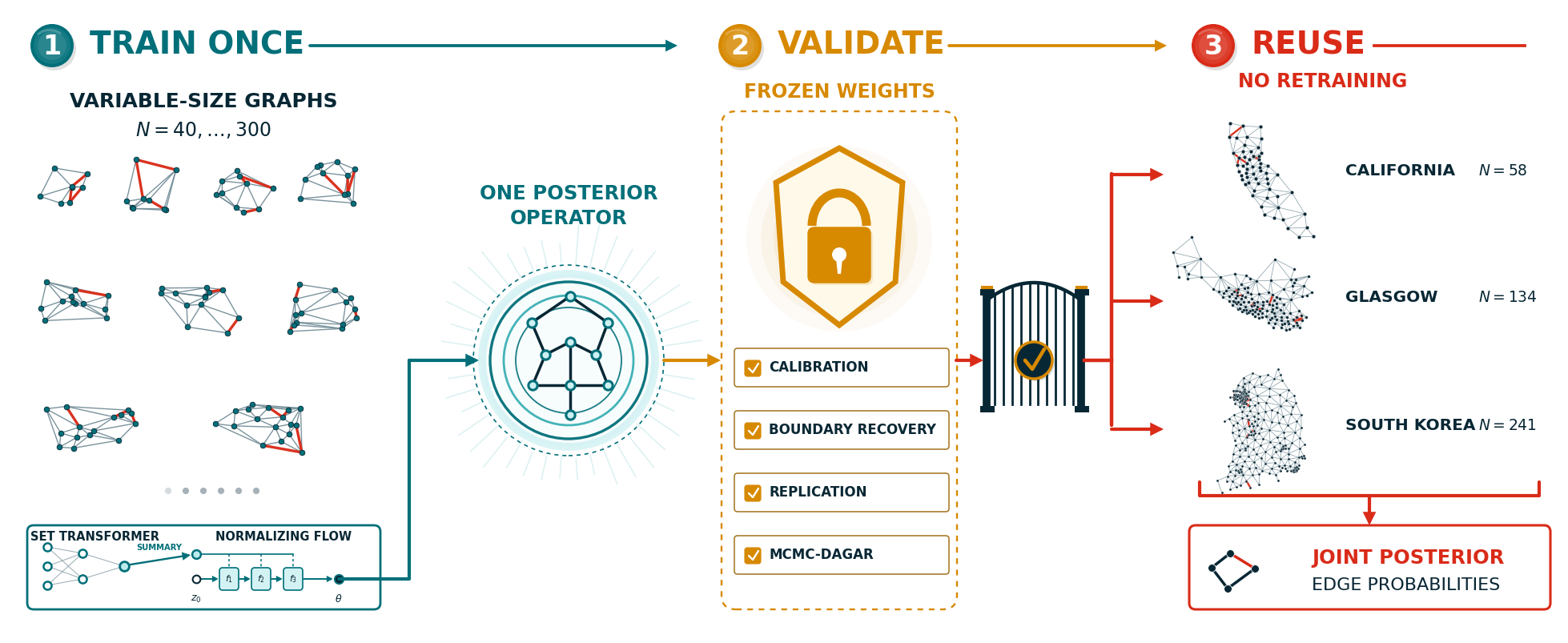}
\caption{ABI-DAGAR workflow. A common posterior operator is trained on variable-size simulated graphs, validated on held-out data, and reused without retraining across heterogeneous disease maps.}
\label{fig:abi_workflow}
\end{figure}

\section{Validation of the amortized posterior approximation}
\label{sec:validation}

We assess held-out parameter calibration, edge-level boundary probabilities, and agreement with prior-matched Bayesian computation. Replication experiments, representation ablation, computational details, and expanded results are reported in Supplementary Section~S3.

\subsection{Simulation design and posterior calibration}
\label{sec:simulation_design}

Training simulations used graphs with 40--300 areas obtained by Delaunay triangulation of random planar coordinates. Covariates, exposures, model parameters, filtered graphs, DAGAR effects, and Poisson counts were generated sequentially from the model in Section~\ref{sec:boundary_model}. Online training used 1,280,000 simulated datasets over 100 epochs; training and deployment are detailed in Supplementary Section~S2.3, with simulation specifications in Sections~S2.4 and S3.1.

We evaluated the fitted approximation on 200 independently generated held-out
datasets using 10,000 posterior draws per dataset. For
$\bs{\theta}=(\beta_0,\sigma_w^2,\eta,\rho)$, absolute posterior-mean biases
did not exceed $0.020$; root mean squared errors (RMSEs) were $0.019$, $0.164$,
$0.150$, and $0.158$, respectively, and empirical 95\% coverage ranged from
$0.945$ to $0.980$. Simulation-based calibration and recovery
diagnostics showed no systematic miscalibration or deterioration with graph
size or edge count. More details are reported in Supplementary
Section~S3.1.

\subsection{Boundary-probability validation}
\label{sec:boundary_validation}

We next evaluated the edge-level posterior boundary probabilities in
\eqref{eq:boundary_probability}. Pooling geographic edges across the 200
held-out datasets, these probabilities yielded an area under the receiver
operating characteristic curve (AUROC) of $0.974$, average precision of
$0.905$, and a Brier score of $0.053$. Formal definitions and complementary
dataset-level summaries are provided in Supplementary Section~S3.2. The AUROC
measures how well the probabilities rank true boundaries above non-boundaries,
average precision summarizes precision across levels of sensitivity,
and the Brier score measures the mean squared difference between posterior
boundary probabilities and binary boundary indicators, with lower values
indicating greater probabilistic accuracy.

Fig.~\ref{fig:boundary_calibration} examines calibration directly. Empirical
boundary frequencies closely follow the corresponding posterior probabilities
in the reliability diagram and remain aligned across edge-dissimilarity
strata. Thus, the amortized approximation recovers not only the ordering of
candidate boundaries but also meaningful variation on the probability scale.

\begin{figure}[t]
\centering
\includegraphics[width=\textwidth]{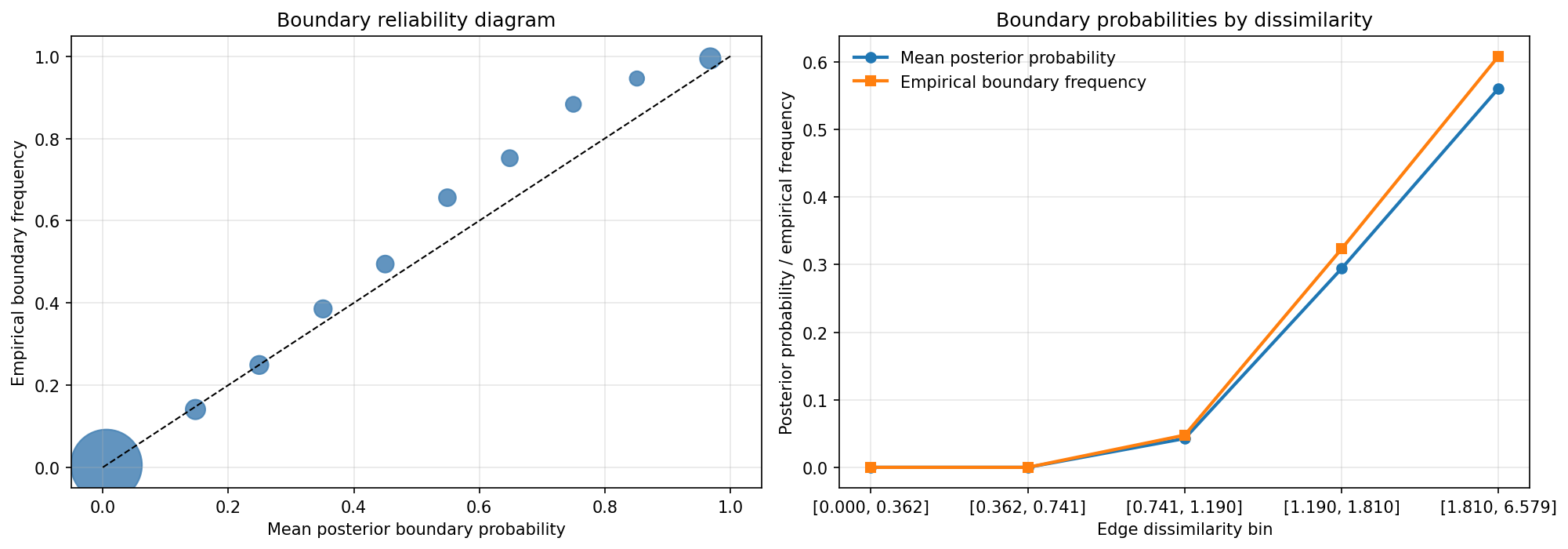}
\caption{Boundary-probability calibration on held-out simulated datasets.
Left: empirical boundary frequency versus mean posterior boundary probability,
with point diameter proportional to the number of edges in each probability
bin. Right: posterior and empirical boundary frequencies by edge-dissimilarity
bin.}
\label{fig:boundary_calibration}
\end{figure}

Under the median-probability rule, which classifies an edge as a boundary when
its posterior probability exceeds $0.5$, mean sensitivity and specificity
were $0.715$ and $0.970$, respectively. The 95\% posterior interval for the
total boundary count covered the truth in $0.975$ of datasets. The Bayesian
false discovery rate (FDR) rule of \citet{li2015bayesian}, applied with
$\alpha=0.05$, was more conservative, with mean sensitivity and specificity
of $0.170$ and $0.998$, respectively; it identified no true boundaries in 121
of the 161 datasets containing boundaries. The FDR rule is appropriate when
limiting the posterior expected proportion of false selections is primary,
whereas the median-probability rule offers greater sensitivity for exploratory
boundary identification while retaining high specificity.

Because thresholded classifications depend on the operating rule, we treat
posterior boundary probabilities as the primary inferential output and use
the median-probability rule to construct descriptive boundary maps. 
Additional threshold-based diagnostics and comparisons between the two
operating rules are provided in Supplementary Section~S3.2.

\subsection{Matched posterior benchmark and recovery-error decomposition}
\label{sec:matched_mcmc_validation}

Held-out calibration evaluates recovery under repeated draws from the training
distribution, but it does not distinguish statistical estimation error under
the target posterior from discrepancies introduced by the amortized
approximation. We therefore analyzed an additional fixed bank of 100 held-out
datasets using both ABI-DAGAR and dataset-specific MCMC. Both methods use the priors in
Section~\ref{sec:boundary_model}, the same Poisson likelihood, DAGAR
construction, and boundary mechanism. Accordingly, differences between the
two analyses reflect posterior-approximation and posterior-simulation error.
Each method retained 10,000 posterior draws; each MCMC fit used 20,000
iterations, with the first 10,000 discarded.

For dataset $g$, let $\bs{\theta}_g$ denote the generating parameter vector,
and let
\[
\bs a_g
=
\mathbb E_{q_{\widehat{\bs{\phi}}}(\bs{\theta}\mid\bs z_g)}
[\bs{\theta}],
\qquad
\bs m_g
=
\mathbb E_{p(\bs{\theta}\mid\mathcal D_g)}
[\bs{\theta}],
\]
where
$\bs z_g=h_{\widehat{\bs{\psi}}}(\mathcal S(\mathcal D_g))$.
The posterior-mean recovery error satisfies
\[
\|\bs a_g-\bs{\theta}_g\|_2^2
=
\|\bs m_g-\bs{\theta}_g\|_2^2
+
\|\bs a_g-\bs m_g\|_2^2
+
2(\bs m_g-\bs{\theta}_g)^\top(\bs a_g-\bs m_g).
\]
For parameter $j$, averaging the corresponding scalar identity over the
$G=100$ datasets gives
\begin{align*}
T_j&=G^{-1}\sum_g(a_{gj}-\theta_{gj})^2, &
R_j&=G^{-1}\sum_g(m_{gj}-\theta_{gj})^2,\\
D_j&=G^{-1}\sum_g(a_{gj}-m_{gj})^2, &
C_j&=2G^{-1}\sum_g(m_{gj}-\theta_{gj})(a_{gj}-m_{gj}).
\end{align*}
so that $T_j=R_j+D_j+C_j$.
Here $T_j$ and $R_j$ are the ABI and reference posterior mean squared error (MSE), 
$D_j$ is the ABI--MCMC posterior-mean discrepancy,
defined as the across-dataset mean squared difference between the two
posterior means, and $C_j$ is the interaction term in the recovery-error
identity, which can be negative.

The quantity $D_j$ is a moment-level diagnostic of posterior approximation,
not a divergence between complete posterior distributions. In particular, it
does not estimate either KL term in
Proposition~\ref{prop:abi_decomposition}. Rather, because matched MCMC
approximates $p(\bs{\theta}\mid\mathcal D)$, $D_j$ measures one
observable consequence of the combined representation and conditional-flow
errors characterized by that proposition; the present experiment cannot
separate those two mechanisms.

Finite posterior samples introduce simulation variability into all four
components. Let $\widehat T_j$, $\widehat R_j$, $\widehat D_j$, and
$\widehat C_j$ denote their corrected estimators. For ABI, posterior-mean
variance is estimated as the posterior variance divided by 10,000; for MCMC,
it is estimated using the spectral-variance Monte Carlo standard error
(MCSE). These corrections preserve
$\widehat T_j=\widehat R_j+\widehat D_j+\widehat C_j$; derivations and MCMC
diagnostics appear in Supplementary Section~S3.3.

\begin{table}[t]
\centering
\small
\caption{Finite-draw-corrected decomposition of posterior-mean recovery error
over 100 prior-matched held-out datasets. The interaction
term can be negative.}
\label{tab:matched_mcmc_decomposition}
\begin{tabular}{lrrrr}
\hline
Parameter
& $10^{3}\widehat T_j$
& $10^{3}\widehat R_j$
& $10^{3}\widehat D_j$
& $10^{3}\widehat C_j$ \\
\hline
$\beta_0$    &  0.652 &  0.322 &  0.278 &   0.052 \\
$\sigma_w^2$ & 28.719 & 29.349 & 10.089 & -10.719 \\
$\eta$       & 21.452 & 13.684 &  8.681 &  -0.913 \\
$\rho$       & 17.392 & 10.907 &  9.926 &  -3.441 \\
\hline
\end{tabular}
\end{table}

Table~\ref{tab:matched_mcmc_decomposition} shows that
$\widehat D_j<\widehat R_j$ for all four parameters, although the two are of
similar magnitude for $\beta_0$ and $\rho$. Negative values of
$\widehat C_j$ for the spatial-structure parameters indicate that the
reference-posterior error and the ABI--MCMC posterior-mean difference tend to
have opposite signs. For example, the smaller $\widehat T_j$ for
$\sigma_w^2$ does not imply a negligible approximation discrepancy; rather,
the negative interaction partially offsets
$\widehat R_j+\widehat D_j$.

Across datasets, paired ABI and MCMC posterior means had correlations of
$1.000$, $0.959$, $0.835$, and $0.933$ for $\beta_0$, $\sigma_w^2$, $\eta$,
and $\rho$, respectively. For the primary edge-level target, within-dataset
correlations between ABI and MCMC boundary probabilities had mean $0.908$ and
median $0.947$, with mean edgewise absolute difference $0.071$. Both methods
achieved perfect within-dataset AUROC and average precision under the shared
scalar threshold mechanism. Matched MCMC had a lower mean Brier score
($0.040$ versus $0.062$) and moderately higher sensitivity and specificity
under the median-probability rule, whereas boundary-count interval coverage
was similar for the two methods. Complete truth-based boundary-recovery
results are reported in Supplementary Section~S3.3. Thus, the matched
benchmark identifies nonzero differences in posterior means and probability
accuracy while showing substantial agreement in the principal edge-level
boundary evidence.

\subsection{Supporting validation and computation}
\label{sec:additional_validation}

Parameter-posterior replication showed appropriate count coverage and no systematic discrepancy in residual spatial dependence or local edge contrasts. Representation ablations showed that restricting the network to core observations substantially degraded inference for $\eta$, $\rho$, and boundary probabilities, whereas removing individual non-core blocks had smaller effects. Training required 5 hours and 45 minutes; posterior generation required
$0.706$ s per validation dataset. In the matched experiment, ABI required
$0.772$ s per dataset compared with $2.707$ s for MCMC. Complete
replicated-data, ablation, MCMC, and computational results are in
Supplementary Sections~S3.3--S3.6.

\section{Spatial health disparity analyses}
\label{sec:real_data}

We now return to the three spatial health applications introduced in
Section~\ref{sec:applications}. The same trained amortized posterior
approximator was applied to all three datasets
without dataset-specific retraining, and posterior summaries were based on
100,000 draws from
$q_{\widehat{\bs{\phi}}}
(\bs{\theta}\mid\bs{z}_{\mathrm{obs}})$.
Our primary inferential targets are the boundary parameter $\eta$, residual
spatial dependence $\rho$, and edge-level posterior boundary probabilities.
We use the median-probability rule for descriptive boundary maps, classifying
an observed geographic edge as a boundary when its posterior boundary
probability exceeds $0.5$. Additional details on data provenance are provided in
Supplementary Section~S4.1, while fitted-risk and DAGAR-MCMC results are
reported in Supplementary Sections~S4.4 and~S4.6, respectively.

Fig.~\ref{fig:real_boundary_probabilities} presents the primary spatial
results from ABI-DAGAR. It displays the continuous posterior evidence that
borrowing is interrupted across each geographic border, before imposing a
decision threshold. The maps reveal
different spatial regimes: boundary evidence is widespread in Greater
Glasgow, concentrated on fewer county borders in California, and organized
into localized municipal clusters in South Korea.
Geographic reference maps and local atlases identifying selected
high-probability boundaries are provided in Supplementary
Section~S4.3, while the posterior boundary probability
and covariate dissimilarity relationship is examined in Supplementary Section~S4.5.

\begin{figure}[t]
\centering
\includegraphics[height=0.26\textwidth,keepaspectratio]{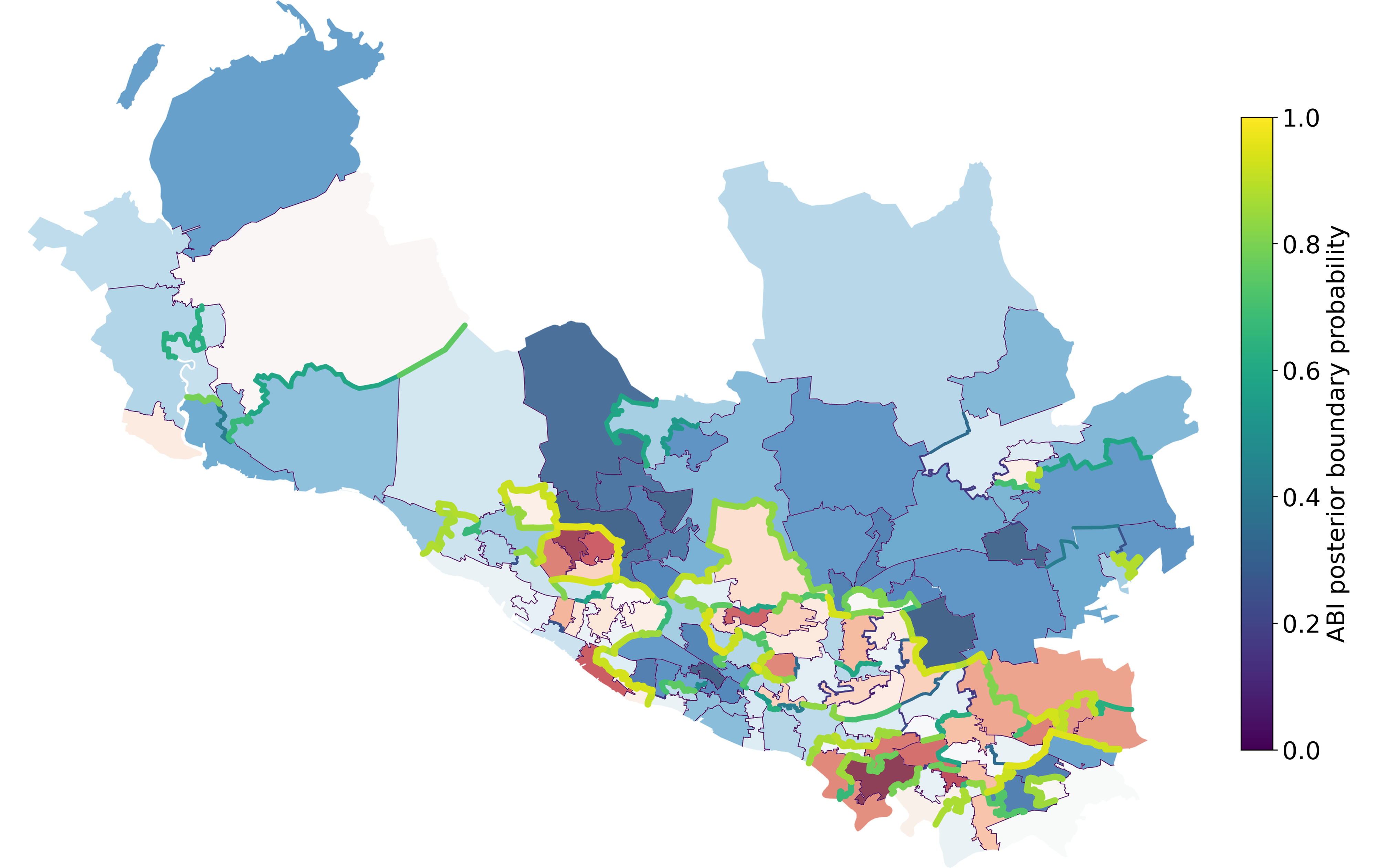}\hfill
\includegraphics[height=0.26\textwidth,keepaspectratio]{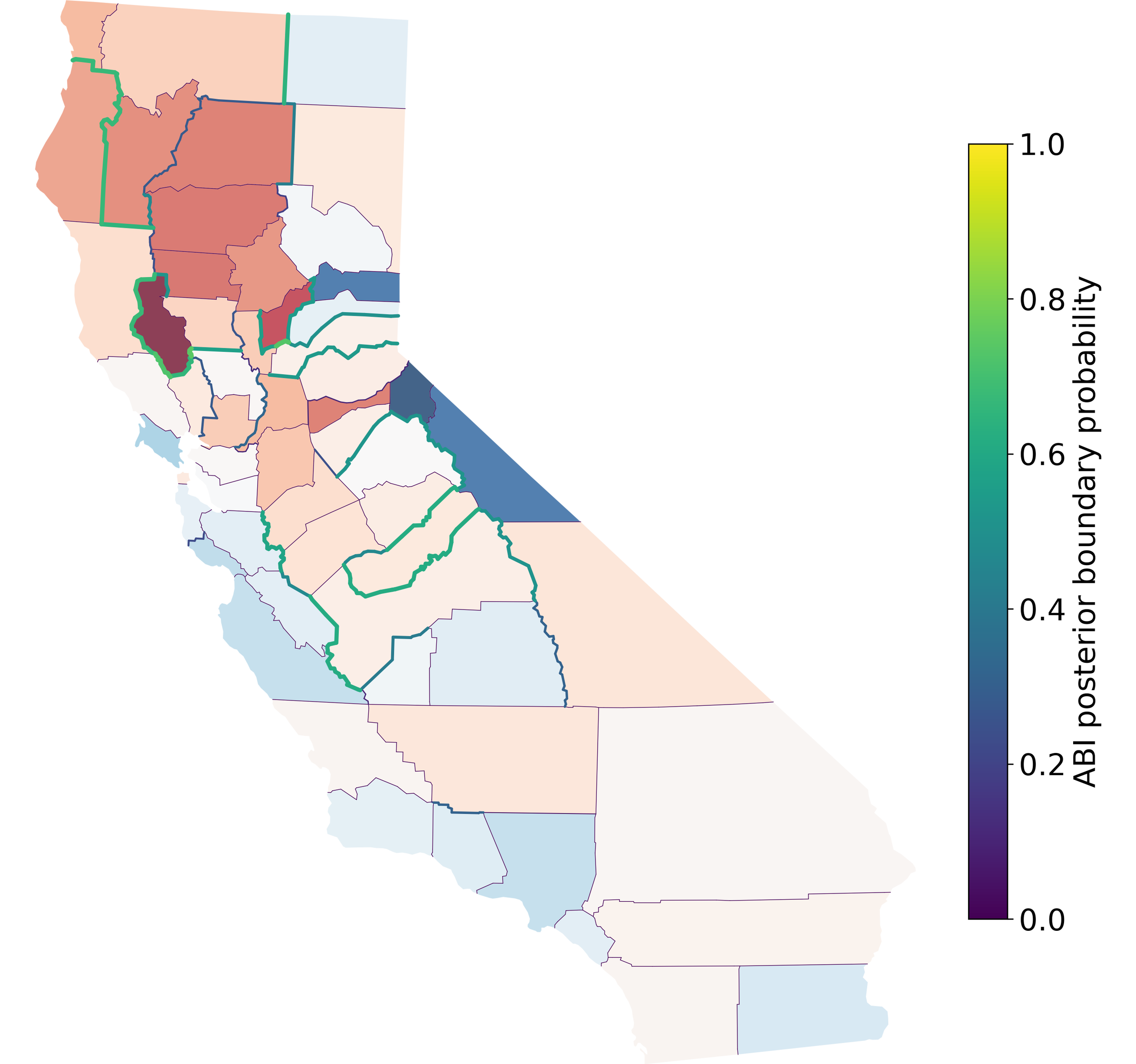}\hfill
\includegraphics[height=0.26\textwidth,keepaspectratio]{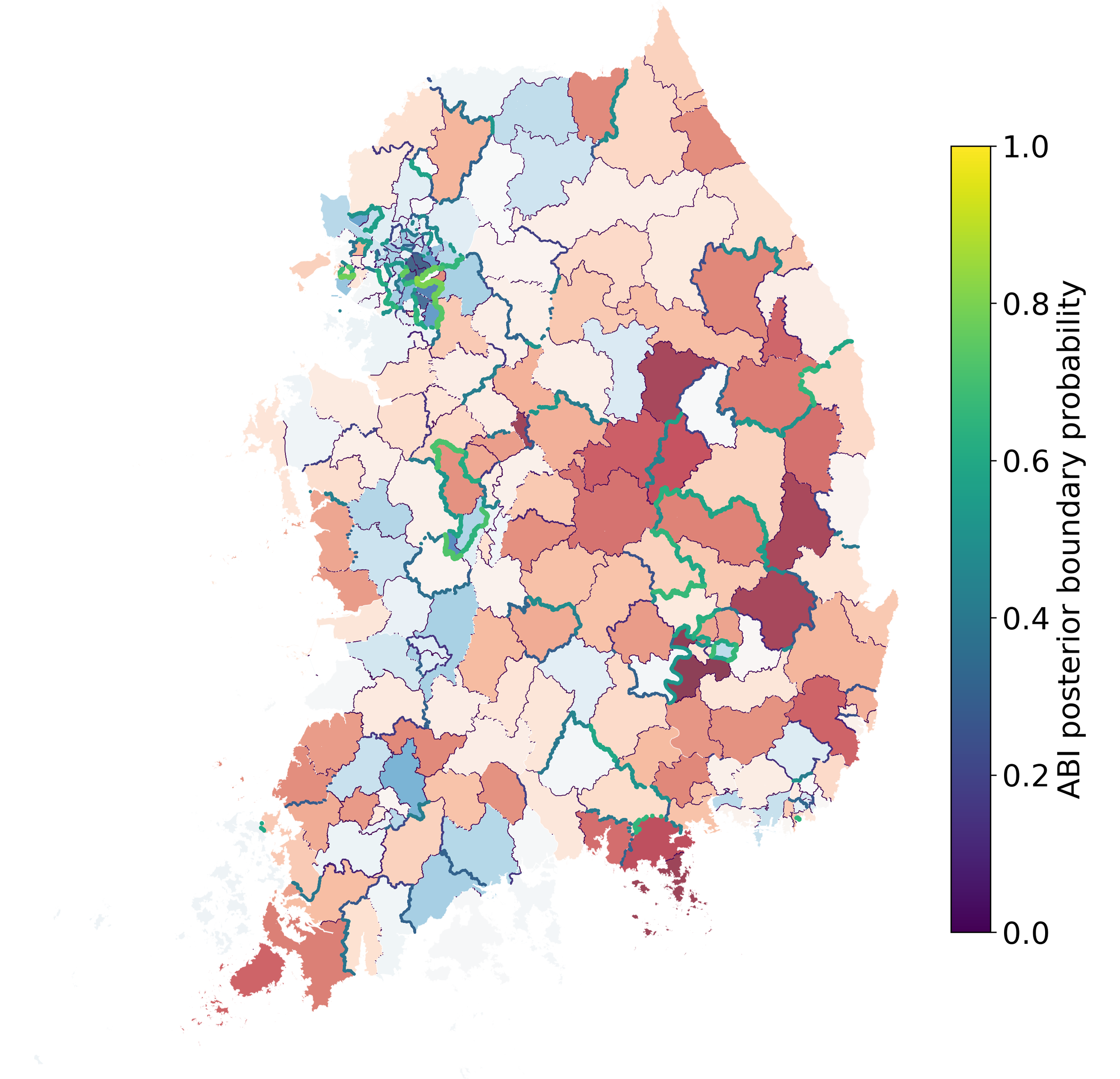}
\caption{ABI-DAGAR posterior boundary probabilities. Area shading gives
observed standardized disease risk, and edges encode boundary probability.
Left to right: Greater Glasgow, California, and South Korea.}
\label{fig:real_boundary_probabilities}
\end{figure}

\subsection{Respiratory disease in Greater Glasgow}
\label{sec:glasgow_results}

The posterior summaries in Table~\ref{tab:applications} indicate substantial
residual spatial dependence together with the most extensive boundary
configuration among the three applications. The left panel of 
Fig.~\ref{fig:real_boundary_probabilities} shows that this
evidence is distributed across Glasgow rather than confined to one cluster,
indicating broad regional persistence with substantial local interruptions in
borrowing.

Two high-probability boundaries provide concrete examples of this behavior.
The South Castlehill and Thorn--Drumchapel South border combines a
44.00 percentage-point difference in income deprivation with standardized
hospitalization ratios of 0.331 and 1.178, respectively. Similarly, the
Garrowhill East and Swinton--North Barlanark and Easterhouse South border has
a 43.50 percentage-point deprivation difference and standardized
hospitalization ratios of 0.554 and 1.387. Thus, both borders separate areas
that differ appreciably in the boundary-driving covariate and in observed
respiratory-disease risk.

\subsection{Lung cancer incidence in California}
\label{sec:california_results}

In California, appreciable residual spatial dependence coexists with a
comparatively sparse boundary pattern. The center panel of
Fig.~\ref{fig:real_boundary_probabilities} shows that the boundary evidence is
concentrated along a relatively small number of county borders, placing
California between a globally smooth spatial surface and the denser
localized-boundary regime observed in Greater Glasgow.

For example, the Lake--Yolo border has a 14.00 percentage-point difference in
smoking prevalence and standardized incidence ratios of 1.719 and 1.014,
respectively. The Placer--Yuba border has a 13.30 percentage-point
smoking-prevalence difference and standardized incidence ratios of 1.057 and
1.566. These two examples also show that the disease contrast need not follow
the smoking-prevalence contrast in a common direction.

\subsection{Tracheal, bronchial, and lung cancer mortality in South Korea}
\label{sec:korea_results}

The South Korean application combines a high posterior median for residual
spatial persistence with the lowest boundary density among the three
applications. Fig.~\ref{fig:real_boundary_probabilities} right panel 
shows that the boundary evidence is
organized in localized clusters rather than distributed uniformly
across the country, yielding a persistent national spatial structure with
selected local interruptions in borrowing.

The Sujeong-gu, Seongnam-si--Gwacheon-si boundary provides one such example.
The two administrative units differ in smoking prevalence by 11.16 percentage
points and have standardized mortality ratios of 0.991 and 0.688,
respectively. The Michuhol-gu--Yeonsu-gu boundary has an 8.66 percentage-point
smoking-prevalence difference and standardized mortality ratios of 1.071 and
0.842. These borders combine marked smoking-prevalence contrasts with visible
differences in observed mortality.

Across the three applications, residual dependence and boundary formation vary
independently: Glasgow combines strong persistence with widespread
interruptions, California exhibits a sparser and more uncertain pattern, and
South Korea combines strong persistence with localized boundaries on the
largest graph. Joint inference on $\eta$ and $\rho$ distinguishes these
contrasting spatial regimes and clarifies how broad spatial persistence can
coexist with markedly different degrees of local boundary formation.

\subsection{Boundary probabilities and local disease contrasts}
\label{sec:real_diagnostics}

Because the boundary mechanism is explicitly defined through covariate
dissimilarity, an increasing relationship between dissimilarity and posterior
boundary probability is expected under the fitted model. A more direct
empirical diagnostic asks whether edges receiving greater posterior boundary
probability also exhibit larger observed disease contrasts. We therefore 
compare, over all adjacent edges, $p_{ij}$ with
$|r_i-r_j|$, where $r_i$ is the crude log-risk of the $i$-th area.

Fig.~\ref{fig:real_residual_contrast} shows positive associations between posterior 
boundary probabilities and observed neighboring log-risk contrasts, with Spearman
correlations of $0.406$ in Greater Glasgow, $0.278$ in California, and $0.106$ in 
South Korea. The association is clearest in Glasgow and weak in South Korea. These
summaries assess correspondence with observed disease contrasts but do not
independently validate individual boundaries: high posterior boundary
probability need not imply a large observed contrast at a given border.

\begin{figure}[t]
\centering
\includegraphics[width=\textwidth]
{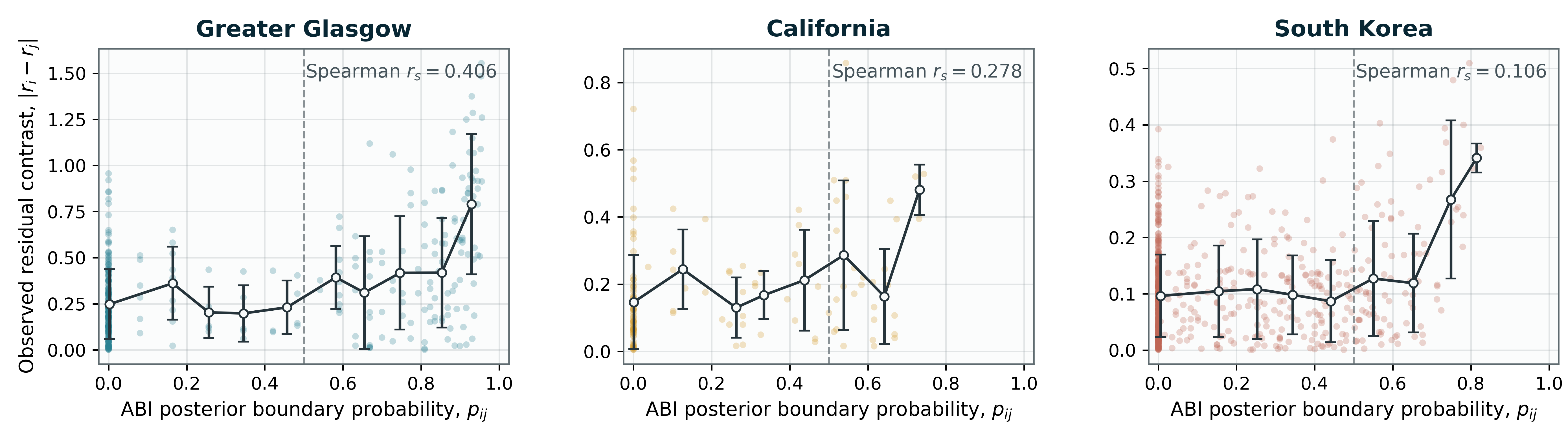}
\caption{Posterior boundary probabilities and neighboring crude log-risk
contrasts. Colored points represent individual adjacent edges. Open circles and
connecting lines show mean contrasts within posterior-probability bins, and
vertical capped bars show one standard deviation above and below each bin mean.
Dashed vertical lines mark the median-probability threshold $p_{ij}=0.5$. Left
to right: Greater Glasgow, California, and South Korea.}
\label{fig:real_residual_contrast}
\end{figure}

\subsection{Comparison with dataset-specific Bayesian inference}
\label{sec:bayesian_benchmarks}

We used two complementary dataset-specific Bayesian benchmarks to assess the
empirical boundary conclusions. The first is the localized CAR model
implemented by \texttt{S.CARdissimilarity} in the \texttt{CARBayes} package
\citep{CARBayes2013}, following the framework of
\citet{lee2012boundary}. This comparison evaluates whether ABI-DAGAR reaches
similar substantive conclusions to an established localized-smoothing model
with a different latent spatial prior. The second benchmark uses
dataset-specific MCMC for a closely aligned thresholded Poisson--DAGAR model.
This second comparison is closer to the ABI generative model than the CAR
benchmark, but it does not hold every modeling choice fixed because its scalar priors
differ from the simulator.

Both ABI-DAGAR and \texttt{CARBayes} use the observed geographic graph and
the same raw boundary-driving covariate, standardized separately in the Python
and R workflows, but their latent spatial priors and dependence
parameterizations differ. Exact agreement in scalar posterior parameters is
therefore neither expected nor required. The most direct cross-model comparison
instead concerns posterior boundary probabilities and selected-boundary
configurations, which capture their shared conclusions about local
discontinuities.

Table~\ref{tab:applications} reports the fitted spatial-structure parameters
and boundary densities and compares ABI-DAGAR with \texttt{CARBayes} through
posterior boundary probabilities and selected boundary sets. Let
$\mathcal B_A$ and $\mathcal B_C$ denote the respective selected sets.
Shared denotes $|\mathcal B_A\cap\mathcal B_C|$, and their Jaccard index is
$|\mathcal B_A\cap\mathcal B_C|/|\mathcal B_A\cup\mathcal B_C|$. Agreement
is substantial in all three applications, and both methods rank them
identically by boundary density. Although posterior-probability agreement is
weaker in South Korea, selected-set overlap remains comparable.

\begin{table}[t]
\centering
\caption{ABI-DAGAR spatial summaries and boundary agreement with
\texttt{CARBayes}. Parameters are posterior medians (95\% credible intervals);
boundary entries are counts (percentages of geographic adjacencies).
Prob. corr. is the Pearson correlation of edge-level posterior probabilities;
Shared and Jaccard summarize selected-set overlap.}
\label{tab:applications}
\resizebox{\textwidth}{!}{%
\begin{tabular}{lccccccc}
\hline
\textbf{Dataset}
& $\bs{\eta}$
& $\bs{\rho}$
& \textbf{ABI, $\#$ (\%)}
& \textbf{\texttt{CARBayes}, $\#$ (\%)}
& \textbf{Prob. corr.}
& \textbf{Shared}
& \textbf{Jaccard} \\
\hline
Glasgow
& $0.831\;(0.116,1.137)$
& $0.878\;(0.449,0.975)$
& $130\;(36.11\%)$
& $99\;(27.50\%)$
& 0.877
& 99
& 0.762 \\
California
& $0.469\;(0.017,0.964)$
& $0.815\;(0.051,0.990)$
& $31\;(22.30\%)$
& $24\;(17.27\%)$
& 0.866
& 24
& 0.774 \\
South Korea
& $0.433\;(0.017,0.961)$
& $0.893\;(0.143,0.994)$
& $89\;(14.22\%)$
& $69\;(11.02\%)$
& 0.721
& 66
& 0.717 \\
\hline
\end{tabular}%
}
\end{table}

Fig.~\ref{fig:boundary_agreement} addresses a different question from
Fig.~\ref{fig:real_boundary_probabilities}: whether the boundary decisions
obtained after applying the median-probability rule agree with those from
\texttt{CARBayes}. The additional ABI-DAGAR selections generally extend the
same local discontinuity structures identified by the benchmark rather than
forming geographically unrelated boundary systems.

\begin{figure}[t]
\centering
\includegraphics[height=0.26\textwidth,keepaspectratio]
{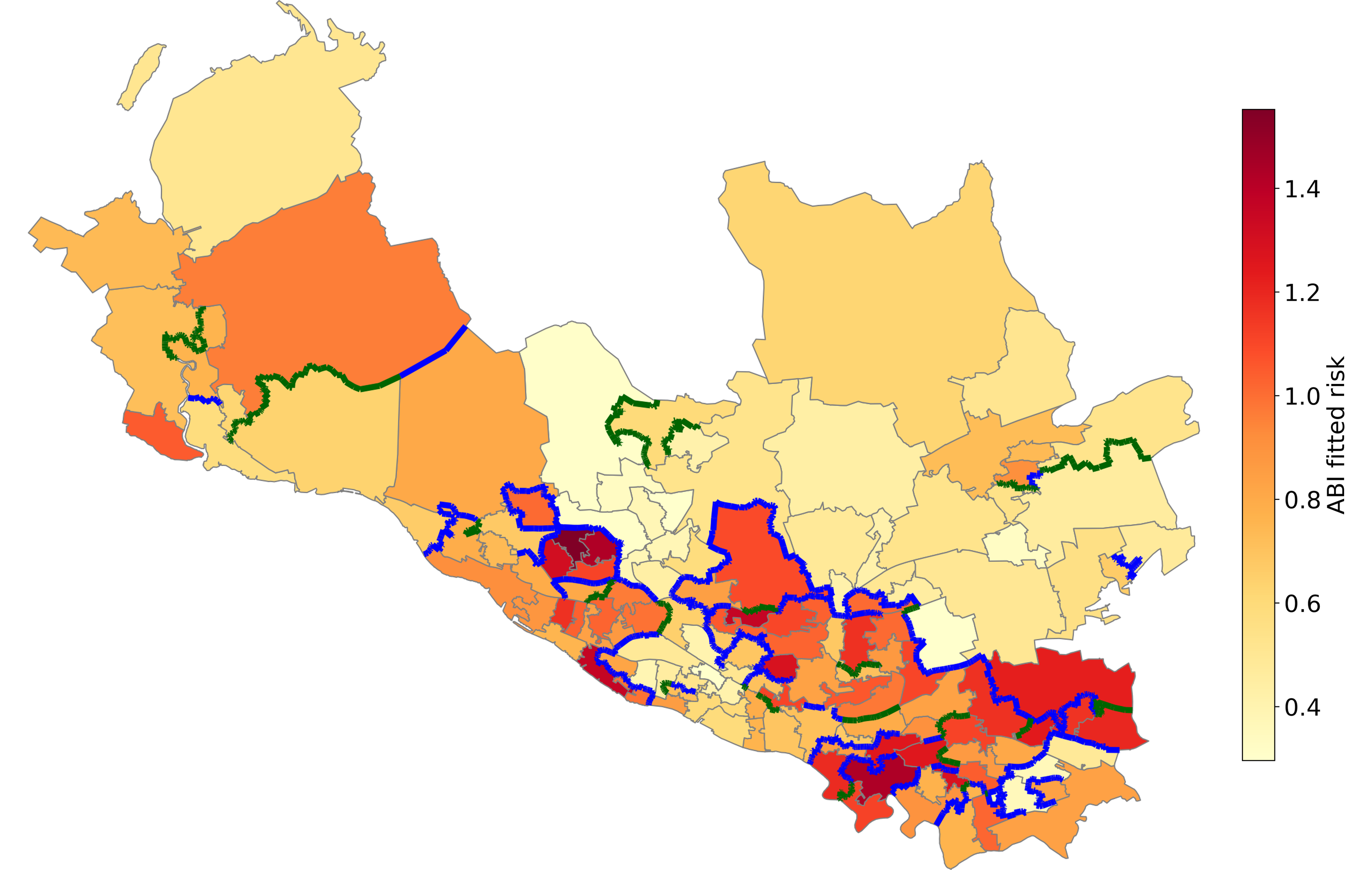}
\hfill
\includegraphics[height=0.26\textwidth,keepaspectratio]
{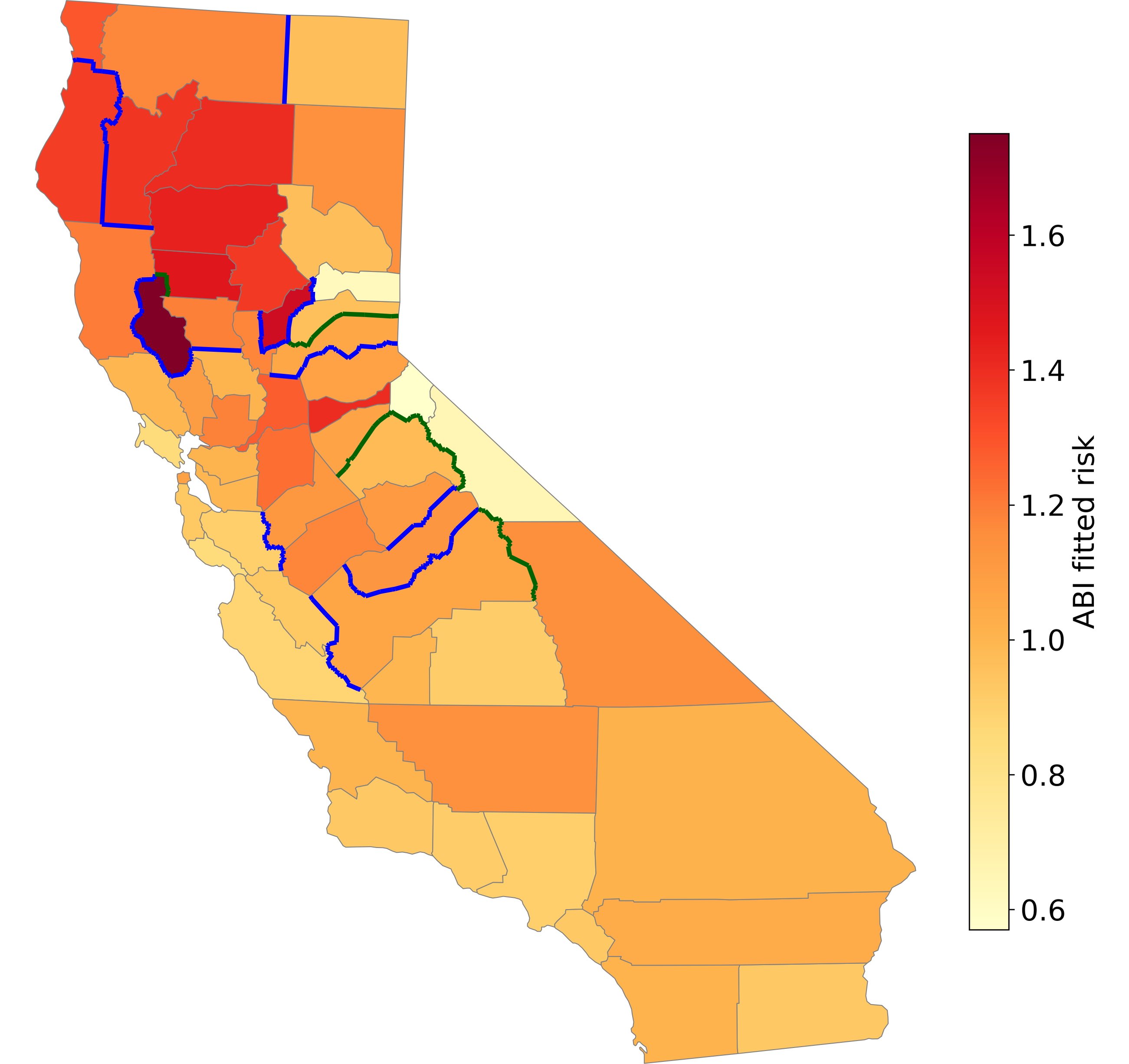}
\hfill
\includegraphics[height=0.26\textwidth,keepaspectratio]
{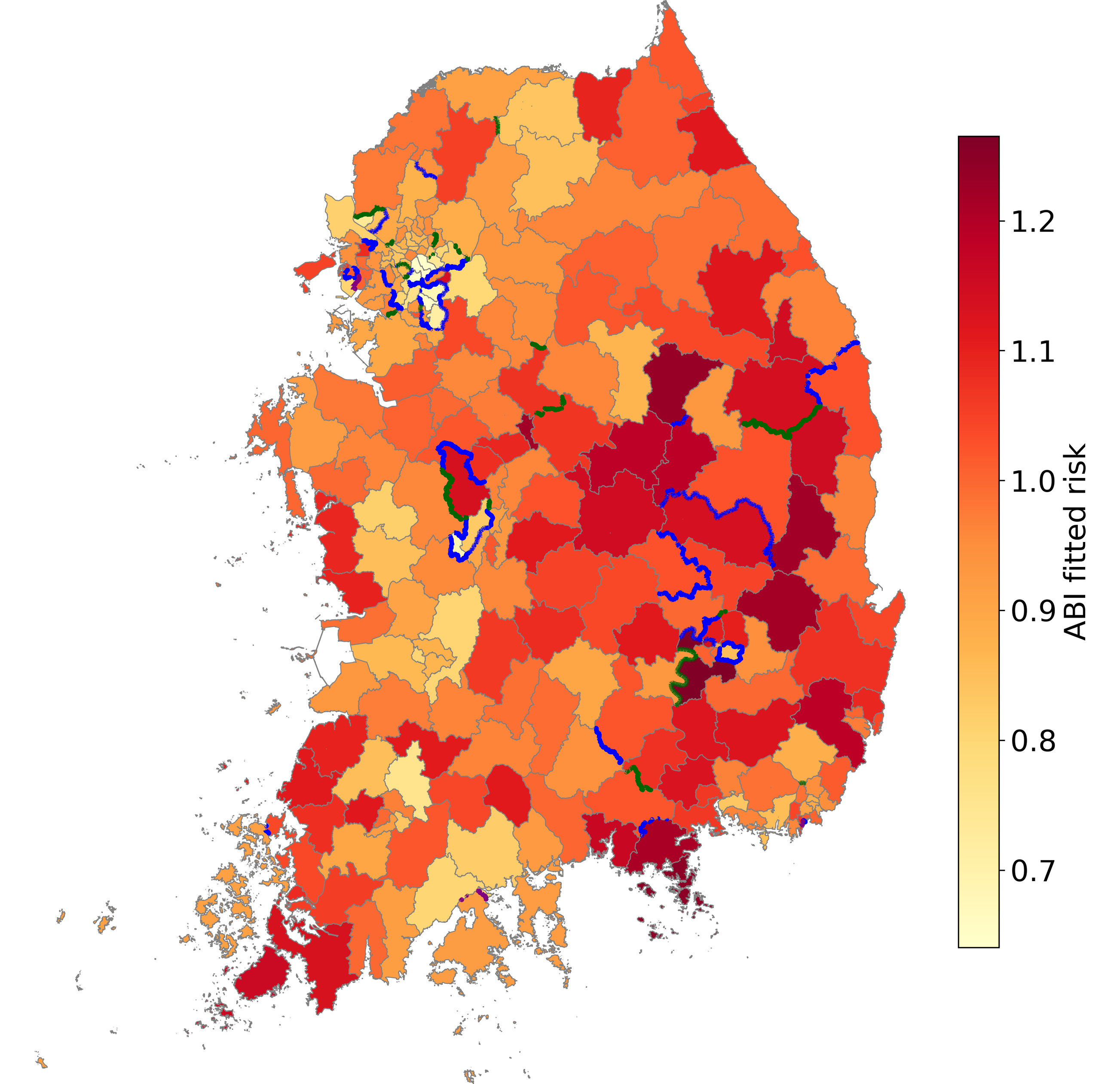}
\caption{Median-probability boundary agreement with \texttt{CARBayes}.
Blue denotes selection by both methods, green selection by ABI-DAGAR only 
and purple selection from \texttt{CARBayes} only.
Area shading gives the reconstructed ABI fitted relative risk.
Left to right: Greater Glasgow, California, and South Korea.}
\label{fig:boundary_agreement}
\end{figure}

The comparison with \texttt{CARBayes} necessarily combines two sources of
difference: the use of distinct latent spatial priors and the replacement of
dataset-specific MCMC by amortized inference. To examine these effects, we
also fitted a closely aligned thresholded Poisson--DAGAR model using
dataset-specific MCMC. Table~\ref{tab:real_mcmc} reports posterior summaries for the two
parameters governing spatial structure, $\eta$ and $\rho$.

\begin{table}[t]
\centering
\caption{ABI-DAGAR and DAGAR-MCMC posterior summaries for the
spatial-structure parameters.}
\label{tab:real_mcmc}
\begin{tabular}{llcc}
\hline
\textbf{Dataset}
& \textbf{Parameter}
& \textbf{ABI-DAGAR}
& \textbf{MCMC-DAGAR} \\
\hline
Glasgow
& $\eta$
& $0.831\;(0.116,1.137)$
& $0.871\;(0.702,1.083)$ \\
&
$\rho$
& $0.878\;(0.449,0.975)$
& $0.857\;(0.643,0.976)$ \\
\hline
California
& $\eta$
& $0.469\;(0.017,0.964)$
& $0.440\;(0.275,0.936)$ \\
&
$\rho$
& $0.815\;(0.051,0.990)$
& $0.578\;(0.264,0.898)$ \\
\hline
South Korea
& $\eta$
& $0.433\;(0.017,0.961)$
& $0.350\;(0.214,0.511)$ \\
&
$\rho$
& $0.893\;(0.143,0.994)$
& $0.634\;(0.438,0.861)$ \\
\hline
\end{tabular}
\end{table}

Posterior locations for $\eta$ are
comparatively close across methods, with especially strong agreement for both
spatial-structure parameters in Glasgow. The ABI-DAGAR posterior favors
stronger residual dependence in California and South Korea and provides wider
intervals for $\eta$ and $\rho$. Edge-level agreement is substantial, with boundary-probability Pearson
correlations of $0.959$, $0.922$, and $0.777$ in Glasgow,
California, and South Korea, respectively. The selected sets are nearly
identical in Glasgow; ABI-DAGAR identifies additional boundaries in California
and South Korea. Full posterior and boundary comparisons, together
with single-chain effective sample sizes and Monte Carlo standard errors for
both MCMC benchmarks, are reported in Supplementary Section~S4.6.

Taken together, the benchmark comparisons support the stability of the
principal edge-level conclusions across the three applications. Posterior
boundary probabilities, the primary inferential target, show substantial
cross-method agreement, while differences in scalar posterior spread are
fully documented. Boundary density does
not simply follow graph size or residual dependence: the largest graph has
the strongest estimated persistence but the most localized boundary system.
The same trained approximation therefore identifies distinct spatial regimes
while remaining substantially aligned with both dataset-specific benchmarks.
Finally, real-data runtime comparisons are reported in Supplementary Section~S4.7.

\section{Discussion}
\label{sec:discussion}

This study addresses a recurring problem in spatial disease mapping: how to
identify localized interruptions in spatial smoothing while retaining a model
for broader residual dependence, and how to carry out that inference across
areal graphs with different sizes and adjacency structures. The three health
applications illustrate why these components need to be distinguished. In
Greater Glasgow, the posterior supports extensive local boundary formation
together with high posterior median residual dependence, whereas the California
lung-cancer analysis exhibits a substantially sparser boundary pattern and
greater uncertainty in the spatial-structure parameters. South Korea adds a
larger municipal graph on which strong residual persistence coexists with
localized smoking-associated discontinuities. Thus, the same
statistical model accommodates qualitatively different spatial regimes rather
than imposing a common degree of smoothing across applications.

The interpretation rests on separating local boundary formation from residual
spatial dependence: $\eta$ controls covariate-guided interruptions in
borrowing, whereas $\rho$ controls persistence over the retained graph.
Because the boundary-driving covariate enters the adjacency mechanism rather
than the mean, the resulting probabilities describe where smoothing is
weakened rather than causal effects of deprivation or smoking. Empirically,
boundary density is greatest in Glasgow and lowest in South Korea despite
strong residual dependence in both, illustrating that local interruption and
global persistence need not vary together. Extensions allowing covariates in 
both components would address a broader target.

A second contribution is reusable posterior inference across variable-size graphs. 
Truth-based validation shows small bias and close-to-nominal interval coverage, 
while the prior-matched experiment separates reference-posterior 
recovery from the additional ABI posterior-mean discrepancy. Finite posterior-simulation
variability is accounted for in the corrected
decomposition, and its qualitative conclusions are overall stable. 
The observed ABI--MCMC posterior-mean discrepancy can
reflect both representation loss and conditional-density approximation,
which Proposition~\ref{prop:abi_decomposition} distinguishes theoretically
but the matched experiment does not separately identify. 
The value of amortization therefore rests on 
posterior quality, not sampling speed alone.

Related spatial work demonstrates the broader applicability of neural
simulation-based posterior inference. \citet{wang2026inference} use BayesFlow for
likelihood-free inference in stationary log-Gaussian Cox process models,
obtaining rapid repeated posterior analyses after training while emphasizing
comparison with conventional Bayesian computation and application-specific
validation; see also \citet{wikle2023statistical} for a broader account of
statistical deep learning for spatial and spatiotemporal data.

The empirical applications remain within the broad training class, i.e., Poisson
areal counts with offsets, sparse connected graphs, one boundary-driving
covariate, and graph sizes within the training range, but differ from typical
simulated configurations in topology and covariate smoothness. Their agreement
with dataset-specific Bayesian benchmarks provides evidence that the principal
boundary conclusions remain stable under these departures. Applications with 
different outcome models, graph structures, overdispersion,
or spatial dependence would require renewed validation and potentially
retraining.

Predictive stacking provides a complementary, non-neural direction for
scalable Bayesian inference. Stacking combines candidate predictive
distributions according to out-of-sample predictive performance
\citep{yao2017using}, with recent extensions to Gaussian geostatistics,
non-Gaussian spatiotemporal models, and transfer learning across large spatial
datasets
\citep{zhang2025jasa,panEtAl2025ba,presicce_bayesian_2024}.
For the present Poisson--DAGAR boundary model, future work could develop a
computationally tractable collection of conditional posterior and
posterior-predictive candidates indexed by spatial-dependence and boundary
parameters, and estimate stacking weights using held-out predictive scores.
More ambitiously, learning or transferring the candidate construction and
stacking rule across simulated datasets could provide a supervised alternative
to the BayesFlow posterior operator for reusable inference, as well as a useful
benchmark for neural posterior estimation. Such an extension is not immediate,
however: because edge-level boundary probabilities are a primary inferential
target, the candidate family and scoring rule would need to preserve posterior
uncertainty about the boundary and spatial-dependence parameters rather than
optimize outcome prediction alone.

Several extensions are natural. Multiple candidate drivers of local
discontinuity could be incorporated by allowing the boundary mechanism to
depend on several covariates rather than a single dissimilarity measure.
Richer graph representations, including graph-neural or hybrid encoders, could
be investigated when the current model-guided summaries are insufficient.
Extending the framework to multivariate disease mapping would also be useful
when the scientific objective is to distinguish shared from outcome-specific
spatial disparities.

A complementary direction concerns the component of the analysis that is
amortized. Our procedure learns a reusable mapping from observed areal data to
posterior distributions for the target parameters. Neural-prior methods such
as $\pi$VAE, PriorVAE, PriorCVAE, and DeepRV instead pretrain neural generators
to emulate draws from latent stochastic-process priors while retaining
dataset-specific Bayesian inference for the surrounding model
\citep{mishra2022pi,semenova2022priorvae,semenova2023priorcvae,navott2026deeprv}.
An analogous strategy could target high-dimensional latent spatial
distributions that contribute substantially to computation in more elaborate
Bayesian models. Examples include Gaussian-process random surfaces in spatial
Dirichlet process models
\citep{gelfand2005spatialdp,duan2007generalized}; CAR, DAGAR, and related
multivariate spatial constructions used in the continuing development of
Bayesian boundary-detection methods
\citep{li2015bayesian,gao2023spatial,aiello2023detecting,gianella2026bayesian};
and latent spatial structures coupled with spatially constrained random
partitions \citep{pavani2025bayesian,pavani2026modeling}. Pretraining validated
generators for these continuous latent components could reduce repeated latent-field computation, particularly in multivariate settings, while preserving the model-specific likelihood, boundary definition, and higher-level Bayesian hierarchy. Direct emulation of discrete partitions and adjacency structures, however, would require structured neural approximators beyond Gaussian-process surrogates. Future work should assess when combining posterior and prior amortization 
improves computation while preserving posterior calibration.

Overall, the results show that local health-disparity boundaries and residual
spatial dependence can be estimated jointly while posterior computation is
reused across heterogeneous disease maps. The main contribution of the
amortized component is therefore not simply faster sampling, but the ability to
carry an empirically evaluated Bayesian boundary-analysis workflow across maps with
different dimensions and adjacency structures while providing explicit posterior uncertainty quantification
for the quantities of primary scientific interest.

\section*{Acknowledgments}

The authors acknowledge the University of California, Los Angeles (UCLA),
Institute for Digital Research and Education for their infrastructure and
support in the computational workflow, experiments, and data analysis
presented in this manuscript.

\section*{Funding}

This work was supported by federal grants NSF/DMS 2113778,
NIH/NIEHS 1R01ES027027, and NIH/NIGMS R01GM148761.

\section*{Data and codes}

Data, Python and R code, the trained network, reproducibility scripts, and
reference outputs are available at
\url{https://github.com/lucaaiello/ABI_boundary_detection}.

\clearpage

\section*{\Large Supplementary Material}

\vspace{1em}

\setcounter{section}{0}
\setcounter{subsection}{0}
\setcounter{subsubsection}{0}
\setcounter{equation}{0}
\setcounter{figure}{0}
\setcounter{table}{0}

\renewcommand{\thesection}{S\arabic{section}}
\renewcommand{\thesubsection}{\thesection.\arabic{subsection}}
\renewcommand{\thesubsubsection}{\thesubsection.\arabic{subsubsection}}
\renewcommand{\theequation}{S\arabic{equation}}
\renewcommand{\thefigure}{S\arabic{figure}}
\renewcommand{\thetable}{S\arabic{table}}

Throughout this supplement, ABI denotes amortized Bayesian inference, DAGAR
denotes directed acyclic graph autoregressive, and MCMC denotes Markov chain
Monte Carlo. We use ABI-DAGAR for the resulting amortized DAGAR procedure and
MCMC-DAGAR for its dataset-specific MCMC counterpart.

\section{Additional details on the spatial boundary model}
\label{supp:model}

This section provides additional details for the spatial boundary model
introduced in Section~3 of the main manuscript. We first give the
complete DAGAR construction and then provide further details on the prior
specification and graph-specific scaling of the boundary parameter.

\subsection{DAGAR construction}
\label{supp:dagar}

Let $\bs{A}^{\ast}=(a_{ij}^{\ast})$ denote the modified adjacency matrix
induced by the covariate-driven boundary mechanism. Let
$\pi=(\pi_1,\ldots,\pi_N)$ denote a fixed ordering of the nodes and define
\begin{equation*}
\mathcal{N}_{\pi}(\pi_k)
=
\left\{
\pi_\ell:
\ell<k,\;
a^{\ast}_{\pi_k,\pi_\ell}=1
\right\}
\end{equation*}
to be the set of predecessors of node $\pi_k$ under this ordering. Let
$n_i=|\mathcal{N}_{\pi}(\pi_i)|$.

Conditional on $\bs{A}^{\ast}$, the DAGAR prior
\citep{datta2019spatial} is
\begin{equation}
\bs{w}
\mid
\sigma_w^2,\rho,\bs{A}^{\ast}
\sim
\mathcal{N}
\left(
\bs{0},
\sigma_w^2\bs{Q}(\rho;\bs{A}^{\ast})^{-1}
\right),
\label{supp:eq:dagar}
\end{equation}
where
\begin{equation}
\bs{Q}(\rho;\bs{A}^{\ast})
=
(\bs{I}-\bs{B})^\top
\bs{\Lambda}
(\bs{I}-\bs{B}).
\label{supp:eq:dagar_precision}
\end{equation}
Here $\bs{B}=(b_{ij})$ is strictly lower triangular and
$\bs{\Lambda}=\operatorname{diag}(\lambda_1,\ldots,\lambda_N)$, with
\begin{equation}
b_{ij}
=
\begin{cases}
\dfrac{\rho}
{1+(n_i-1)\rho^2},
&
\text{if } \pi_j\in\mathcal{N}_{\pi}(\pi_i),\\[8pt]
0,
&
\text{otherwise},
\end{cases}
\label{supp:eq:B}
\end{equation}
and
\begin{equation}
\lambda_i
=
\frac{1+(n_i-1)\rho^2}
{1-\rho^2}.
\label{supp:eq:Lambda}
\end{equation}
Equations~\eqref{supp:eq:dagar}--\eqref{supp:eq:Lambda} give the complete
DAGAR precision used in the analyses. Further discussion of the construction
is given by \citet{datta2019spatial} and \citet{aiello2023detecting}.

\subsection{Prior specification and graph-specific scaling}
\label{supp:priors}

The parameter vector is
\begin{equation*}
\bs{\theta}
=
(\beta_0,\sigma_w^2,\eta,\rho),
\end{equation*}
with priors
\begin{equation*}
\beta_0\sim\mathcal{N}(0,\sigma_\beta^2),
\qquad
\sigma_w^2\sim\left|\mathcal{N}(0,0.5)\right|,
\qquad
\eta\sim\mathcal{U}(0,M),
\qquad
\rho\sim\mathcal{U}(0,1).
\label{supp:eq:priors}
\end{equation*}
For neighboring areas, let
\begin{equation*}
z_{ij}=|x_i-x_j|,
\qquad
Z_{0.5}
=
\operatorname{median}
\{z_{ij}:a_{ij}=1\}.
\end{equation*}
The upper bound for the boundary parameter is
\begin{equation*}
M
=
\frac{\log 2}{Z_{0.5}}.
\label{supp:eq:M}
\end{equation*}
These are the prior specifications used throughout training and validation.
The graph-specific bound $M$ rescales $\eta$ to the neighboring-dissimilarity
scale and restricts removable edges to those with above-median dissimilarity,
as discussed in Section~3.1 of the main manuscript.

\FloatBarrier

\section{Additional details on amortized Bayesian inference}
\label{supp:abi}

This section provides the detailed observed-data representation, conditional
normalizing-flow formulation, training and deployment algorithm, and
implementation details for the amortized posterior approximation described in
Section~3.2 of the main manuscript.

\subsection{Observed-data representation}
\label{supp:representation}

The observed-data representation is constructed from the observed counts
$y_i$, offsets $e_i$, covariate values $x_i$, and the observed adjacency
matrix $\bs{A}=(a_{ij})$. We now give the complete construction of the 
observable node- and graph-level features summarized in Section~3.2 of the main manuscript.

Let $x_i$ denote the covariate used to define edge dissimilarity. Before
constructing the summaries, we standardize the covariate as
\begin{equation*}
\widetilde{x}_i
=
\frac{x_i-\bar{x}}{s_x},
\end{equation*}
where $\bar{x}$ and $s_x$ denote its sample mean and standard deviation. For
neighboring areas, define
\begin{equation*}
z_{ij}
=
|\widetilde{x}_i-\widetilde{x}_j|,
\qquad
\widetilde{z}_{ij}
=
\frac{z_{ij}}{Z_{0.5}},
\end{equation*}
where $Z_{0.5}$ is the median of $z_{ij}$ over observed neighboring pairs.

We define the offset-adjusted residual proxy as
\begin{equation*}
r_i
=
\log(y_i+0.5)-\log e_i.
\label{supp:eq:residual_proxy}
\end{equation*}
The baseline node-level summaries are
\begin{equation*}
\widetilde{x}_i,\quad
y_i,\quad
e_i,\quad
\log(1+y_i),\quad
\log e_i,\quad
r_i.
\end{equation*}

Let
\begin{equation*}
\mathcal{N}(i)
=
\{j:a_{ij}=1\},
\qquad
d_i
=
\sum_{j=1}^N a_{ij}.
\end{equation*}
For $d_i>0$, define
\begin{equation}
\bar{r}_i
=
\frac{1}{d_i}
\sum_{j\in\mathcal{N}(i)}r_j,
\qquad
\delta_i
=
\frac{1}{d_i}
\sum_{j\in\mathcal{N}(i)}
|r_i-r_j|.
\label{supp:eq:local_residual}
\end{equation}
We also compute
\begin{equation}
\bar{z}_i
=
\frac{1}{d_i}
\sum_{j\in\mathcal{N}(i)}
\widetilde{z}_{ij},
\qquad
z_i^{\max}
=
\max_{j\in\mathcal{N}(i)}
\widetilde{z}_{ij}.
\label{supp:eq:local_dissimilarity}
\end{equation}
When $d_i=0$, the neighborhood summaries in
\eqref{supp:eq:local_residual} and \eqref{supp:eq:local_dissimilarity} are set
to zero.

To capture the relationship between covariate dissimilarity and local residual
contrasts, we partition each neighborhood into
\begin{align*}
\mathcal{N}_L(i)
&=
\{j\in\mathcal{N}(i):
\widetilde{z}_{ij}\leq0.75\},\\
\mathcal{N}_M(i)
&=
\{j\in\mathcal{N}(i):
0.75<\widetilde{z}_{ij}\leq1.25\},\\
\mathcal{N}_H(i)
&=
\{j\in\mathcal{N}(i):
\widetilde{z}_{ij}>1.25\}.
\end{align*}
Let $d_i^{(B)}=|\mathcal{N}_B(i)|$ for
$B\in\{L,M,H\}$. For each bin,
\begin{equation}
\bar{r}_i^{(B)}
=
\frac{1}{d_i^{(B)}}
\sum_{j\in\mathcal{N}_B(i)}r_j,
\qquad
\delta_i^{(B)}
=
\frac{1}{d_i^{(B)}}
\sum_{j\in\mathcal{N}_B(i)}
|r_i-r_j|,
\qquad
p_i^{(B)}
=
\frac{d_i^{(B)}}{d_i}.
\label{supp:eq:binned_local}
\end{equation}
The summaries in \eqref{supp:eq:binned_local} are set to zero whenever the
relevant denominator is zero.

Let
\begin{equation*}
\mathcal{E}
=
\{(i,j):i<j,\;a_{ij}=1\}
\end{equation*}
denote the observed edge set. For each edge, define
\begin{equation*}
\Delta_{ij}=|r_i-r_j|,
\qquad
C_{ij}
=
(r_i-\bar{r}_{\cdot})
(r_j-\bar{r}_{\cdot}),
\qquad
\bar{r}_{\cdot}
=
\frac{1}{N}\sum_{i=1}^N r_i.
\end{equation*}
Partition $\mathcal{E}$ into
$\mathcal{E}_L,\mathcal{E}_M,\mathcal{E}_H$ using the same thresholds on
$\widetilde{z}_{ij}$. Let $\Delta_B$ and $C_B$ denote the corresponding
within-bin means. We define
\begin{equation}
G_{\Delta}
=
\Delta_H-\Delta_L,
\qquad
G_C
=
C_L-C_H.
\label{supp:eq:contrast_gradients}
\end{equation}
The contrast gradients in \eqref{supp:eq:contrast_gradients} compare the high-
and low-dissimilarity edge bins.
We also compute
\begin{equation}
B_{\Delta}
=
\frac{
\sum_{(i,j)\in\mathcal{E}}
(\widetilde{z}_{ij}-\bar{z}_{\mathcal{E}})
(\Delta_{ij}-\bar{\Delta}_{\mathcal{E}})
}{
\sum_{(i,j)\in\mathcal{E}}
(\widetilde{z}_{ij}-\bar{z}_{\mathcal{E}})^2
},
\label{supp:eq:contrast_slope}
\end{equation}
where
\begin{equation*}
\bar{z}_{\mathcal{E}}
=
|\mathcal{E}|^{-1}
\sum_{(i,j)\in\mathcal{E}}\widetilde{z}_{ij},
\qquad
\bar{\Delta}_{\mathcal{E}}
=
|\mathcal{E}|^{-1}
\sum_{(i,j)\in\mathcal{E}}\Delta_{ij}.
\end{equation*}
If the denominator in \eqref{supp:eq:contrast_slope} is zero, we set
$B_{\Delta}=0$.

Let
\begin{equation}
s_r^2
=
\frac{1}{N}
\sum_{i=1}^N
(r_i-\bar{r}_{\cdot})^2.
\label{supp:eq:residual_variance}
\end{equation}
When $s_r^2=0$, as defined in \eqref{supp:eq:residual_variance}, the
standardized autocorrelation summaries below are set to zero. At the node level
we include
\begin{equation}
I_i
=
\frac{
(r_i-\bar{r}_{\cdot})\bar{r}_i
}{
s_r^2
},
\qquad
V_i
=
\frac{
(r_i-\bar{r}_i)^2
}{
s_r^2
}.
\label{supp:eq:local_spatial}
\end{equation}
In \eqref{supp:eq:local_spatial}, $I_i$ is a local Moran-type measure of
agreement between node $i$ and its neighborhood, whereas $V_i$ measures local
semivariance.

At the graph level, we compute
\begin{equation}
C_{\mathrm{edge}}
=
\frac{1}
{|\mathcal{E}|s_r^2}
\sum_{(i,j)\in\mathcal{E}}C_{ij},
\qquad
C_{\mathrm{lag}}
=
\operatorname{corr}
\{r_i,\bar{r}_i:i=1,\ldots,N\},
\label{supp:eq:graph_corr}
\end{equation}
and
\begin{equation}
B_{\mathrm{lag}}
=
\frac{
N^{-1}
\sum_{i=1}^N
(r_i-\bar{r}_{\cdot})
(\bar{r}_i-\bar{r}_N)
}{
s_r^2
},
\qquad
\bar{r}_N
=
\frac{1}{N}
\sum_{i=1}^N
\bar{r}_i.
\label{supp:eq:lag_slope}
\end{equation}
Finally,
\begin{equation}
\Gamma_{\mathrm{edge}}
=
\frac{1}{2|\mathcal{E}|}
\sum_{(i,j)\in\mathcal{E}}
(r_i-r_j)^2.
\label{supp:eq:edge_semivariance}
\end{equation}
The graph-level summaries in \eqref{supp:eq:graph_corr}--\eqref{supp:eq:edge_semivariance} quantify
residual association and variation over the observed graph.

The graph-level quantities are replicated across nodes and appended to each
node-specific feature vector. The resulting representation is
\begin{align*}
\bs{s}_i
=
\Big(
&\widetilde{x}_i,\,
y_i,\,
e_i,\,
\log(1+y_i),\,
\log e_i,\,
r_i,\,
d_i,\,
\bar{r}_i,\,
\delta_i,\,
\bar{r}_i^{(L)},\,
\bar{r}_i^{(M)},\,
\bar{r}_i^{(H)},\,
\delta_i^{(L)},\,
\delta_i^{(M)},\,
\delta_i^{(H)},\,
\nonumber\\
&p_i^{(L)},\,
p_i^{(M)},\,
p_i^{(H)},\,
\bar{z}_i,\,
z_i^{\max},\,
I_i,\,
V_i,\,
M,\,
B_{\Delta},\,
G_{\Delta},\,
G_C,\,
C_{\mathrm{edge}},\,
C_{\mathrm{lag}},\,
B_{\mathrm{lag}},\,
\Gamma_{\mathrm{edge}}
\Big).
\label{supp:eq:full_summary}
\end{align*}

\subsection{Conditional normalizing flow and training objective}
\label{supp:flow}

This subsection provides the formal arguments underlying two properties used
in Section~3.2 of the main manuscript. First, we verify that the conditional
normalizing flow defines a proper probability density for every fixed learned
representation. Second, we relate the training objective to the
full-data posterior discrepancy and prove
Proposition~3.1, which separates representation loss
from conditional-density approximation error. We then state the resulting
condition for exact posterior recovery. Throughout,
$\bs{z}=h_{\bs{\psi}}(\mathcal{S}(\mathcal{D}))$.
The marginal $p(\bs{z})$ and conditional distribution
$p(\bs{\theta}\mid\bs{z})$ are those induced by the joint prior-predictive
distribution through the representation
$\bs{z}=h_{\bs{\psi}}(\mathcal S(\mathcal D))$; their dependence on
$\bs{\psi}$ is suppressed for notational simplicity.

For fixed $\bs{z}$,
\begin{equation*}
q_{\bs{\phi}}(\bs{\theta}\mid\bs{z})
=
p_{\bs{u}}
(f_{\bs{\phi}}(\bs{\theta};\bs{z}))
\left|
\det
\frac{\partial f_{\bs{\phi}}(\bs{\theta};\bs{z})}
{\partial\bs{\theta}}
\right|,
\qquad
\bs{u}\sim\mathcal N(\bs 0,\bs I).
\end{equation*}

For any fixed $\bs{z}$, invertibility of
$f_{\bs{\phi}}(\cdot;\bs{z})$ implies
\begin{align*}
\int_{\Theta}
q_{\bs{\phi}}(\bs{\theta}\mid\bs{z})
\,d\bs{\theta}
=
\int_{\Theta}
p_{\bs{u}}
(f_{\bs{\phi}}(\bs{\theta};\bs{z}))
\left|
\det
\frac{\partial f_{\bs{\phi}}(\bs{\theta};\bs{z})}
{\partial\bs{\theta}}
\right|
d\bs{\theta}
=
\int_{\mathbb R^d}
p_{\bs{u}}(\bs{u})\,d\bs{u}
=
1.
\end{align*}
Hence $q_{\bs{\phi}}(\cdot\mid\bs{z})$ is a proper conditional
density for every fixed representation $\bs{z}$. The summary network affects
the conditioning information but not the normalization of the flow density.

We next connect this conditional density to the objective used for
joint training. The negative conditional log-density objective in Section~3.2
can be rewritten as
\begin{align}
\mathcal L(\bs{\phi},\bs{\psi})
&=
H(\bs{\theta}\mid\mathcal D)
+
\mathbb E_{p(\mathcal D)}
\left[
\operatorname{KL}
\left\{
p(\bs{\theta}\mid\mathcal D)
\,\Vert\,
q_{\bs{\phi}}(\bs{\theta}\mid\bs{z})
\right\}
\right].
\label{supp:eq:full_posterior_loss}
\end{align}
where
\begin{equation*}
    H(\bs{\theta}\mid\mathcal D)
=
-
\mathbb E_{p(\bs{\theta},\mathcal D)}
\left[
\log p(\bs{\theta}\mid\mathcal D)
\right]
\end{equation*}
is the conditional entropy. Since this term does not depend on
$(\bs{\phi},\bs{\psi})$, minimizing the training loss is equivalent to
minimizing the expected posterior discrepancy in
\eqref{supp:eq:full_posterior_loss}. Proposition~3.1 decomposes this 
discrepancy into information lost through the
representation and conditional-density approximation error. We now
give its proof.


\begin{proof}[Proof of Proposition~3.1]
By the definition of KL divergence,
\begin{align*}
&\mathbb E_{p(\mathcal D)}
\left[
\operatorname{KL}
\left\{
p(\bs{\theta}\mid\mathcal D)
\,\Vert\,
q_{\bs{\phi}}(\bs{\theta}\mid\bs{z})
\right\}
\right] =
\mathbb E_{p(\bs{\theta},\mathcal D)}
\left[
\log
\frac{p(\bs{\theta}\mid\mathcal D)}
     {q_{\bs{\phi}}(\bs{\theta}\mid\bs{z})}
\right].
\end{align*}
Inserting $p(\bs{\theta}\mid\bs{z})$ into the density ratio gives
\begin{align*}
\log
\frac{p(\bs{\theta}\mid\mathcal D)}
     {q_{\bs{\phi}}(\bs{\theta}\mid\bs{z})}
&=
\log
\frac{p(\bs{\theta}\mid\mathcal D)}
     {p(\bs{\theta}\mid\bs{z})}
+
\log
\frac{p(\bs{\theta}\mid\bs{z})}
     {q_{\bs{\phi}}(\bs{\theta}\mid\bs{z})}.
\end{align*}
Consequently,
\begin{align*}
&\mathbb E_{p(\mathcal D)}
\left[
\operatorname{KL}
\left\{
p(\bs{\theta}\mid\mathcal D)
\,\Vert\,
q_{\bs{\phi}}(\bs{\theta}\mid\bs{z})
\right\}
\right] =
\mathbb E_{p(\bs{\theta},\mathcal D)}
\left[
\log
\frac{p(\bs{\theta}\mid\mathcal D)}
     {p(\bs{\theta}\mid\bs{z})}
\right]
+
\mathbb E_{p(\bs{\theta},\mathcal D)}
\left[
\log
\frac{p(\bs{\theta}\mid\bs{z})}
     {q_{\bs{\phi}}(\bs{\theta}\mid\bs{z})}
\right].
\end{align*}

To identify the first term, begin with the standard definition of conditional
mutual information:
\begin{align*}
I(\bs{\theta};\mathcal D\mid\bs{z})
&=
\mathbb E_{p(\mathcal D,\bs{z})}
\left[
\operatorname{KL}
\left\{
p(\bs{\theta}\mid\mathcal D,\bs{z})
\,\Vert\,
p(\bs{\theta}\mid\bs{z})
\right\}
\right] \\
&= 
\mathbb E_{p(\mathcal D,\bs{z})}
\left[
\mathbb E_{p(\bs{\theta}\mid\mathcal D,\bs{z})}
\left\{
\log
\frac{p(\bs{\theta}\mid\mathcal D,\bs{z})}
     {p(\bs{\theta}\mid\bs{z})}
\right\}
\right] =
\mathbb E_{p(\bs{\theta},\mathcal D,\bs{z})}
\left[
\log
\frac{p(\bs{\theta}\mid\mathcal D,\bs{z})}
     {p(\bs{\theta}\mid\bs{z})}
\right].
\end{align*}
Because $\bs{z}=h_{\bs{\psi}}(\mathcal S(\mathcal D))$ is a deterministic function
of $\mathcal D$,
\[
p(\bs{\theta}\mid\mathcal D,\bs{z})
=
p(\bs{\theta}\mid\mathcal D)
\qquad\text{almost surely}.
\]
The joint distribution of $(\bs{\theta},\mathcal D,\bs{z})$ is therefore
induced by $p(\bs{\theta},\mathcal D)$ together with
$\bs{z}=h_{\bs{\psi}}(\mathcal S(\mathcal D))$. It follows that
\begin{align*}
I(\bs{\theta};\mathcal D\mid\bs{z})
&=
\mathbb E_{p(\bs{\theta},\mathcal D)}
\left[
\log
\frac{p(\bs{\theta}\mid\mathcal D)}
     {p(\bs{\theta}\mid\bs{z})}
\right].
\end{align*}
Thus, the first term in the decomposition is
$I(\bs{\theta};\mathcal D\mid\bs{z})$.

For the second term, marginalizing $\mathcal D$ under the induced joint
distribution gives
$p(\bs{\theta},\bs{z})=p(\bs{z})p(\bs{\theta}\mid\bs{z})$. Therefore,
the law of iterated expectation gives
\begin{align*}
\mathbb E_{p(\bs{\theta},\mathcal D)}
\left[
\log
\frac{p(\bs{\theta}\mid\bs{z})}
     {q_{\bs{\phi}}(\bs{\theta}\mid\bs{z})}
\right] &=
\mathbb E_{p(\bs{z})}
\left[
\mathbb E_{p(\bs{\theta}\mid\bs{z})}
\left\{
\log
\frac{p(\bs{\theta}\mid\bs{z})}
     {q_{\bs{\phi}}(\bs{\theta}\mid\bs{z})}
\right\}
\right]
\\
& =
\mathbb E_{p(\bs{z})}
\left[
\operatorname{KL}
\left\{
p(\bs{\theta}\mid\bs{z})
\,\Vert\,
q_{\bs{\phi}}(\bs{\theta}\mid\bs{z})
\right\}
\right].
\end{align*}
Substitution of these two identities proves the result.
\end{proof}

The decomposition also makes explicit when the amortized approximation can
recover the full-data posterior exactly. The first term vanishes when the
learned representation is posterior-sufficient, while the second vanishes
when the conditional flow exactly represents the posterior given that
representation. This yields the following immediate consequence.

\begin{corollary}[Exact posterior recovery]
Under the conditions of Proposition~3.1, suppose
that
\begin{equation*}
\bs{\theta}\perp\!\!\!\perp\mathcal D\mid\bs{z}
\end{equation*}
and that the conditional flow satisfies
\begin{equation*}
q_{\bs{\phi}}(\bs{\theta}\mid\bs{z})
=
p(\bs{\theta}\mid\bs{z})
\end{equation*}
almost surely. Then
\begin{equation*}
q_{\bs{\phi}}(\bs{\theta}\mid\bs{z})
=
p(\bs{\theta}\mid\mathcal D)
\end{equation*}
almost surely.
\end{corollary}

\begin{proof}
Posterior sufficiency gives
$p(\bs{\theta}\mid\bs{z})=p(\bs{\theta}\mid\mathcal D)$,
and the result follows from the assumed exact conditional-density
approximation.
\end{proof}

This corollary clarifies the assumptions underlying the exact-recovery result
of \citet{radev2020bayesflow}. Their argument assumes perfect convergence of
the summary and invertible networks together with the existence of a
sufficient finite-dimensional summary. In practice, they separately identify
information loss through the summary network and imperfect approximation by
the invertible network as sources of error. We likewise do not treat
$h_{\bs{\psi}}(\mathcal S(\mathcal D))$ as provably sufficient; the validation
analyses in Sections~S3--S4 instead assess whether the combined representation
and conditional-density approximation is adequate over the deployment regime.

\subsection{Training and deployment algorithm}
\label{supp:algorithm}

Algorithm~\ref{supp:alg:abi} summarizes the training and deployment workflow.

\begin{algorithm}[t]
\caption{Amortized Bayesian inference for the DAGAR boundary-detection model}
\label{supp:alg:abi}
\begin{algorithmic}[1]

\Require Simulator for $(\bs{\theta},\mathcal{D})$,
feature map $\mathcal{S}(\cdot)$,
summary network $h_{\bs{\psi}}$,
conditional flow $f_{\bs{\phi}}$,
batch size $B$,
number of posterior draws $L$

\Statex
\State \textbf{Training phase}

\Repeat
    \For{$b=1,\ldots,B$}
        \State Simulate
        $(\bs{\theta}^{(b)},\mathcal{D}^{(b)})
        \sim p(\bs{\theta},\mathcal{D})$
        \State Construct
        $\mathcal{S}^{(b)}
        =
        \mathcal{S}(\mathcal{D}^{(b)})$
        \State Compute
        $\bs{z}^{(b)}
        =
        h_{\bs{\psi}}(\mathcal{S}^{(b)})$
        \State Evaluate
        $q_{\bs{\phi}}
        (\bs{\theta}^{(b)}\mid\bs{z}^{(b)})$
    \EndFor
    \State Update $(\bs{\phi},\bs{\psi})$ by minimizing the minibatch
    negative log-density loss
\Until{100 training epochs have been completed}

\Statex
\State \textbf{Deployment phase for an observed dataset
$\mathcal{D}_{\mathrm{obs}}$}

\State Construct
$\mathcal{S}_{\mathrm{obs}}
=
\mathcal{S}(\mathcal{D}_{\mathrm{obs}})$

\State Compute
$\bs{z}_{\mathrm{obs}}
=
h_{\widehat{\bs{\psi}}}(\mathcal{S}_{\mathrm{obs}})$

\For{$\ell=1,\ldots,L$}
    \State Draw
    $\bs{u}^{(\ell)}
    \sim
    \mathcal{N}(\bs{0},\bs{I})$
    \State Set
    $\bs{\theta}^{(\ell)}
    =
    f_{\widehat{\bs{\phi}}}^{-1}
    (\bs{u}^{(\ell)};\bs{z}_{\mathrm{obs}})$
\EndFor

\State Return
$\{\bs{\theta}^{(\ell)}\}_{\ell=1}^{L}$
as approximate posterior draws from
$q_{\widehat{\bs{\phi}}}
(\bs{\theta}\mid\bs{z}_{\mathrm{obs}})$

\end{algorithmic}
\end{algorithm}

\subsection{Neural-network implementation}
\label{supp:implementation}

The amortized posterior approximation was implemented in BayesFlow using
Keras with the TensorFlow backend. NumPy randomness was seeded, whereas
TensorFlow and Keras randomness was not explicitly seeded.

The flow target is
$(\beta_0,\sigma_w^2,\eta_{\mathrm{raw}},\rho)$, with
$\eta=M\eta_{\mathrm{raw}}$. The graph-specific quantity $M$ is a
deterministic function of the observed covariates and adjacency structure and
is also included among the graph-level components of
$\mathcal S(\mathcal D)$. In the implementation, however, the conversion
between $\eta_{\mathrm{raw}}$ and the original-scale parameter $\eta$ uses the
observed value of $M$ directly. Thus, on the original parameter scale the
induced transformation is dataset-specific through $M=M(\mathcal D)$; this
known deterministic dependence is suppressed in the notation
$f_{\bs{\phi}}(\bs{\theta};\bs{z})$ used in the main manuscript.

More specifically, the complete original-scale transformation represented by
$f_{\bs{\phi}}(\bs{\theta};\bs{z})$ comprises the deterministic rescaling
$\eta_{\mathrm{raw}}=\eta/M$, the one-to-one support transformations and
standardization applied before inference, and the learned coupling-flow
transformation. For each fixed dataset with $M>0$, these transformations are
one-to-one, so their composition with the invertible coupling flow remains
bijective. Consequently, the change-of-variables density in Section~3.2 is
understood on the original parameter scale with the Jacobian of the complete
composite map.

Ignoring the subsequent affine standardization, the support transformations
used before the coupling flow are
\begin{equation*}
\sigma_w^2
\mapsto
\operatorname{softplus}^{-1}(\sigma_w^2),
\qquad
\eta
\mapsto
\operatorname{logit}\!\left(\frac{\eta}{M}\right),
\qquad
\rho
\mapsto
\operatorname{logit}(\rho),
\end{equation*}
while $\beta_0$ is unconstrained. Their inverses enforce
$\sigma_w^2>0$, $\eta\in(0,M)$, and $\rho\in(0,1)$. BayesFlow
standardization was used for both summary and inference variables.

Training used online simulation with the default optimizer and no
user-defined callbacks, early stopping, or custom learning-rate schedule.
Table~\ref{supp:tab:implementation} summarizes the software environment,
network architecture, training configuration, hardware, and timing.

\begin{table}[t]
\centering
\caption{Neural-network and training configuration for ABI-DAGAR.}
\label{supp:tab:implementation}
\begin{tabular}{ll}
\hline
\textbf{Component} & \textbf{Specification} \\
\hline
Software stack
& Python 3.10.19; NumPy 2.2.6; SciPy 1.15.3; \\
& pandas 2.3.3; GeoPandas 1.1.3; BayesFlow 2.0.8; \\
& Keras 3.12.1; TensorFlow 2.21.0 \\
Backend
& \texttt{KERAS\_BACKEND=tensorflow} \\
Random seed
& \texttt{np.random.seed(123)} \\
Summary network
& BayesFlow \texttt{SetTransformer} \\
Summary output dimension
& 32 \\
Inference network
& BayesFlow \texttt{CouplingFlow} \\
Transform type
& Spline \\
Training mode
& Online simulation \\
Epochs
& 100 \\
Batch size
& 64 \\
Batches per epoch
& 200 \\
Total simulated datasets
& $1{,}280{,}000$ \\
Hardware
& Intel(R) Core(TM) i7-10750H CPU @ 2.600GHz \\
Training time
& 5 hours and 45 minutes \\
\hline
\end{tabular}
\end{table}

\FloatBarrier

\section{Additional validation results}
\label{supp:validation}

This section provides additional results supporting Section~4 of the main manuscript. 
Unless otherwise stated, calibration and boundary-quality results use 200 held-out 
simulated datasets with graph sizes from 40 to 300. The prior- and construction-matched 
MCMC experiment uses an additional fixed bank of 100 held-out datasets, and the 
representation ablation uses the separate design described in Section~\ref{supp:ablation}.

\subsection{Simulation design, recovery, and posterior calibration}
\label{supp:recovery}

For each training or validation dataset, the number of areas was sampled uniformly from $\{40,\ldots,300\}$. Planar coordinates were sampled independently over $[0,10]^2$ and connected by Delaunay triangulation. Node covariates followed $x_i\stackrel{iid}{\sim}\mathcal N(0,1)$ and exposures satisfied $\log(e_i)\stackrel{iid}{\sim}\mathcal U\{\log(2),\log(30,000)\}$. Parameters were sampled from the priors in Section~3.1 of the main manuscript, after which the filtered graph, DAGAR effects, and Poisson counts were generated sequentially. Training used 100 epochs, batch size 64, and 200 batches per epoch, for 1,280,000 online simulations. Figure~\ref{supp:fig:loss} shows the corresponding training loss.


The 200 held-out datasets were generated independently after training. For
parameter $\theta$ and held-out dataset $g=1,\ldots,G$, let $\theta^{(g)}$
denote the generating value, let $\widehat\theta^{(g)}$ denote the mean of its
approximate posterior, and let $\bar\theta=G^{-1}\sum_g\theta^{(g)}$. 
We report bias, root mean squared error (RMSE), and $R^2$ as
\[
\operatorname{Bias}
=
\frac{1}{G}\sum_{g=1}^{G}
\{\widehat\theta^{(g)}-\theta^{(g)}\},
\qquad
\operatorname{RMSE}
=
\left[
\frac{1}{G}\sum_{g=1}^{G}
\{\widehat\theta^{(g)}-\theta^{(g)}\}^2
\right]^{1/2},
\]
and
\[
R^2
=
1-
\frac{
\sum_{g=1}^{G}\{\theta^{(g)}-\widehat\theta^{(g)}\}^2
}{
\sum_{g=1}^{G}\{\theta^{(g)}-\bar\theta\}^2
}.
\]
Here $r_s$ is the Spearman rank correlation between
$\{\theta^{(g)}\}_{g=1}^{G}$ and
$\{\widehat\theta^{(g)}\}_{g=1}^{G}$. If
$[L^{(g)}_{0.025},U^{(g)}_{0.975}]$ is the central 95\% approximate-posterior
interval, empirical coverage is
\[
\frac{1}{G}\sum_{g=1}^{G}
\mathbf 1\!\left\{
L^{(g)}_{0.025}\leq\theta^{(g)}
\leq U^{(g)}_{0.975}
\right\}.
\]
The graph-size diagnostic uses the posterior $z$-score
\[
z^{(g)}
=
\frac{\theta^{(g)}-\widehat\theta^{(g)}}
{\widehat{\operatorname{sd}}(\theta\mid\mathcal D_g)}.
\]
Table~\ref{supp:tab:recovery} reports these posterior-mean recovery and
interval-calibration summaries for the four model parameters.

\begin{table}[htbp!]
\centering
\caption{Parameter recovery and interval calibration on 200 held-out simulated datasets.}
\label{supp:tab:recovery}
\begin{tabular}{lccccc}
\hline
Parameter & Bias & RMSE & $r_s$ & $R^2$ & 95\% coverage \\
\hline
$\beta_0$     &  0.003 & 0.019 & 0.999 & 0.999 & 0.980 \\
$\sigma_w^2$  & -0.008 & 0.164 & 0.889 & 0.754 & 0.945 \\
$\eta$        & -0.020 & 0.150 & 0.714 & 0.491 & 0.955 \\
$\rho$        & -0.011 & 0.158 & 0.823 & 0.679 & 0.955 \\
\hline
\end{tabular}
\end{table}

Table~\ref{supp:tab:recovery} shows particularly strong recovery for
$\beta_0$, with informative recovery and close-to-nominal coverage for the
three spatial-structure parameters. Figs.~\ref{supp:fig:recovery}--%
\ref{supp:fig:regime_errors} provide the corresponding recovery,
simulation-based calibration, coverage, graph-size, and parameter-regime
diagnostics. Together, the table and calibration figures do not reveal pronounced
systematic miscalibration. Recovery error varies across true parameter
regimes, but no isolated deterioration is apparent with graph size or edge
count.

For the calibration diagnostics, approximately uniform simulation-based
calibration (SBC) ranks indicate calibrated marginal posteriors, 
and empirical cumulative distribution function (ECDF)
difference curves close to zero and within their simultaneous reference bands
indicate agreement with uniform ranks. Empirical-versus-nominal coverage close
to the diagonal indicates calibrated credible intervals. Posterior
$z$-score distributions centered near zero, without systematic shifts across
graph-size bins, indicate no systematic graph-size-dependent location error.

\begin{figure}[htbp!]
\centering
\includegraphics[width=\textwidth]{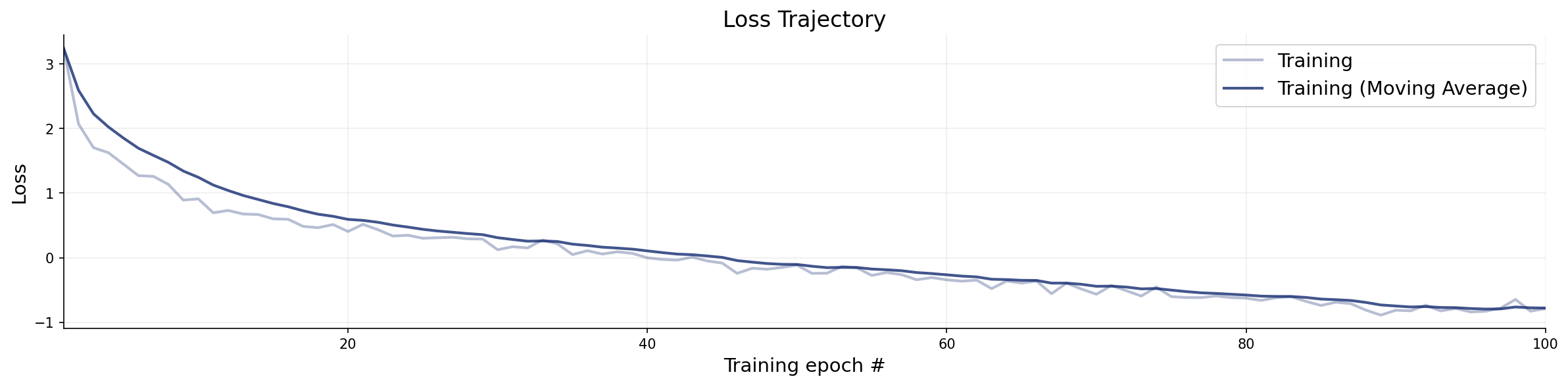}
\caption{Training loss over 100 epochs. The light line is the epoch-specific
training loss and the dark line is its moving average. A sustained decrease
followed by stabilization, rather than late divergence, is the expected
training pattern.}
\label{supp:fig:loss}
\end{figure}

\begin{figure}[htbp!]
\centering
\includegraphics[width=\textwidth]{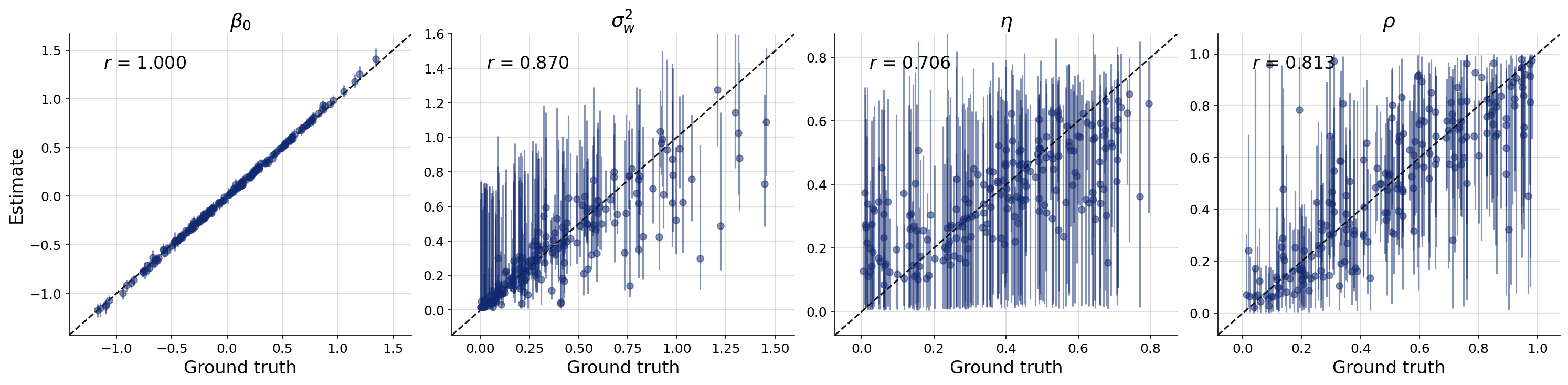}
\caption{Parameter recovery on 200 held-out simulated datasets. Points are
approximate-posterior medians and vertical lines are central 95\% posterior
intervals; the horizontal coordinate is the generating value. The dashed line
denotes equality, and $r$ is the Pearson correlation between generating values
and posterior medians. Satisfactory recovery is indicated by points near the
equality line and intervals containing the corresponding generating values at
approximately their nominal frequency.}
\label{supp:fig:recovery}
\end{figure}

\begin{figure}[htbp!]
\centering
\includegraphics[width=\textwidth]{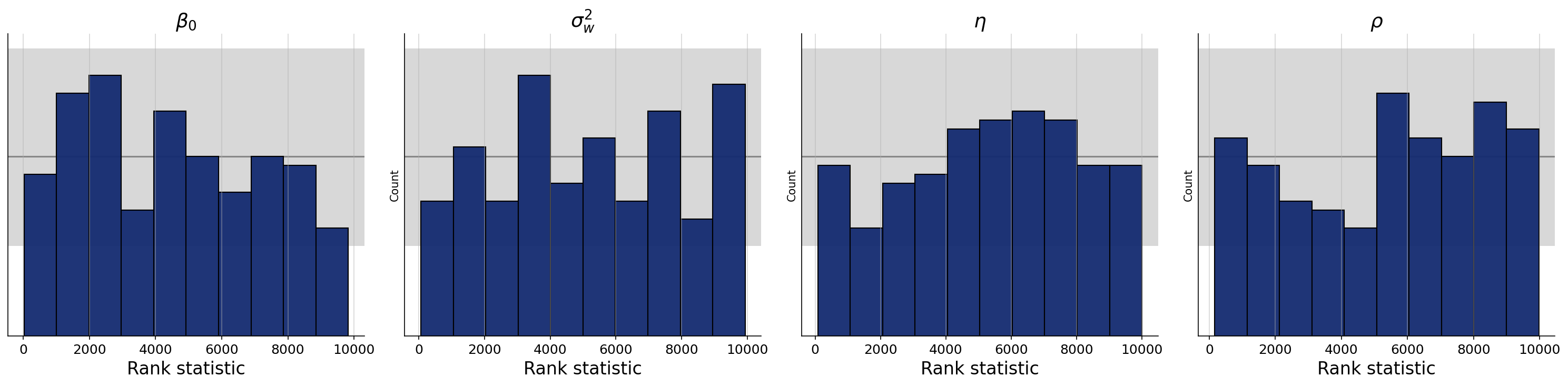}
\caption{Simulation-based calibration rank histograms for the four model
parameters on 200 held-out simulated datasets, using 10,000 posterior draws per
dataset and ten rank bins. The horizontal gray line is the expected bin count
under uniform ranks, and the gray band is the 99\% binomial reference interval
for each bin. Approximately uniform bin heights, with no systematic departures
from the reference band, indicate satisfactory marginal calibration.}
\label{supp:fig:sbc_hist}
\end{figure}

\begin{figure}[htbp!]
\centering
\includegraphics[width=0.75\textwidth]{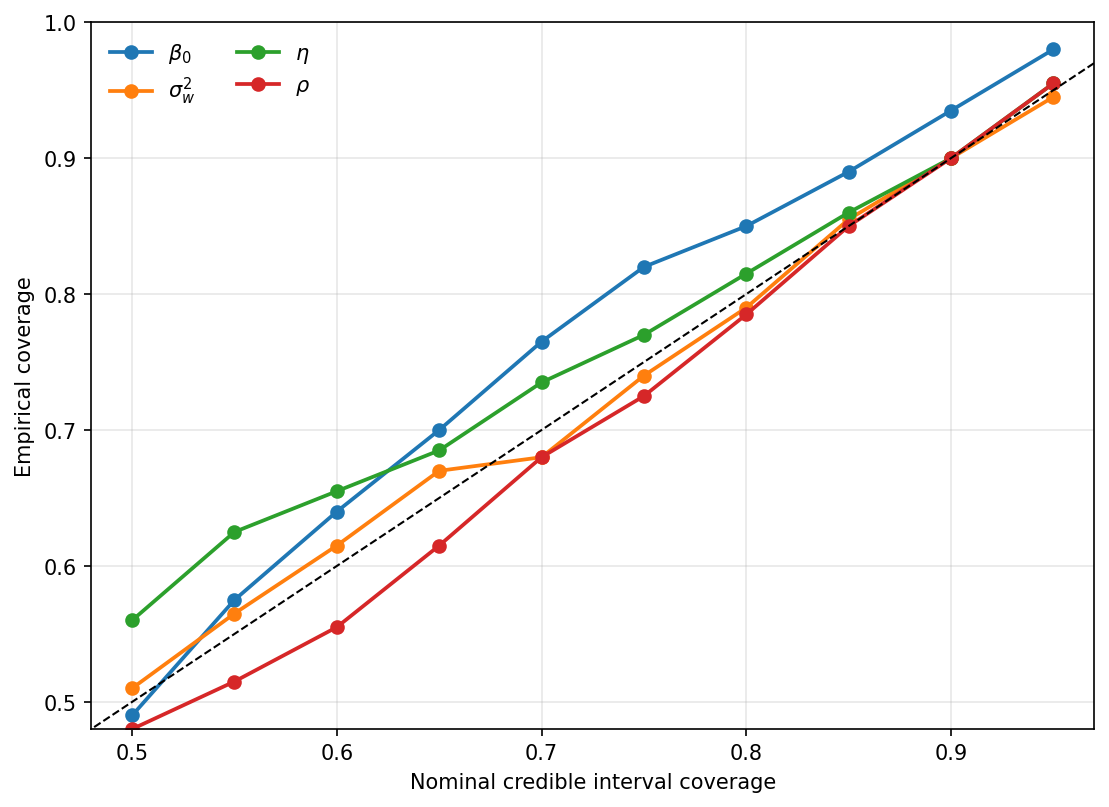}
\caption{Empirical coverage of central approximate-posterior intervals versus
their nominal coverage for the four model parameters. The dashed diagonal
denotes equality between empirical and nominal coverage. Curves close to this
line indicate calibrated interval probabilities.}
\label{supp:fig:coverage}
\end{figure}

\begin{figure}[htbp!]
\centering
\includegraphics[width=\textwidth]{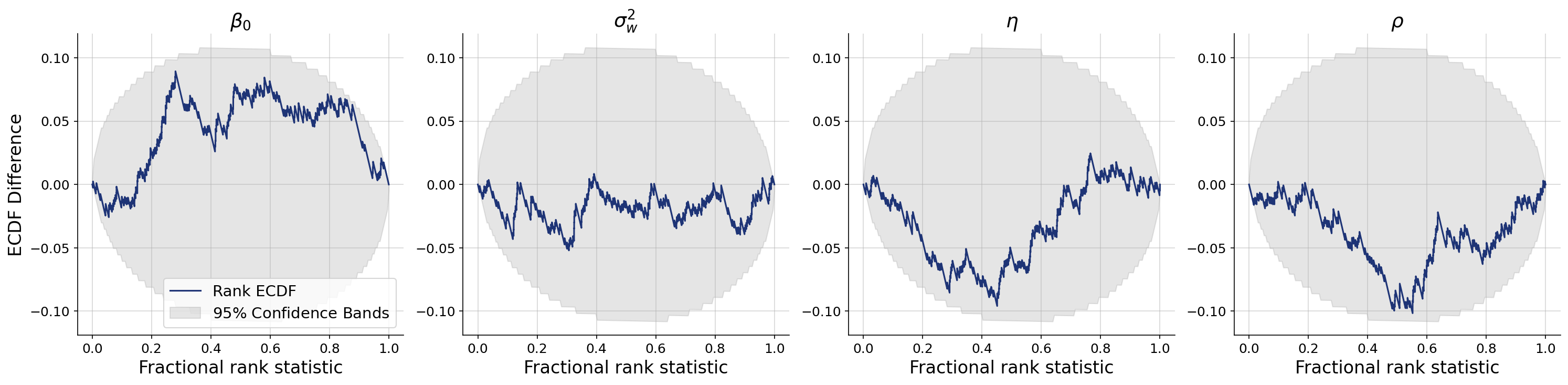}
\caption{Simulation-based calibration ECDF-difference plots for the four model
parameters. The dark curve is the empirical CDF of the fractional posterior
ranks minus the CDF of a uniform distribution, and the gray region is the 95\%
simultaneous reference band under uniform ranks. Curves close to zero and
remaining within the band indicate satisfactory marginal calibration.}
\label{supp:fig:sbc_ecdf}
\end{figure}

\begin{figure}[htbp!]
\centering
\includegraphics[width=\textwidth]{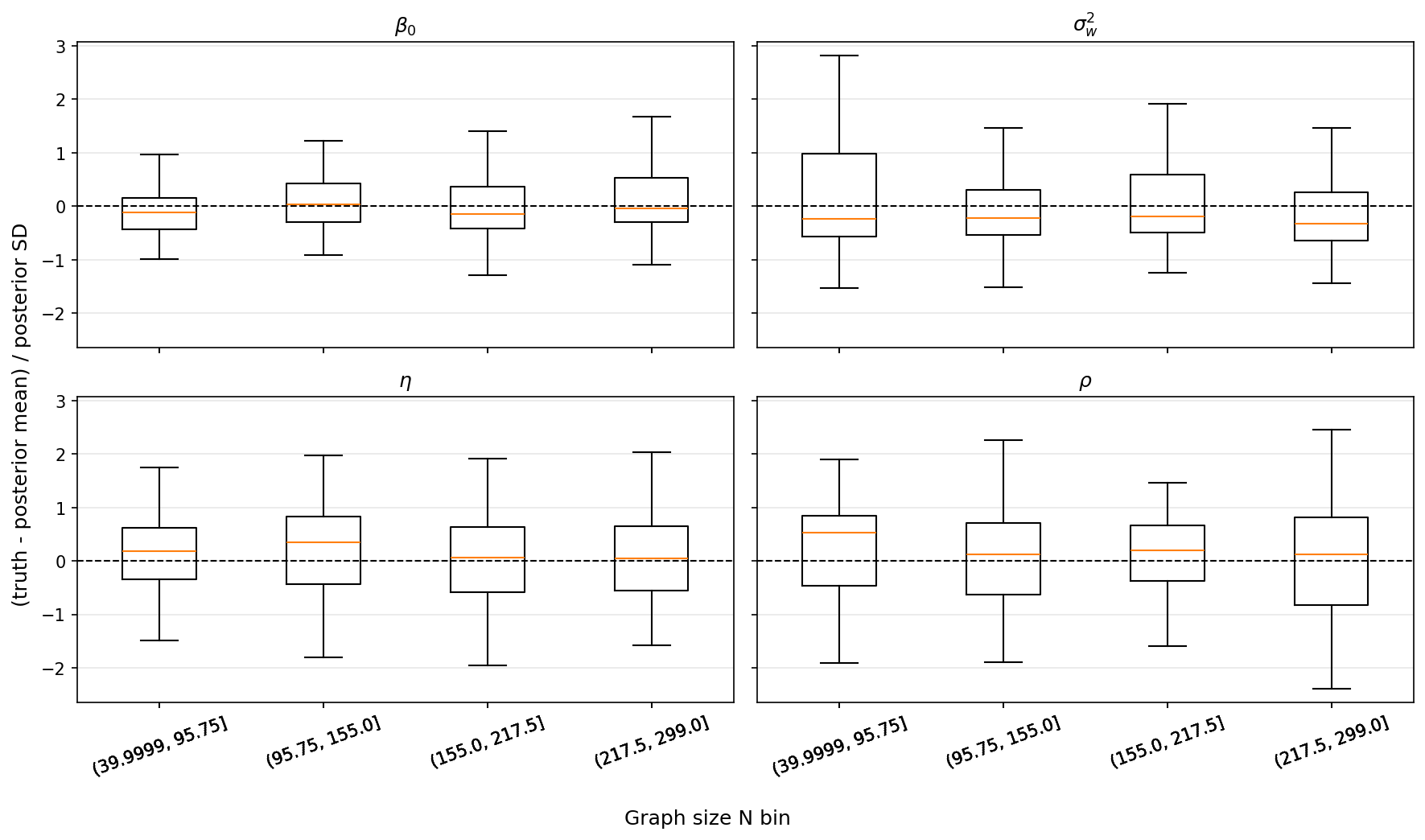}
\caption{Posterior $z$-scores
$\{\theta^{(g)}-\widehat\theta^{(g)}\}/
\widehat{\operatorname{sd}}(\theta\mid\mathcal D_g)$
stratified by graph-size quartile. Boxes show the interquartile range, orange
lines the median, and whiskers extend to the usual 1.5-interquartile-range
limits; outlying points are suppressed. The dashed line marks zero.
Distributions centered near zero without systematic shifts across graph-size
bins indicate no graph-size-dependent location error.}
\label{supp:fig:graph_zscores}
\end{figure}

\begin{figure}[htbp!]
\centering
\includegraphics[width=\textwidth]{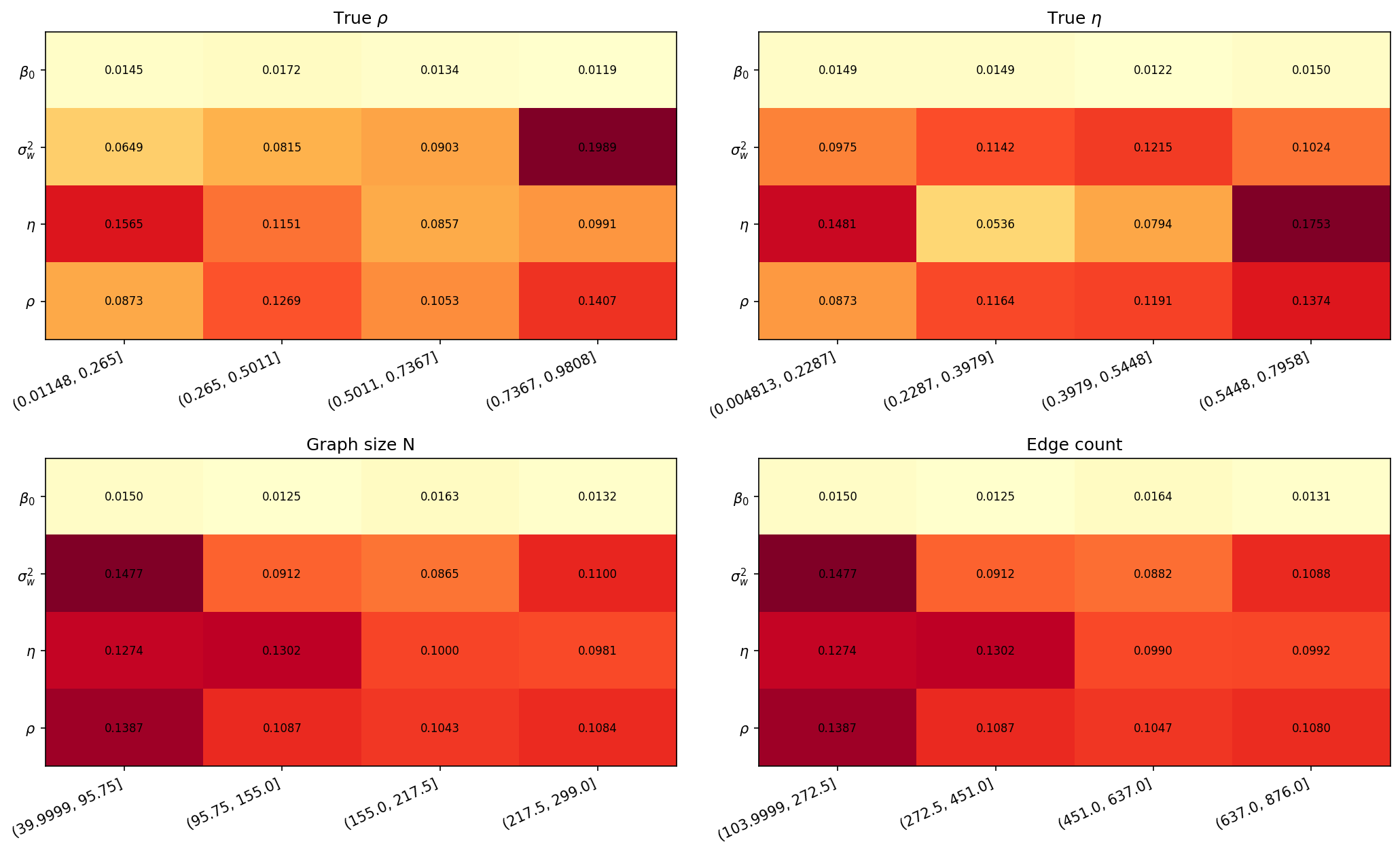}
\caption{Mean absolute posterior-mean recovery error by quartiles of the true
$\rho$, true $\eta$, graph size, and edge count. Rows identify parameters,
columns identify the corresponding quartile bins, and each cell reports its
mean absolute error. Color scales are determined separately within each panel,
so comparisons across panels should use the printed values. Smaller values and
the absence of systematic deterioration across graph-size or edge-count bins
indicate more stable recovery.}
\label{supp:fig:regime_errors}
\end{figure}

\FloatBarrier
\subsection{Boundary-probability and decision diagnostics}
\label{supp:boundary}

Using the edge-level posterior boundary probabilities defined in Section~3.1
of the main manuscript, 161 of the 200 held-out datasets contained at least one
true boundary, whereas 39 contained none. Sensitivity, area under the receiver operating characteristic curve (AUROC), and average
precision are therefore summarized over the former group; specificity, Brier
scores, and boundary-count summaries remain defined for all datasets.

For a given dataset, let $\mathcal E$ denote its set of observed geographic
edges. For each edge $e\in\mathcal E$, let $b_e\in\{0,1\}$ be the true
boundary indicator and let $\widehat p_e$ be its posterior boundary
probability. At a probability threshold $c$, edges satisfying
$\widehat p_e\geq c$ are classified as boundaries. Let
$\operatorname{TP}(c)$, $\operatorname{FP}(c)$,
$\operatorname{TN}(c)$, and $\operatorname{FN}(c)$ denote the resulting
numbers of true positives, false positives, true negatives, and false
negatives. Sensitivity, specificity, and precision at threshold $c$ are
\[
\operatorname{Sens}(c)
=
\frac{\operatorname{TP}(c)}
     {\operatorname{TP}(c)+\operatorname{FN}(c)},
\qquad
\operatorname{Spec}(c)
=
\frac{\operatorname{TN}(c)}
     {\operatorname{TN}(c)+\operatorname{FP}(c)},
\]
and
\[
\operatorname{Prec}(c)
=
\frac{\operatorname{TP}(c)}
     {\operatorname{TP}(c)+\operatorname{FP}(c)}.
\]

The receiver operating characteristic curve plots sensitivity against
$1-\operatorname{Spec}(c)$ as $c$ varies. Its area, the AUROC, can
equivalently be interpreted as the probability that a randomly selected true
boundary receives a higher posterior boundary probability than a randomly
selected non-boundary, with half weight assigned to ties.

Average precision summarizes the corresponding precision--sensitivity curve.
Specifically, let $c_1,\ldots,c_T$ denote the distinct probability thresholds,
ordered so that the resulting sensitivity values satisfy
$0=R_0\leq R_1\leq\cdots\leq R_T$, and let
$P_t=\operatorname{Prec}(c_t)$ and
$R_t=\operatorname{Sens}(c_t)$. Average precision is
\[
\operatorname{AP}
=
\sum_{t=1}^{T}(R_t-R_{t-1})P_t.
\]
Thus, AP gives greater weight to precision attained over larger increments
in sensitivity and is particularly informative when boundaries are less frequent
than non-boundaries.

Probability accuracy is measured by the Brier score
\[
\operatorname{BS}
=
\frac{1}{|\mathcal E|}
\sum_{e\in\mathcal E}
(\widehat p_e-b_e)^2.
\]
It equals the mean squared error of the posterior boundary probabilities
relative to the binary boundary indicators, so smaller values are better.

For thresholded decisions, the median-probability rule classifies edge $e$
as a boundary when $\widehat p_e > 0.5$. The Bayesian false-discovery-rate (FDR) rule of
\citet{li2015bayesian} instead considers each distinct posterior probability $c$ as a candidate
threshold. For the selected set
\[
\mathcal B(c)
=
\{e\in\mathcal E:\widehat p_e\geq c\},
\]
its estimated posterior false discovery rate is
\[
\widehat{\operatorname{FDR}}(c)
=
\frac{
\sum_{e\in\mathcal B(c)}(1-\widehat p_e)
}{
|\mathcal B(c)|
}.
\]
Among the candidate thresholds satisfying
$\widehat{\operatorname{FDR}}(c)\leq0.05$, we use the smallest threshold,
thereby selecting the largest admissible set of edges. If no candidate
threshold satisfies this condition, no boundary is selected. Equal posterior
probabilities are thresholded together. This criterion limits 
the posterior expected false-discovery proportion
computed from the approximate posterior boundary probabilities. It should not
be interpreted as establishing frequentist FDR control under repeated
sampling.

The pooled metrics in the main manuscript apply these definitions after
combining edges across all 200 datasets. The dataset-level results below
apply them separately within each dataset and then summarize the resulting
values across datasets.

Table~\ref{supp:tab:boundary} summarizes the resulting probability-quality
and decision-performance measures. Under the scalar threshold mechanism,
posterior boundary probabilities are nondecreasing in edge dissimilarity
within each dataset, so within-dataset AUROC and average precision are
expected to be essentially perfect when true boundaries are present. The
Brier score and thresholded decision measures therefore provide the more
informative dataset-level assessment of probability quality and operating
behavior. These dataset-level summaries differ 
from the pooled edge metrics reported
in the main manuscript because the metrics are computed separately within
each dataset and then averaged here, whereas the main-manuscript values are
computed after pooling edges across datasets.

\begin{table}[htbp!]
\centering
\caption{Dataset-level boundary-probability and decision summaries on held-out
simulations. Sensitivity, area under the receiver operating characteristic
curve (AUROC), and average precision use the 161 datasets containing at least
one true boundary; all other summaries use all 200 datasets. SD denotes
standard deviation, IQR interquartile range, and FDR false discovery rate.}
\label{supp:tab:boundary}
\resizebox{\textwidth}{!}{%
\begin{tabular}{lccccc}
\hline
\textbf{Metric}
& $\bs{n}$
& \textbf{Mean}
& \textbf{SD}
& \textbf{Median}
& \textbf{IQR} \\
\hline
Dataset AUROC & 161 & 1.000 & 0.000 & 1.000 & 0.000 \\
Dataset average precision & 161 & 1.000 & 0.000 & 1.000 & 0.000 \\
Dataset Brier score & 200 & 0.057 & 0.049 & 0.048 & 0.046 \\
FDR sensitivity & 161 & 0.170 & 0.332 & 0.000 & 0.000 \\
FDR specificity & 200 & 0.998 & 0.011 & 1.000 & 0.000 \\
Median-probability selected boundaries
& 200 & 78.430 & 77.718 & 56.500 & 91.750 \\
True boundaries
& 200 & 91.240 & 93.327 & 65.500 & 142.000 \\
Median-probability sensitivity
& 161 & 0.715 & 0.323 & 0.848 & 0.527 \\
Median-probability specificity
& 200 & 0.970 & 0.057 & 1.000 & 0.036 \\
\hline
\end{tabular}%
}
\end{table}

At $\alpha=0.05$, the Bayesian FDR rule yields a conservative operating
point: 121 of the 161 datasets containing true boundaries have zero
sensitivity, while median specificity is one. The median-probability rule
produces a less conservative operating point: only seven of these datasets
have zero sensitivity, while specificity remains high. This comparison does
not imply that the FDR rule is generally undesirable. It is appropriate when
limiting the posterior expected proportion of false selections is the primary
objective, whereas continuous posterior probabilities and the
median-probability rule are more informative for exploratory boundary mapping
in the present setting. The 95\% posterior interval for the total number of
boundaries contained the true count in 195 of the 200 held-out datasets
($0.975$).

Figs.~\ref{supp:fig:ranking} and \ref{supp:fig:mpm} provide the
corresponding ranking and thresholded-decision diagnostics.
They illustrate the sensitivity--specificity tradeoff between the two
operating rules, while also showing that the posterior probabilities retain
useful edge rankings independently of either threshold.

\begin{figure}[htpb!]
\centering
\includegraphics[width=\textwidth]{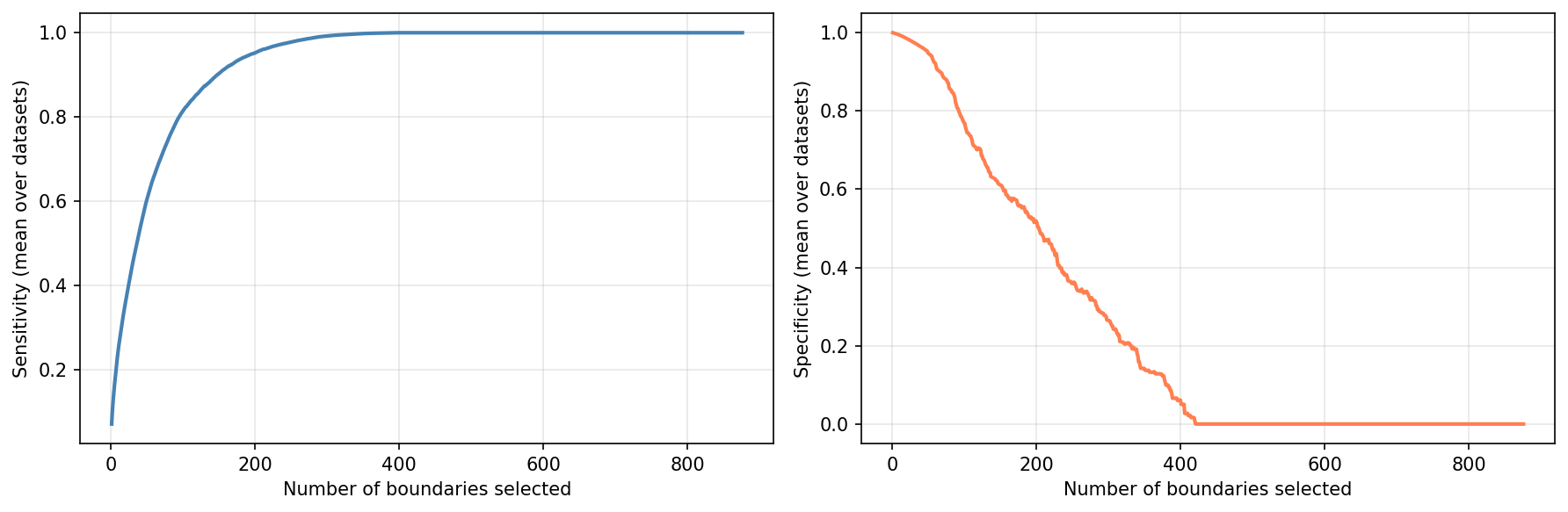}
\caption{Ranking diagnostic for posterior boundary probabilities. Mean
sensitivity (left) and specificity (right) are shown as functions of the
number of selected boundaries. Sensitivity excludes datasets containing no
true boundaries.}
\label{supp:fig:ranking}
\end{figure}

\begin{figure}[htbp!]
\centering
\includegraphics[width=\textwidth]{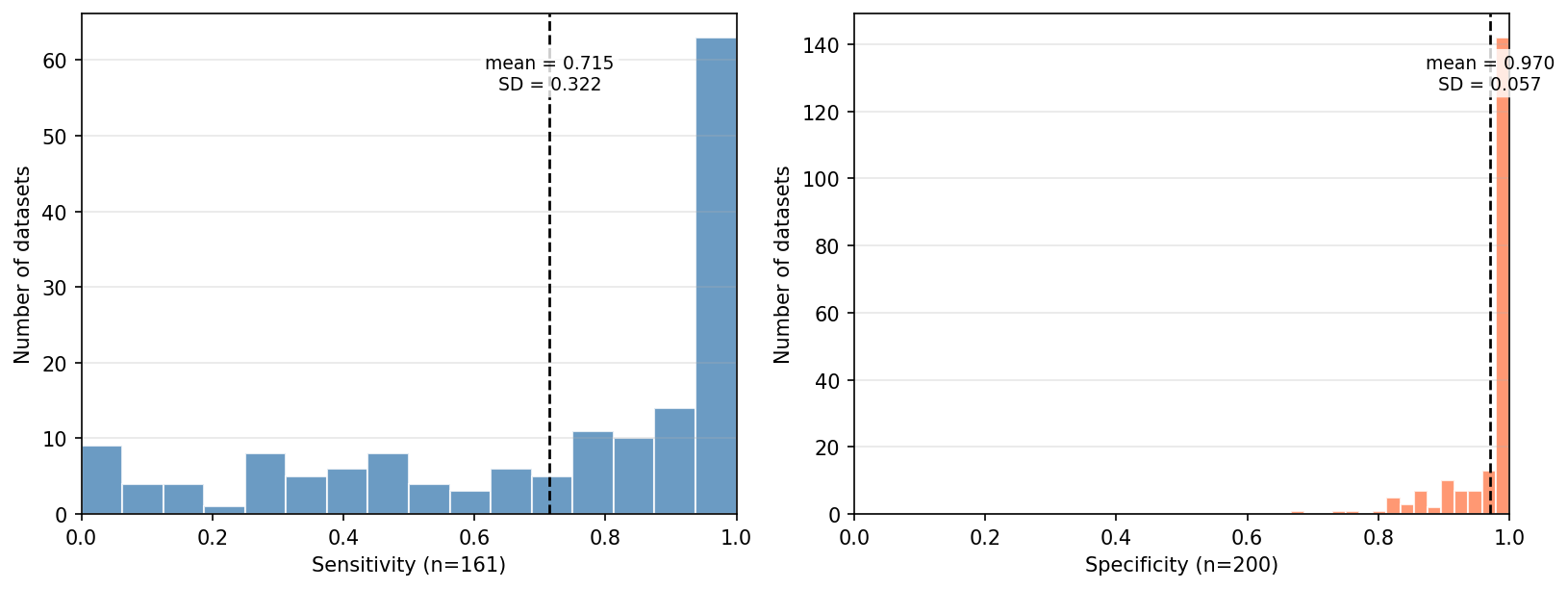}
\caption{Sensitivity (left) and specificity (right) under the
median-probability rule. Sensitivity is shown for the 161 datasets containing
at least one true boundary; specificity includes all 200 held-out datasets.}
\label{supp:fig:mpm}
\end{figure}

\FloatBarrier

\subsection{Prior- and construction-matched MCMC comparison}
\label{supp:mcmc_sim}

The matched benchmark in Section~4.3 of the
main manuscript uses 100 fixed held-out datasets and 10,000 retained draws
from each method. The MCMC implementation is constructed to target the same posterior implied
by the model and priors used for ABI-DAGAR training:
$\beta_0\sim\mathcal N(0,0.5^2)$, $\sigma_w^2\sim|\mathcal N(0,0.5)|$, $\eta/M\sim\mathcal U(0,1)$, and $\rho\sim\mathcal U(0,1)$. 
Thus, this experiment
holds the Bayesian model fixed and permits posterior differences to be
interpreted as approximation or posterior-simulation discrepancies.

Each MCMC fit used one chain of 20,000 iterations with 10,000 burn-in iterations. Adaptive Metropolis-within-Gibbs updates were applied to the intercept, latent field, log spatial variance, threshold parameter, and logit-transformed dependence parameter. Mean acceptance rates were $0.350$, $0.465$, $0.364$, $0.545$, and $0.370$ for $\beta_0$, the latent field, $\sigma_w^2$, $\eta$, and $\rho$, respectively.

For $S$ retained MCMC draws, sample posterior variance $s^2$, and spectral
variance estimate $\widehat v_0$, the effective sample size (ESS) and Monte
Carlo standard error (MCSE) of the posterior mean were computed as
\[
\widehat{\operatorname{ESS}}
=
\frac{Ss^2}{\widehat v_0},
\qquad
\widehat{\operatorname{MCSE}}(\bar\theta)
=
\sqrt{\frac{\widehat v_0}{S}}
=
\frac{s}{\sqrt{\widehat{\operatorname{ESS}}}},
\]
These quantities summarize the numerical
precision of the estimated posterior means.
Table~\ref{supp:tab:sim_mcmc_precision} reports ESS, MCSE, and MCSE relative
to the posterior standard deviation across the 100 matched datasets. Precision
is strongest for $\beta_0$ and $\eta$ and more heterogeneous for
$\sigma_w^2$ and $\rho$.

\begin{table}[htbp!]
\centering
\small
\caption{Single-chain precision in the matched MCMC experiment. Entries are median [IQR] across 100 datasets.}
\label{supp:tab:sim_mcmc_precision}
\begin{tabular}{lcccc}
\hline
Parameter & ESS & MCSE $(\times10^{-3})$ & MCSE / SD (\%) & ESS $\geq100$ \\
\hline
$\beta_0$    & 448.800 [346.800, 573.300]  & 0.600 [0.355, 0.842]  & 4.720 [4.180, 5.370]  & 100 \\
$\sigma_w^2$ & 160.400 [59.900, 467.000]   & 11.277 [3.798, 20.871] & 7.900 [4.630, 12.920] & 62 \\
$\eta$       & 626.100 [191.000, 1290.600] & 3.277 [1.939, 6.550]  & 4.000 [2.780, 7.240]  & 87 \\
$\rho$       & 173.100 [73.700, 380.700]   & 6.620 [4.215, 10.387] & 7.600 [5.130, 11.650] & 69 \\
\hline
\end{tabular}
\end{table}

As a sensitivity analysis, we recomputed the corrected decomposition after
retaining, separately for each parameter, only datasets with MCMC ESS of at
least 20. This retained 100, 96, 98, and 96 datasets for $\beta_0$,
$\sigma_w^2$, $\eta$, and $\rho$, respectively. On the $10^3$ scale,
$\widehat D_j$ changed from $(0.278,10.089,8.681,9.926)$ to
$(0.278,10.558,8.624,9.420)$, while the corresponding posterior-mean
correlations changed from $(1.000,0.959,0.835,0.933)$ to
$(1.000,0.957,0.839,0.932)$. The inequality
$\widehat D_j<\widehat R_j$ also held under ESS cutoffs of 50 and 100.
Thus, the decomposition conclusions are not driven by the comparatively
small number of lower-ESS fits.

These diagnostics quantify the numerical precision of the MCMC reference.
We next describe how finite posterior-simulation error enters the
posterior-mean decomposition reported in the main manuscript.

The decomposition in the main manuscript concerns posterior means, not the
distance between complete posterior distributions. For dataset $g$ and
parameter $j$, let $a_{gj}$ and $m_{gj}$ be the exact ABI and reference
posterior means, respectively, and let $\theta_{gj}$ be the generating value.
With $G=100$, the four finite-bank targets are
\[
T_j=G^{-1}\sum_g(a_{gj}-\theta_{gj})^2,\quad
R_j=G^{-1}\sum_g(m_{gj}-\theta_{gj})^2,\quad
D_j=G^{-1}\sum_g(a_{gj}-m_{gj})^2,
\]
\[
C_j=2G^{-1}\sum_g(m_{gj}-\theta_{gj})(a_{gj}-m_{gj}),
\qquad T_j=R_j+D_j+C_j.
\]
Thus $T_j$ and $R_j$ are the ABI and reference posterior-mean mean squared error (MSE),
$D_j$ is the ABI--MCMC posterior-mean discrepancy, defined as the
across-dataset mean squared difference between the two posterior means, and
$C_j$ is the interaction term in the recovery-error identity.

The posterior means are not observed exactly but estimated using 10,000 draws.
Write $\widehat a_{gj}=a_{gj}+\epsilon_{A,gj}$ and
$\widehat m_{gj}=m_{gj}+\epsilon_{M,gj}$, where the simulation errors
are conditionally mean zero and independent between methods, with variances
$v_{A,gj}$ and $v_{M,gj}$. We estimate $v_{A,gj}$ by the ABI posterior variance divided by 10,000,
since posterior draws from the fitted flow are generated independently, and
$v_{M,gj}$ by squared MCMC MCSE. The
finite-draw-corrected estimators are
\[
\widehat T_j=G^{-1}\sum_g\{(\widehat a_{gj}-\theta_{gj})^2-\widehat v_{A,gj}\},
\quad
\widehat R_j=G^{-1}\sum_g\{(\widehat m_{gj}-\theta_{gj})^2-\widehat v_{M,gj}\},
\]
\[
\widehat D_j=G^{-1}\sum_g\{(\widehat a_{gj}-\widehat m_{gj})^2
-\widehat v_{A,gj}-\widehat v_{M,gj}\},
\]
\[
\widehat C_j=G^{-1}\sum_g\{2(\widehat m_{gj}-\theta_{gj})
(\widehat a_{gj}-\widehat m_{gj})+2\widehat v_{M,gj}\}.
\]
These satisfy
$\widehat T_j=\widehat R_j+\widehat D_j+\widehat C_j$
up to numerical precision. The $2\widehat v_{M,gj}$ term in
$\widehat C_j$ corrects the finite-draw covariance induced by the appearance
of the same estimated MCMC mean in both factors of the interaction. 
Fig.~\ref{supp:fig:matched_decomposition} displays the resulting corrected
components across the four parameters.

The recovery-error decomposition summarizes recovery and ABI--MCMC agreement
through posterior means. We therefore complement it with diagnostics of
posterior uncertainty, marginal distributional shape, and the edge-level
boundary probabilities that are the primary inferential target.

Table~\ref{supp:tab:matched_parameter_agreement} addresses two questions
distinct from the squared-error decomposition. For parameter $j$, the reported
ABI and matched-MCMC mean absolute errors (MAEs) are
\[
\operatorname{MAE}_{A,j}
=
\frac{1}{G}\sum_{g=1}^{G}
\left|\widehat a_{gj}-\theta_{gj}\right|,
\qquad
\operatorname{MAE}_{M,j}
=
\frac{1}{G}\sum_{g=1}^{G}
\left|\widehat m_{gj}-\theta_{gj}\right|,
\]
where $\widehat a_{gj}$ and $\widehat m_{gj}$ are the posterior means estimated
from the retained ABI and MCMC draws. Coverage is the fraction of the 100
datasets in which the corresponding central 95\% posterior interval contains
the generating value. Mean correlation is the Pearson correlation, across
datasets, between the 100 paired ABI and MCMC posterior means; it is not a
within-posterior parameter correlation. Width ratio is the median, across
datasets, of the ABI 95\% interval width divided by the MCMC width.

\begin{table}[htbp!]
\centering
\small
\caption{Parameter recovery and ABI--MCMC posterior agreement over the 100 matched datasets. Width ratio is the median ABI-to-MCMC 95\% interval-width ratio.}
\label{supp:tab:matched_parameter_agreement}
\resizebox{\textwidth}{!}{%
\begin{tabular}{lcccccc}
\hline
Parameter & ABI MAE & MCMC MAE & ABI coverage & MCMC coverage & Posterior-mean corr. & Width ratio \\
\hline
$\beta_0$    & 0.019 & 0.013 & 0.900 & 0.890 & 1.000 & 1.529 \\
$\sigma_w^2$ & 0.110 & 0.116 & 0.960 & 0.890 & 0.959 & 1.150 \\
$\eta$       & 0.115 & 0.081 & 0.950 & 0.940 & 0.835 & 1.566 \\
$\rho$       & 0.099 & 0.078 & 0.960 & 0.900 & 0.933 & 1.405 \\
\hline
\end{tabular}%
}
\end{table}

ABI posterior-mean MAE is larger for $\beta_0$, $\eta$, and $\rho$, whereas it
is slightly smaller for $\sigma_w^2$. Coverage remains between $0.890$ and
$0.960$ for both methods. Posterior means track one another most closely for
$\beta_0$ and least closely for $\eta$, and all median width ratios exceed
one, indicating wider ABI intervals. For each dataset, interval overlap was
also measured as the length of the intersection divided by the width of the
shorter interval; its mean across datasets ranges from $0.948$ to $0.986$ over
the four parameters. Fig.~\ref{supp:fig:matched_truth_recovery} shows the
dataset-level posterior means and intervals against the generating values,
whereas Fig.~\ref{supp:fig:matched_parameter_scatter} directly plots each ABI
posterior mean against its matched-MCMC counterpart. The first view assesses
recovery relative to truth; the second isolates agreement between methods. 
Figure~\ref{supp:fig:matched_recovery_bars} provides the corresponding
aggregate comparison of mean absolute error and empirical interval coverage.
Figure~\ref{supp:fig:matched_bias_width} complements these summaries by
showing the across-dataset distributions of posterior-mean errors and
95\% interval widths.

To compare marginal posterior shape rather than posterior means alone, let
$F_{A,gj}$ and $F_{M,gj}$ denote the empirical marginal posterior
distributions of the ABI and matched-MCMC draws for dataset $g$ and parameter
$j$. Their one-dimensional Wasserstein distance is
\[
W_1(F_{A,gj},F_{M,gj})
=
\int_0^1
\left|
F_{A,gj}^{-1}(u)-F_{M,gj}^{-1}(u)
\right|\,du.
\]
The reported scaled distance is
$W_1(F_{A,gj},F_{M,gj})/s_{M,gj}$, where $s_{M,gj}$ is the matched-MCMC
posterior standard deviation. Fig.~\ref{supp:fig:matched_posterior} displays
the 100 dataset-specific scaled distances and interval-width ratios for each
parameter. Median scaled Wasserstein distances range from $0.404$ for
$\sigma_w^2$ to $0.991$ for $\eta$; the long upper tails show that a smaller
number of datasets have substantially larger distributional differences.

Because posterior boundary probability is the primary edge-level target, we
compared the ABI and MCMC probabilities across all geographic edges within
each dataset. Let
$\mathcal B_{A,g}=\{e:\widehat p_{e_{A,g}}>0.5\}$ and
$\mathcal B_{M,g}=\{e:\widehat p_{e_{M,g}}>0.5\}$ be the corresponding
median-probability boundary sets. Their Jaccard index is
\[
J_g
=
\frac{
|\mathcal B_{A,g}\cap\mathcal B_{M,g}|
}{
|\mathcal B_{A,g}\cup\mathcal B_{M,g}|
}.
\]
Following the implemented convention, $J_g=1$ when both sets are empty.
The resulting within-dataset Pearson correlations have mean $0.908$ and median
$0.947$ across the 100 datasets, and the mean of the within-dataset edgewise
absolute differences is $0.071$. The Jaccard indices have mean $0.575$ and
median $0.675$; ABI and MCMC select 91.650 and 86.850 boundaries per dataset
on average.

Agreement between methods does not by itself establish recovery of the true
boundary indicators. Table~\ref{supp:tab:matched_boundary_recovery} therefore
reports truth-based edge metrics for each method, and
Fig.~\ref{supp:fig:matched_boundary_metrics} displays the same comparison.
AUROC, average precision, and sensitivity are averaged over the 83 datasets
containing at least one true boundary; Brier score, specificity, and boundary-count
coverage use all 100 datasets. As expected under the shared scalar threshold mechanism, both methods show
perfect within-dataset boundary ranking by AUROC and average precision. Matched MCMC has the lower mean Brier score and slightly higher
median-probability sensitivity and specificity, while the true boundary count
lies in the corresponding 95\% posterior interval in $97\%$ of ABI datasets
and $96\%$ of MCMC datasets.

\begin{table}[htbp!]
\centering
\small
\caption{Truth-based boundary recovery in the 100 matched datasets. Entries
are means across eligible datasets. AUROC, average precision (AP), and
median-probability-model (MPM) sensitivity use the 83 datasets containing at
least one true boundary; all other columns use all 100 datasets. Count coverage
is the proportion for which the 95\% posterior interval for the number of
boundaries contains the true count.}
\label{supp:tab:matched_boundary_recovery}
\resizebox{\textwidth}{!}{%
\begin{tabular}{lcccccc}
\hline
Method & AUROC & AP & Brier score & MPM sensitivity & MPM specificity & Count coverage \\
\hline
ABI          & 1.000 & 1.000 & 0.062 & 0.689 & 0.964 & 0.970 \\
Matched MCMC & 1.000 & 1.000 & 0.040 & 0.734 & 0.984 & 0.960 \\
\hline
\end{tabular}%
}
\end{table}

Finally, we record posterior-generation time for this matched experiment as a
computational descriptor rather than an accuracy criterion.
Mean posterior-generation times are $0.772$ seconds per dataset for ABI and
$2.707$ seconds for MCMC, with a mean dataset-specific MCMC-to-ABI ratio of
$3.578$. These timings compare posterior generation after training and do not
include the one-time ABI training cost. Figure~\ref{supp:fig:matched_runtime} 
shows the paired dataset-level runtimes
and the distribution of the matched-MCMC-to-ABI elapsed-time ratio.

\begin{figure}[htbp!]\centering
\includegraphics[width=\textwidth]{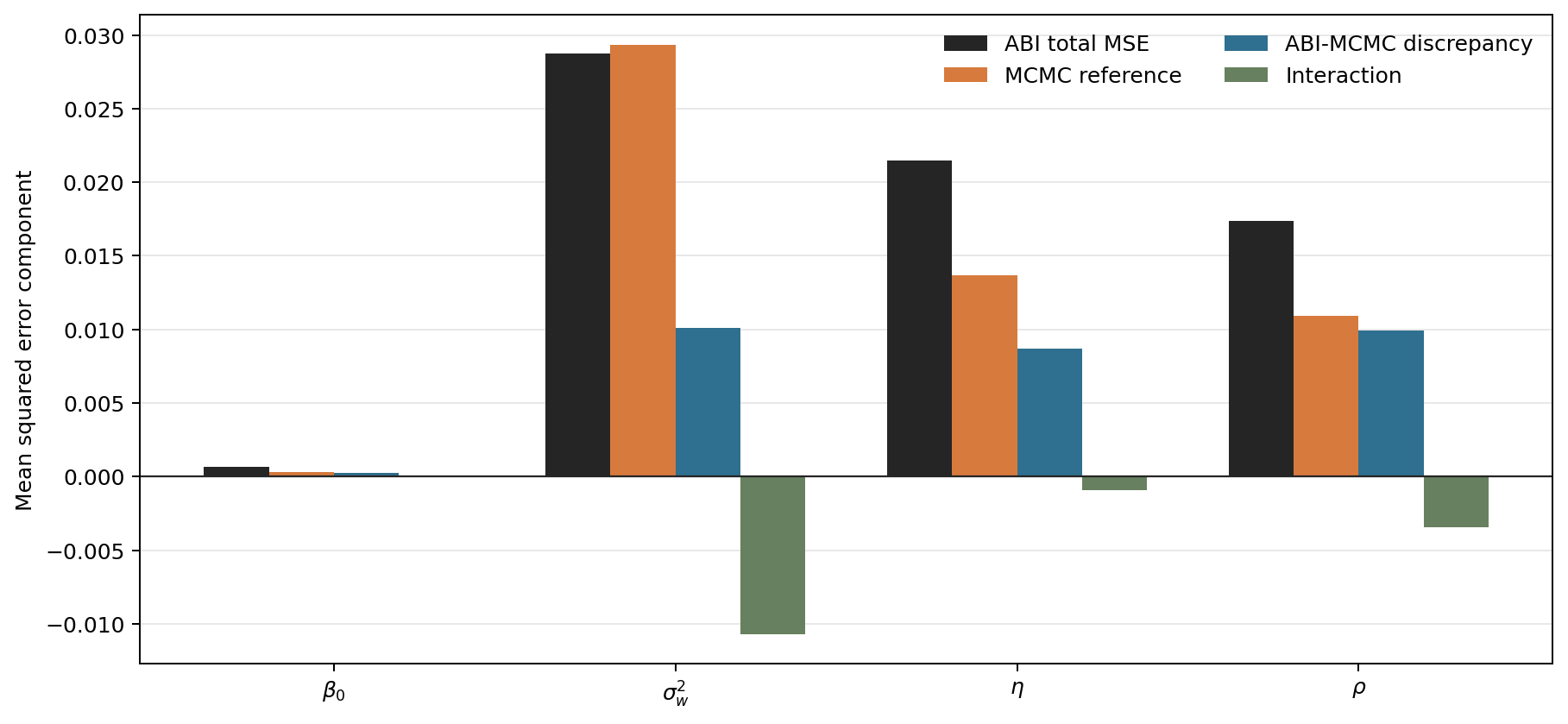}
\caption{Finite-draw-corrected posterior-mean recovery-error decomposition for
the four model parameters, averaged across 100 held-out datasets. Bars show
ABI posterior-mean MSE $\widehat T_j$, reference-posterior MSE $\widehat R_j$, ABI–MCMC posterior-mean discrepancy
$\widehat D_j$, and interaction $\widehat C_j$; the latter is
not a variance component and can be negative.}
\label{supp:fig:matched_decomposition}
\end{figure}

\begin{figure}[htbp!]
\centering
\includegraphics[width=\textwidth]{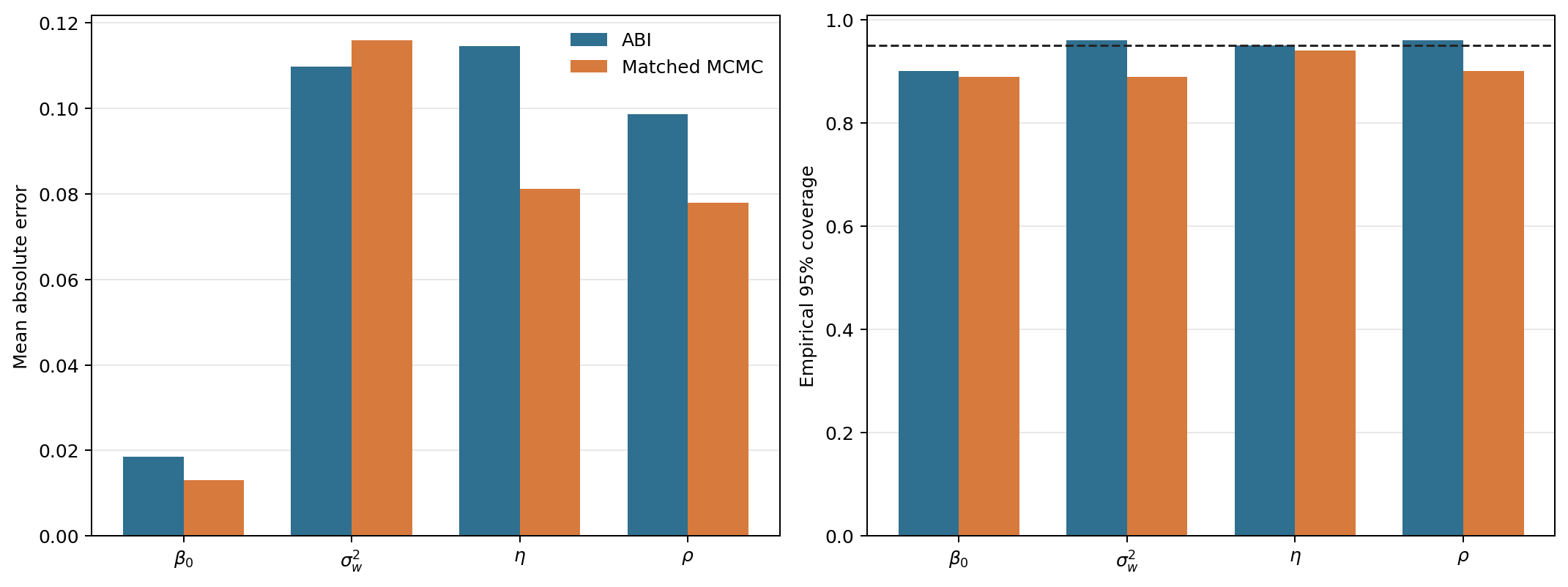}
\caption{Aggregate parameter recovery over the 100 matched datasets.
Left: mean absolute error of the posterior means relative to the generating
values. Right: empirical coverage of the central 95\% posterior intervals;
the dashed line denotes the nominal coverage level.}
\label{supp:fig:matched_recovery_bars}
\end{figure}

\begin{figure}[p]\centering
\includegraphics[width=\textwidth]{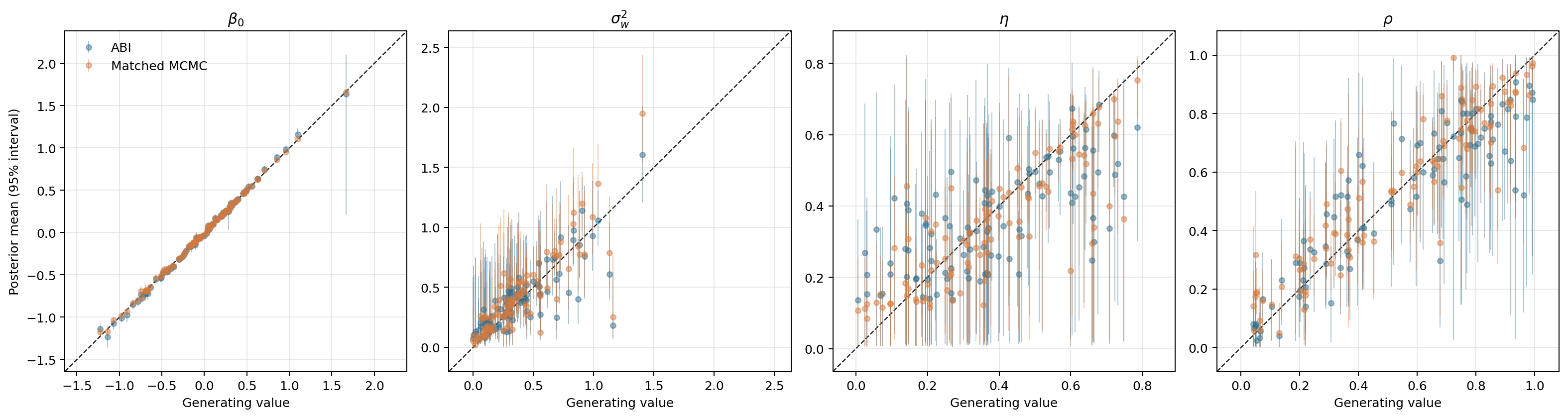}
\caption{Parameter recovery over the 100 matched datasets. Points are posterior
means and vertical lines are central 95\% posterior intervals for ABI-DAGAR and
matched MCMC; the horizontal coordinate is the generating value. The dashed
line denotes equality.}
\label{supp:fig:matched_truth_recovery}
\end{figure}

\begin{figure}[p]\centering
\includegraphics[width=\textwidth]{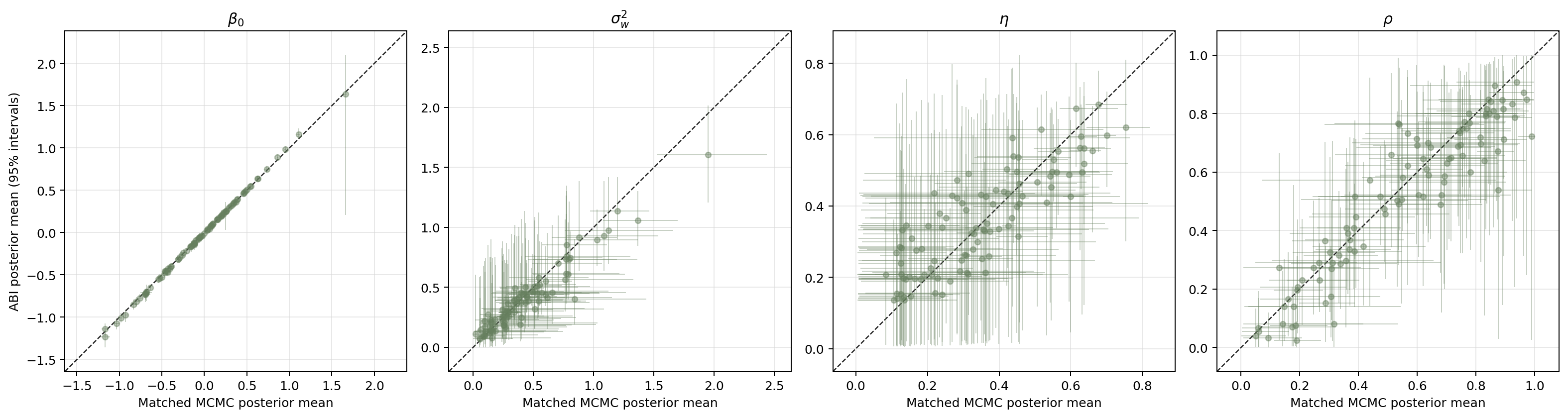}
\caption{Dataset-level agreement between ABI-DAGAR and matched MCMC. Each point
pairs the two posterior means for one dataset; horizontal and vertical lines
are the corresponding matched-MCMC and ABI-DAGAR central 95\% posterior
intervals. The dashed line denotes equal posterior means.}
\label{supp:fig:matched_parameter_scatter}
\end{figure}

\begin{figure}[htbp!]
\centering
\includegraphics[width=\textwidth]
{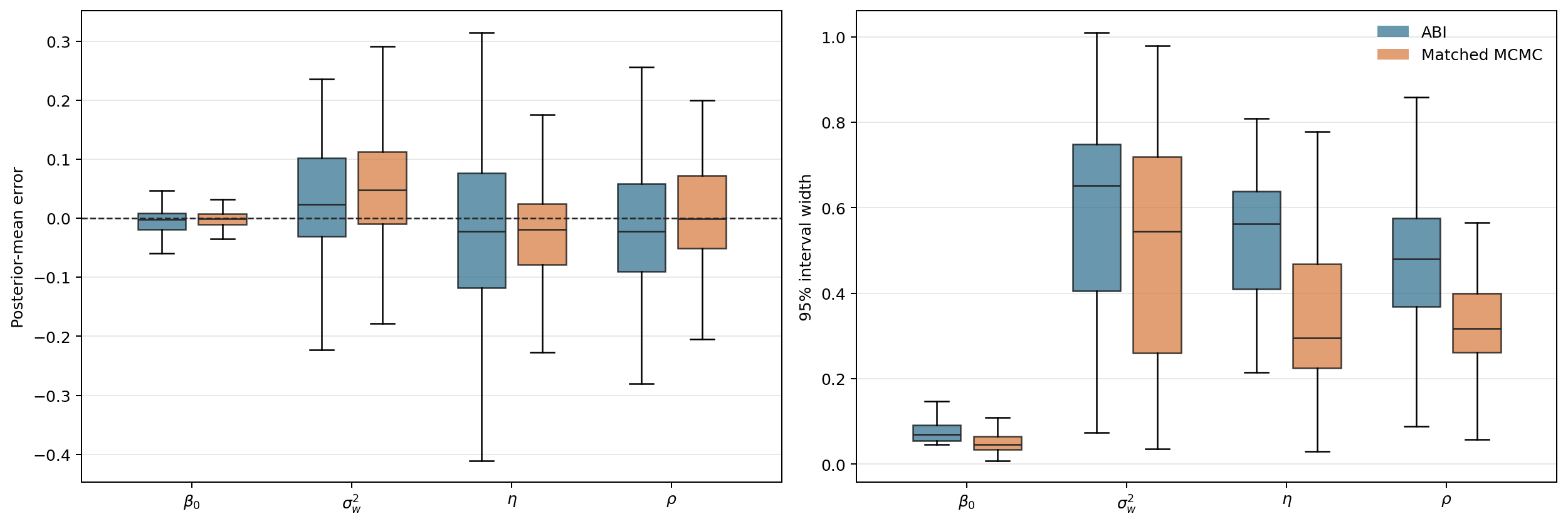}
\caption{Across-dataset parameter recovery and posterior uncertainty in the
100 matched datasets. Left: distributions of posterior-mean errors relative
to the generating values. Right: distributions of central 95\% posterior
interval widths. Boxplots compare ABI-DAGAR with matched MCMC for the four
reported model parameters.}
\label{supp:fig:matched_bias_width}
\end{figure}

\begin{figure}[htbp!]\centering
\includegraphics[width=\textwidth]{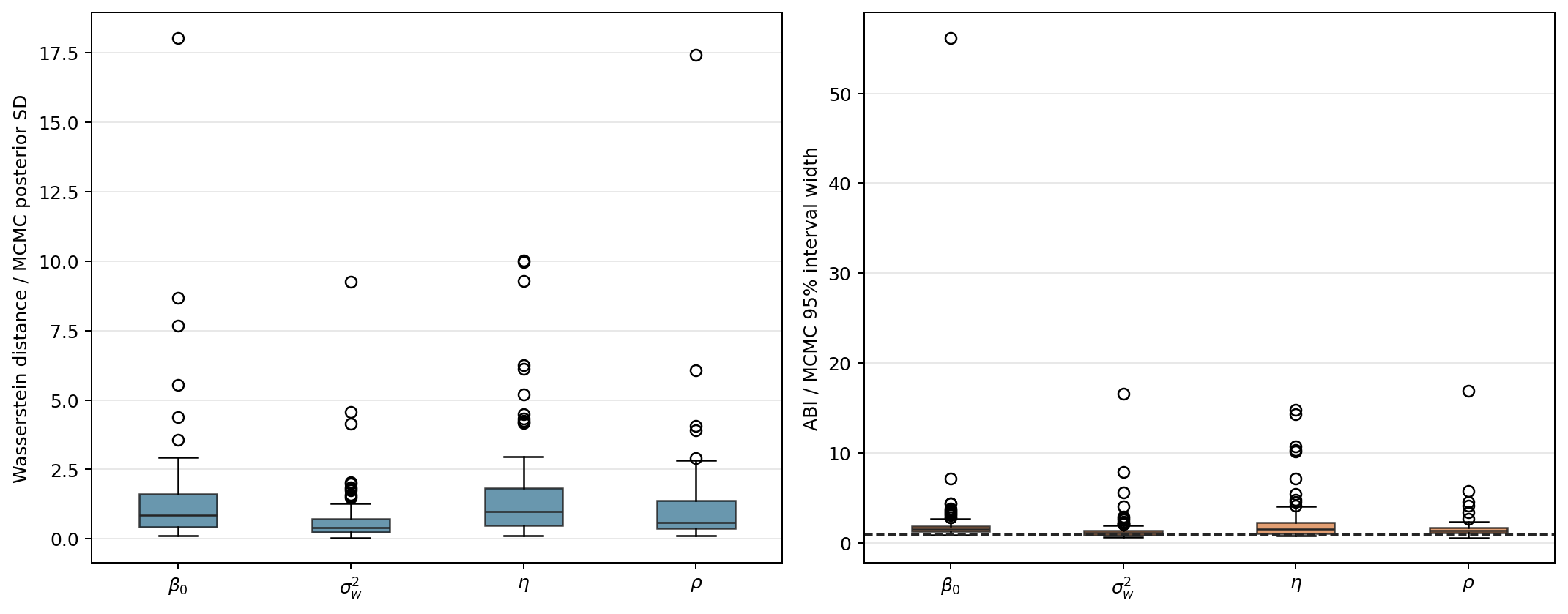}
\caption{Marginal posterior agreement across the 100 matched datasets. Each
boxplot contains one value per dataset. Left: Wasserstein distance between ABI
and MCMC marginal draws, scaled by the MCMC posterior standard deviation.
Right: ABI-to-MCMC 95\% interval-width ratio; the dashed line at one denotes
equal widths.}
\label{supp:fig:matched_posterior}
\end{figure}

\begin{figure}[htbp!]\centering
\includegraphics[width=\textwidth]{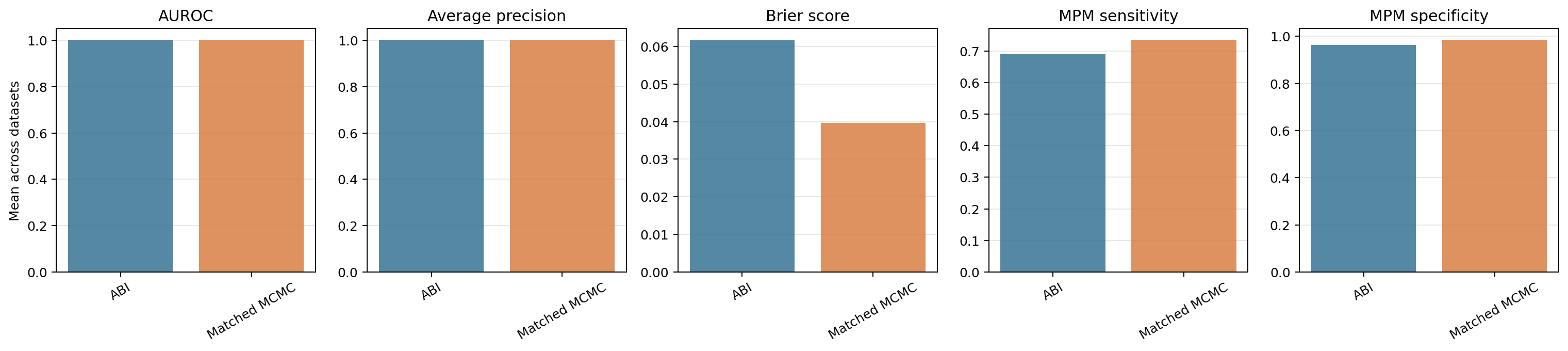}
\caption{Truth-based boundary recovery for ABI-DAGAR and matched MCMC. Bars are
means across the eligible matched datasets. AUROC, average precision, and MPM
sensitivity use the 83 datasets containing at least one true boundary; Brier
score and MPM specificity use all 100 datasets. Lower Brier score is better;
higher values are better for the other metrics.}
\label{supp:fig:matched_boundary_metrics}
\end{figure}

\begin{figure}[htbp!]
\centering
\includegraphics[width=\textwidth]{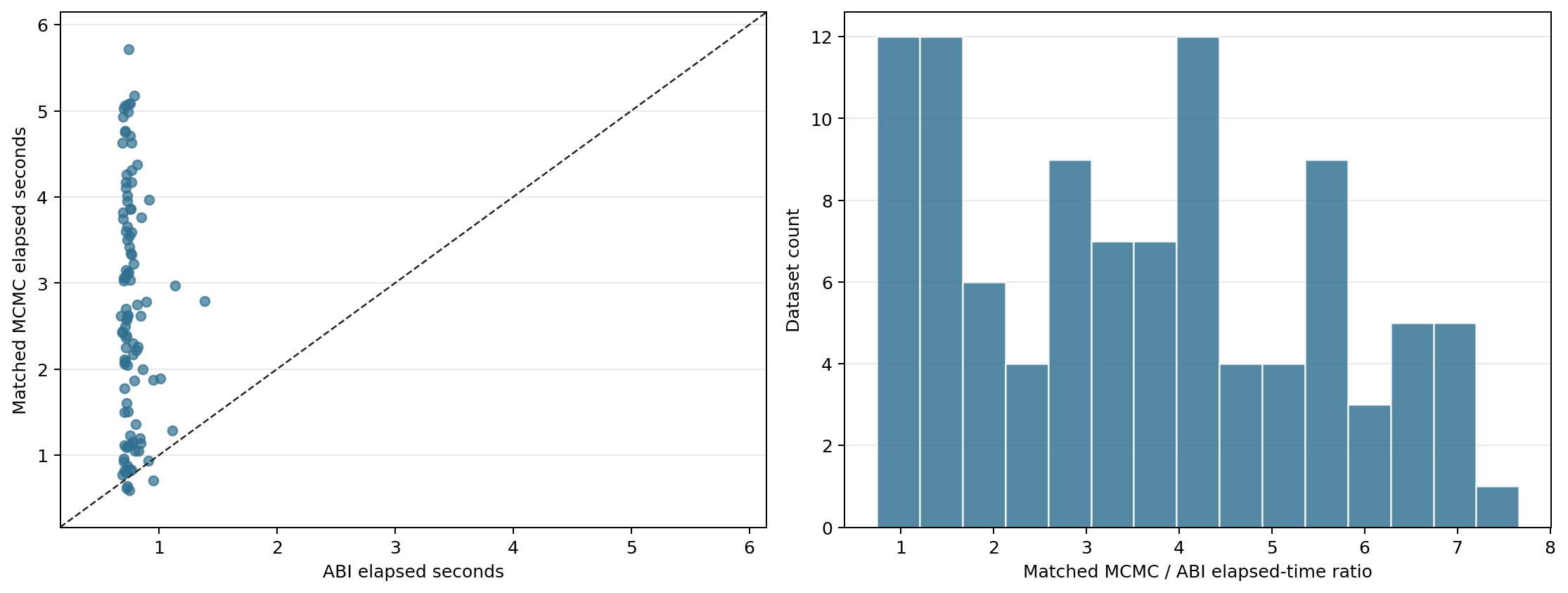}
\caption{Per-dataset posterior-generation runtime in the 100 matched
datasets. Left: matched-MCMC elapsed time against ABI-DAGAR elapsed time; the
dashed line denotes equal runtimes. Right: distribution of the
matched-MCMC-to-ABI elapsed-time ratio. Ratios greater than one indicate
faster ABI-DAGAR inference. The one-time ABI training cost is excluded.}
\label{supp:fig:matched_runtime}
\end{figure}

\FloatBarrier
\subsection{Parameter-posterior replicated-data checks}
\label{supp:ppc}

Parameter-posterior replication was conducted on 40 held-out datasets selected
to span the graph-size range, using $B=100$ replicates per dataset. For each
replicate, the implementation selected a draw of
$(\beta_0,\sigma_w^2,\eta,\rho)$ from the approximate posterior, constructed
the filtered graph, generated a
new centered DAGAR field from its conditional prior using the identity
ordering, and then generated new Poisson counts. It did not sample $\bs w$
from its posterior conditional on the observed counts. These are therefore
parameter-posterior replicated-data diagnostics rather than full latent-field
posterior predictive checks.

For node $i$, let
$[L_i^{\mathrm{rep}},U_i^{\mathrm{rep}}]$ be the central 95\% interval of
its $B$ replicated counts. The node-level count coverage for a dataset is
\[
\frac{1}{N}\sum_{i=1}^{N}
\mathbf 1\!\left\{
L_i^{\mathrm{rep}}\leq y_i\leq U_i^{\mathrm{rep}}
\right\},
\]
and Table~\ref{supp:tab:ppc} summarizes this dataset-level proportion across
the 40 datasets.

For the residual-spatial diagnostic, define
$r_i=\log(y_i+0.5)-\log(e_i)$,
$\bar r=N^{-1}\sum_i r_i$, and
$W_{ij}=a_{ij}/d_i$, where $d_i=\sum_ja_{ij}$, so that $\bs{W}$ is the
row-standardized geographic weight matrix. The implemented Moran-type
statistic is
\[
T_{\mathrm M}(\bs y)
=
\frac{
\sum_{i=1}^{N}(r_i-\bar r)(\bs{W}\bs r)_i
}{
\sum_{i=1}^{N}(r_i-\bar r)^2
}.
\]

For adjacent pair $(i,j)$, let
$z_{ij}=|x_i-x_j|/z_{\mathrm{med}}$, where $z_{\mathrm{med}}$ is the
median covariate dissimilarity over geographic edges. The low-, medium-, and
high-dissimilarity sets are
\[
\mathcal E_{\mathrm L}=\{(i,j):z_{ij}\leq0.75\},\qquad
\mathcal E_{\mathrm M}=\{(i,j):0.75<z_{ij}\leq1.25\},\qquad
\mathcal E_{\mathrm H}=\{(i,j):z_{ij}>1.25\}.
\]
For $k\in\{\mathrm L,\mathrm M,\mathrm H\}$, the corresponding edge-contrast
summary is
\[
T_k(\bs y)
=
\frac{1}{|\mathcal E_k|}
\sum_{(i,j)\in\mathcal E_k}|r_i-r_j|.
\]

For any such summary $T$, the replicated-data probability is
\[
\widehat p_T
=
\frac{1}{B}\sum_{b=1}^{B}
\mathbf 1\!\left\{
T(\bs y^{\mathrm{rep},b})\leq T(\bs y)
\right\}.
\]
Values near zero or one place the observed summary in a tail of its replicated
distribution, whereas values near one-half indicate that it is central in
that distribution. The mean count coverage was close to its nominal level,
and the Moran-type and edge-contrast probabilities did not show a consistent
one-sided displacement across datasets or dissimilarity regimes.
Table~\ref{supp:tab:ppc} reports their means, standard deviations, medians,
and interquartile ranges across datasets.

\begin{table}[htbp!]
\centering
\small
\caption{Additional parameter-posterior replicated-data summaries on held-out simulated datasets.}
\label{supp:tab:ppc}
\begin{tabular}{lcccc}
\hline
\textbf{Metric}
& \textbf{Mean}
& \textbf{SD}
& \textbf{Median}
& \textbf{IQR} \\
\hline
Node-level count coverage (95\% replicated-count interval)
& 0.942 & 0.017 & 0.941 & 0.021 \\
Moran-type replicated-data probability
& 0.493 & 0.197 & 0.495 & 0.225 \\
Low-dissimilarity edge-contrast probability
& 0.532 & 0.262 & 0.545 & 0.410 \\
Mid-dissimilarity edge-contrast probability
& 0.489 & 0.253 & 0.460 & 0.365 \\
High-dissimilarity edge-contrast probability
& 0.555 & 0.239 & 0.620 & 0.288 \\
\hline
\end{tabular}
\end{table}

Fig.~\ref{supp:fig:ppc} displays the dataset-level count coverage and
replicated-data probabilities for the Moran-type and edge-contrast summaries.
Its fourth panel is a separate boundary-count calibration diagnostic based on
all 200 validation datasets and is not one of the parameter-posterior
replicated-data summaries.

\begin{figure}[htbp!]
\centering
\includegraphics[width=\textwidth]{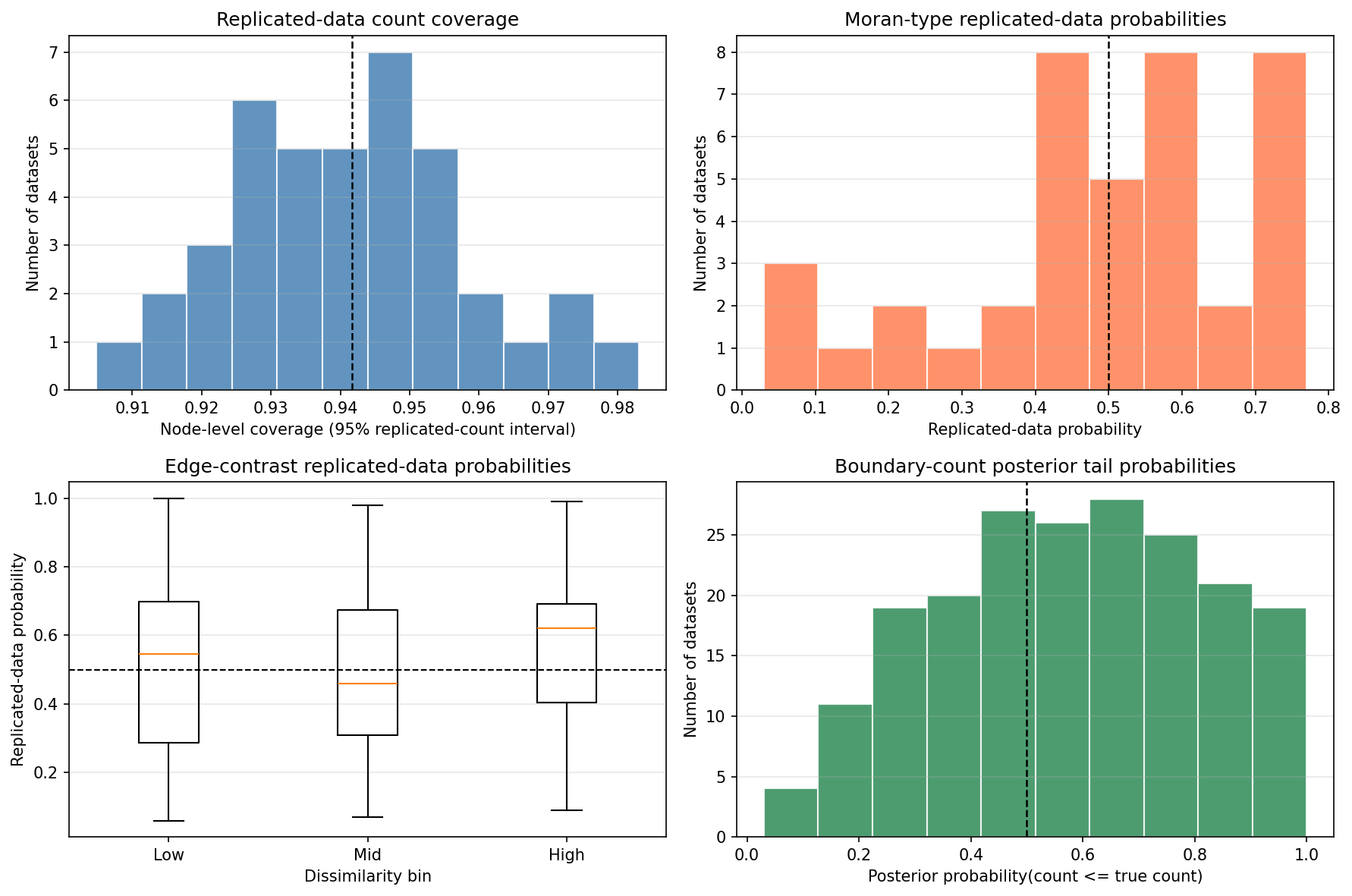}
\caption{Replicated-data and boundary-count diagnostics. Top left: distribution
across 40 datasets of the proportions of observed node counts contained in
their node-specific central 95\% replicated-count intervals; the dashed line
marks the mean across datasets. Top right: Moran-type replicated-data
probabilities across the same 40 datasets, with the dashed line at $0.5$.
Bottom left: replicated-data probabilities for mean absolute adjacent-edge
contrasts in the low-, medium-, and high-dissimilarity bins; boxes summarize
the 40 dataset-specific values and the dashed line is at $0.5$. Each of these
three panels uses 100 replicates per dataset. Bottom right: posterior
probabilities that the boundary count is no greater than the true count,
computed from posterior boundary-count distributions for all 200 validation
datasets; the dashed line is at $0.5$. The bottom-right panel is a
boundary-count calibration diagnostic, not a parameter-posterior
replicated-data summary.}
\label{supp:fig:ppc}
\end{figure}

\FloatBarrier

\subsection{Representation ablation}
\label{supp:ablation}

We compared the full representation with six alternatives: removing the core
observation, graph-topology, dissimilarity, local spatial, or global graph
feature blocks one at a time, and retaining only the core observation
features. Each representation was used to train a separate posterior
approximator from scratch under the same architecture and training protocol.
Evaluation used 4,050 held-out datasets, with 50 datasets at each graph size
$N=40,\ldots,120$.

The full representation yielded an average MAE of
$0.108$, an average RMSE of $0.143$, empirical
95\% interval coverage of $0.941$, and an arithmetic mean of the four parameter-specific Spearman rank correlations
of $0.788$. Recovery was strongest for $\beta_0$, whereas $\eta$ was the most
difficult parameter. Fig.~\ref{supp:fig:ablation_delta} reports changes in
scalar recovery relative to the full-representation baseline.

The core-observation-only representation continued to recover $\beta_0$
accurately, with an MAE of $0.017$ and a correlation of $0.999$, but
substantially degraded inference for parameters requiring spatial
information: correlations fell to $0.172$ for $\eta$ and $0.464$ for
$\rho$. Removing the core observation block produced the largest degradation
among the one-block ablations, whereas removing the individual non-core
blocks generally produced smaller changes. These results indicate that the
core observations are essential and that the graph-aware components provide
complementary, partly redundant information about spatial structure.

At the edge level, the full representation achieved a pooled AUROC of
$0.949$, average precision of $0.792$, and a Brier score of $0.070$,
whereas the core-observation-only representation achieved an AUROC of
$0.907$, average precision of $0.600$, and a Brier score of $0.094$.
Fig.~\ref{supp:fig:ablation_boundary_delta} reports the corresponding changes
in boundary-probability quality. Together, the scalar and edge-level results
show that the graph-aware representation contributes most clearly to
inference on the spatial-structure parameters and posterior boundary
probabilities.

\begin{figure}[htbp!]
\centering
\includegraphics[width=\textwidth]
{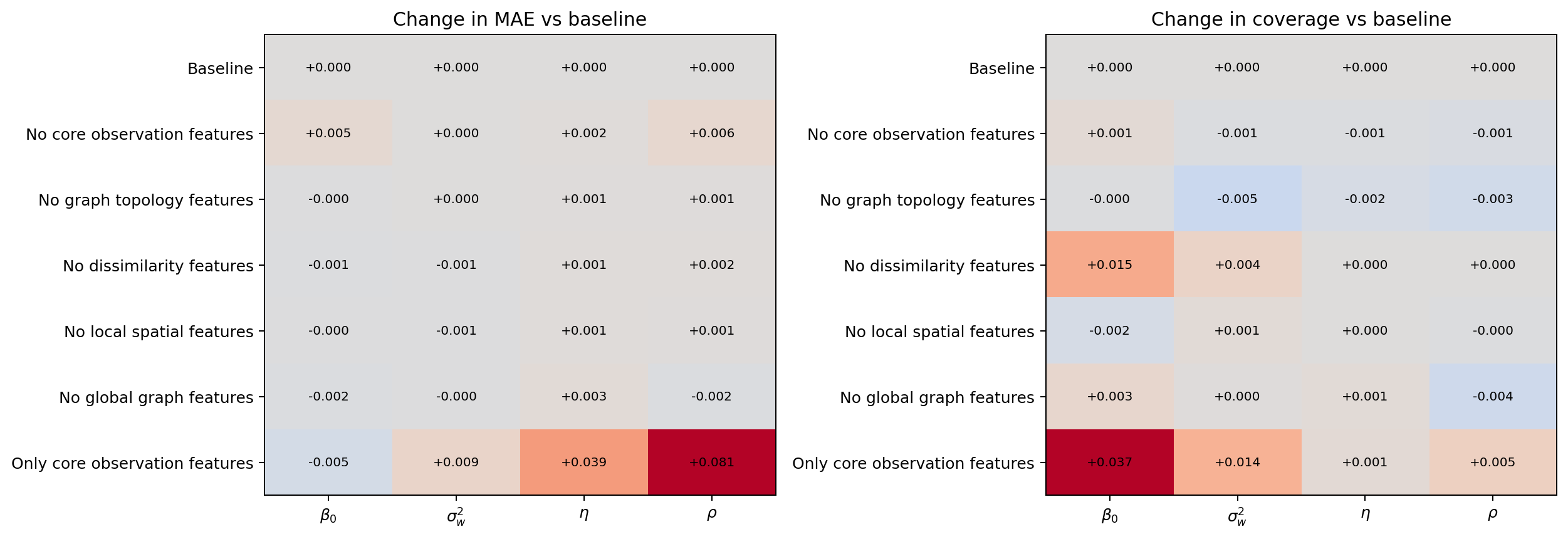}
\caption{Changes in scalar recovery relative to the full-representation
baseline. Every cell is the metric for the listed ablation minus the
corresponding full-representation metric, so the baseline row is zero.
Positive MAE changes indicate deterioration. Coverage changes indicate only
the direction of change relative to the baseline; calibration is improved only
when the resulting coverage is closer to the nominal level of $0.95$, so a
positive change is not necessarily an improvement.}
\label{supp:fig:ablation_delta}
\end{figure}

\begin{figure}[htbp!]
\centering
\includegraphics[width=\textwidth]
{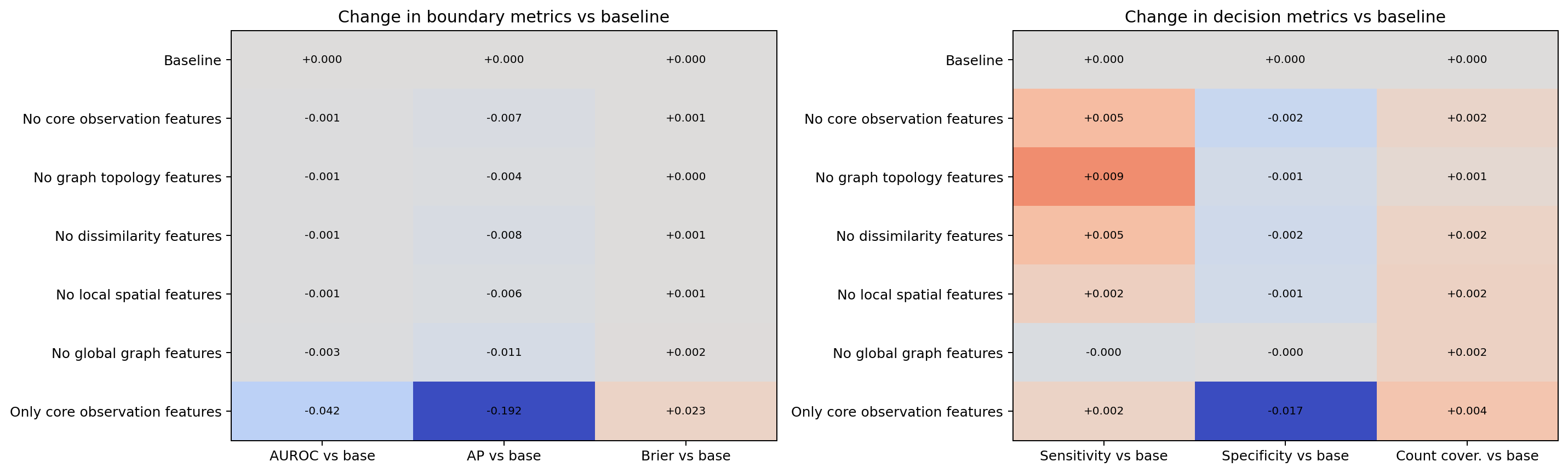}
\caption{Changes in boundary-probability and boundary-decision metrics relative
to the full-representation baseline. Every cell is the metric for the listed
ablation minus the corresponding full-representation metric, so the baseline
row is zero. Negative changes in AUROC, average precision, sensitivity, or
specificity indicate deterioration, whereas positive changes in Brier score
indicate deterioration because lower Brier scores are better. Count-coverage
changes indicate their direction relative to the baseline; improvement depends
on whether the resulting coverage is closer to the nominal level of $0.95$.}
\label{supp:fig:ablation_boundary_delta}
\end{figure}

\FloatBarrier
\subsection{Computational details}
\label{supp:computation}

Table~\ref{supp:tab:computation} summarizes the computational configuration and
cost of the simulation study, including the training setup and hardware, the
size of the held-out validation experiment, posterior-sampling time and memory
use, and the cost of the replicated-data diagnostics.

\begin{table}[htbp!]
\centering
\small
\caption{Computational details for the simulation study.}
\label{supp:tab:computation}
\resizebox{\textwidth}{!}{%
\begin{tabular}{ll}
\hline
\textbf{Quantity} & \textbf{Value} \\
\hline
Hardware
& Intel(R) Core(TM) i7-10750H CPU @ 2.600GHz \\
Training epochs
& 100 \\
Batch size
& 64 \\
Batches per epoch
& 200 \\
Total training simulations
& 1,280,000 \\
Training time
& 5 hours and 45 minutes \\
Held-out validation datasets
& 200 \\
Posterior draws per dataset
& 10,000 \\
Validation sampling time (total)
& 141.10 seconds \\
Validation sampling time (per dataset)
& 0.71 seconds \\
Validation memory
& 0.813 GB \\
Replicated-data datasets / replicates per dataset
& 40 / 100 \\
Replicated-data runtime
& 24.81 seconds \\
\hline
Matched-comparison datasets & 100 \\
Matched ABI / MCMC retained draws per dataset & 10,000 / 10,000 \\
Matched MCMC iterations / burn-in & 20,000 / 10,000 \\
Matched ABI / MCMC total posterior-generation time & 77.15 / 270.75 seconds \\
Matched ABI / MCMC mean time per dataset & 0.77 / 2.71 seconds \\
\end{tabular}
}
\end{table}

\FloatBarrier
\section{Additional spatial health application results}
\label{supp:applications}

This section provides data-processing details, additional diagnostics, and
dataset-specific benchmark results for the three spatial health applications
analyzed in Section~5 of the main manuscript.

\subsection{Empirical data definitions and harmonization}
\label{supp:data_provenance}

The datasets and scientific context are described in Section~2 of the main
manuscript. Here we record the additional preprocessing choices required for
reproduction.

For the Greater Glasgow data, the analysis uses the respiratory-disease data
distributed with \texttt{CARBayesdata}; neighboring areas are defined by queen
contiguity.

The processed Surveillance, Epidemiology, and End Results (SEER) analysis for the California data follows
\citet{gao2023spatial}, with expected counts constructed across the 38
age--sex strata described in the main manuscript; the spatial graph uses queen
contiguity. For six groups of sparsely populated counties, the tobacco report provides
pooled smoking-prevalence estimates; the reported group estimate was assigned
to each constituent county.

Reported South Korean parent-city smoking values were used where a
separate municipality estimate was unavailable. Municipal geometries use the
May 2019 boundary system, and the analysis retains the largest connected
component, excluding nine geographically disconnected units; queen contiguity
defines neighboring municipalities.

Across applications, counts, expected counts, boundary-driving covariates, and
adjacency matrices were aligned to a common analysis structure. The analysis
variables contain no missing values, and the boundary-driving covariates were
standardized identically in the Python and R workflows using their empirical
means and population standard deviations.

\subsection{Exploratory compatibility with the shared model class}
\label{supp:application_exploration}

We additionally assessed whether the empirical inputs occupy the deployment
regime represented during training. Fig.~\ref{supp:fig:deployment_support}
compares the three applications with input configurations generated from the
graph, covariate, and exposure components of the training simulator.

Relative to fixed-$N$ Delaunay reference graphs, all three empirical
contiguity graphs have fewer edges and lower mean degree, although their upper
degree quantiles remain supported. The empirical boundary-driving covariates
are also more spatially smooth than the iid covariates used during training:
their median neighboring dissimilarities are smaller, their graph-specific
scaling factors $M$ are larger, and their Moran coefficients lie above the
central simulated range. Median log-exposure is supported for California and
South Korea, whereas Greater Glasgow lies below the corresponding fixed-$N$
interval.

For Greater Glasgow, California, and South Korea, respectively, the median
standardized neighboring dissimilarities are $0.599$, $0.704$, and $0.698$,
yielding graph-specific scaling factors $M=1.158$, $0.985$, and $0.993$.

These departures provide a realistic test of transfer beyond the simulator's
most typical configurations; the dataset-specific benchmarks in
Section~\ref{supp:real_mcmc} assess whether the principal empirical conclusions
remain stable.

\begin{figure}[t]
\centering
\includegraphics[width=\textwidth]
{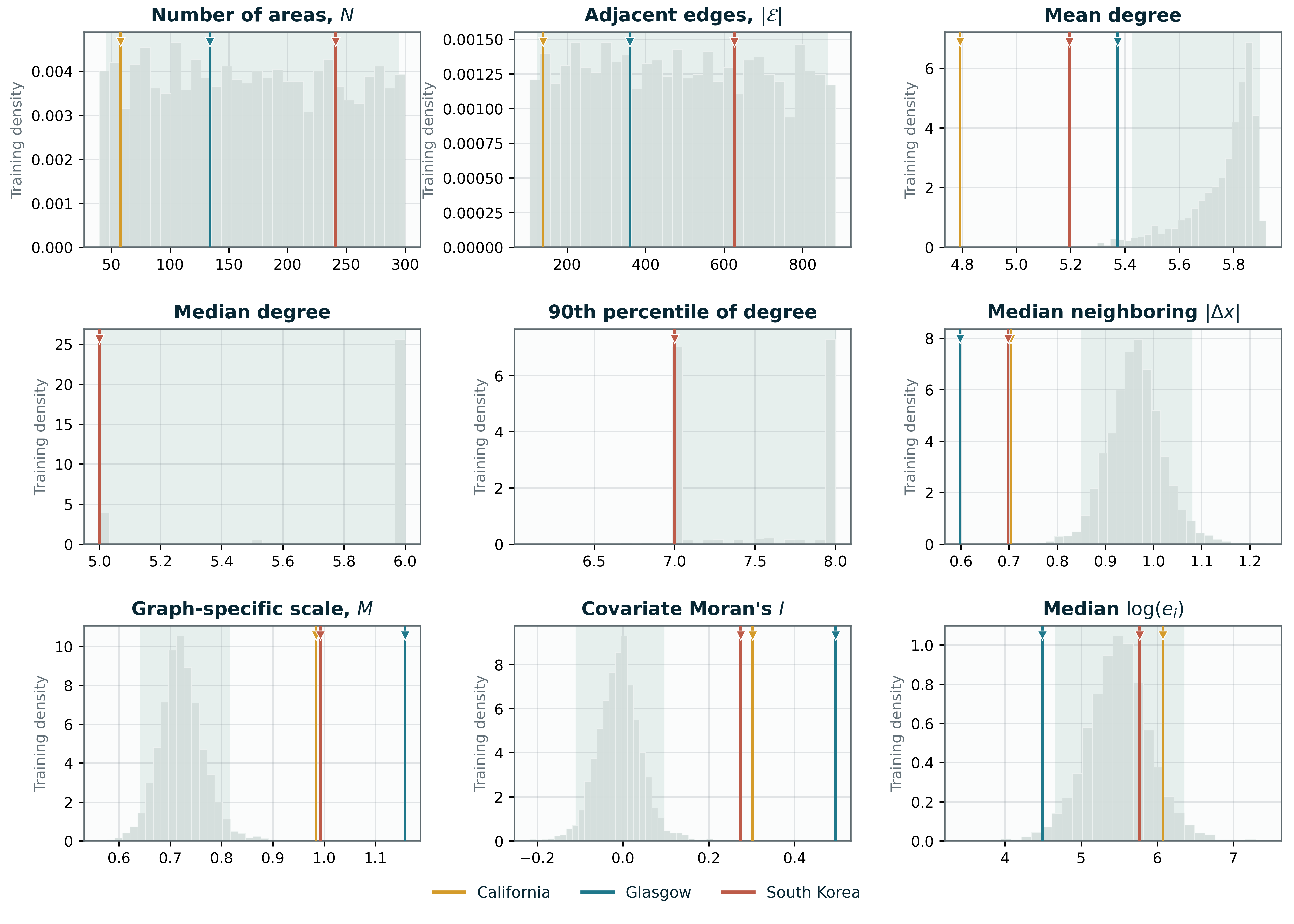}
\caption{Empirical inputs relative to the original ABI-DAGAR training
distribution. Histograms represent 3,000 input configurations generated from
the graph, covariate, and exposure components of the original simulator, and
shaded regions denote central 95\% intervals. Colored lines and triangles mark
the three empirical applications. Although all applications fall within the
explicit graph-size range, their contiguity graphs are sparser and their
covariates are smoother than the iid-covariate Delaunay training
configurations.}
\label{supp:fig:deployment_support}
\end{figure}

Fig.~\ref{supp:fig:observed_expected} compares observed and expected counts.
The scales differ substantially across applications and areas, but the
offset-based representation remains well defined throughout each graph. The
figure also shows why modeling counts with an explicit expected-count offset is
preferable to comparing unadjusted counts across spatial units with very
different underlying population sizes.

For a vector $\bs v=(v_1,\ldots,v_N)^\top$, Moran's $I$ was computed as
\[
I(\bs v)
=
\frac{N}{S_0}
\frac{
\sum_{i=1}^{N}\sum_{j=1}^{N}
a_{ij}(v_i-\bar v)(v_j-\bar v)
}{
\sum_{i=1}^{N}(v_i-\bar v)^2
},
\qquad
S_0=\sum_{i=1}^{N}\sum_{j=1}^{N}a_{ij},
\]
using the symmetric binary geographic adjacency matrix $A$. This calculation
was applied separately to the crude log-risk vector
$r_i=\log(y_i+0.5)-\log(e_i)$ and to the boundary-driving covariate.

Table~\ref{supp:tab:application_exploration} reports the numerical summaries
underlying the exploratory diagnostics in Section~2 of the main manuscript.

\begin{table}[htbp!]
\centering
\small
\caption{Exploratory characteristics of the three spatial health applications.
$I_r$ and $I_x$ are Moran's $I$ for crude log risk and the boundary-driving
covariate. ``Adj.'' is the mean absolute crude log-risk contrast over adjacent
pairs, and ``Non-neigh.'' is the corresponding mean for an equally sized random
sample of non-neighboring pairs. Q5/Q1 is the ratio of mean adjacent risk
contrasts in the highest and lowest covariate-dissimilarity quintiles.}
\label{supp:tab:application_exploration}
\resizebox{\textwidth}{!}{%
\begin{tabular}{lccccccccc}
\toprule
\textbf{Application} & $\bs{N}$ & \textbf{Edges} & \textbf{Mean degree}
& \textbf{Standardized risk, median [IQR]} & $\bs{I_r}$ & $\bs{I_x}$
& \textbf{Adj.} & \textbf{Non-neigh.} & \textbf{Q5/Q1} \\
\midrule
Greater Glasgow & 134 & 360 & 5.370 & 0.846 [0.593, 1.080] & 0.415 & 0.495 & 0.346 & 0.477 & 2.980 \\
California & 58 & 139 & 4.790 & 1.060 [0.942, 1.200] & 0.283 & 0.302 & 0.180 & 0.235 & 1.300 \\
South Korea & 241 & 626 & 5.200 & 1.030 [0.945, 1.120] & 0.466 & 0.274 & 0.105 & 0.140 & 1.380 \\
\bottomrule
\end{tabular}%
}
\end{table}

Fig.~\ref{supp:fig:covariate_risk} provides a separate descriptive view of the
relationship between the boundary-driving covariate and crude log risk.
Spearman correlations are $0.849$, $0.602$, and $0.428$ in Greater Glasgow,
California, and South Korea, respectively. These associations help establish
the scientific relevance of deprivation and smoking contrasts, but they are not
interpreted as causal effects and the covariates are not used as regressors in
the mean structure. Their role in the fitted model is to identify borders at
which spatial borrowing may be inappropriate.

\begin{figure}[htbp!]
\centering
\includegraphics[width=\textwidth]{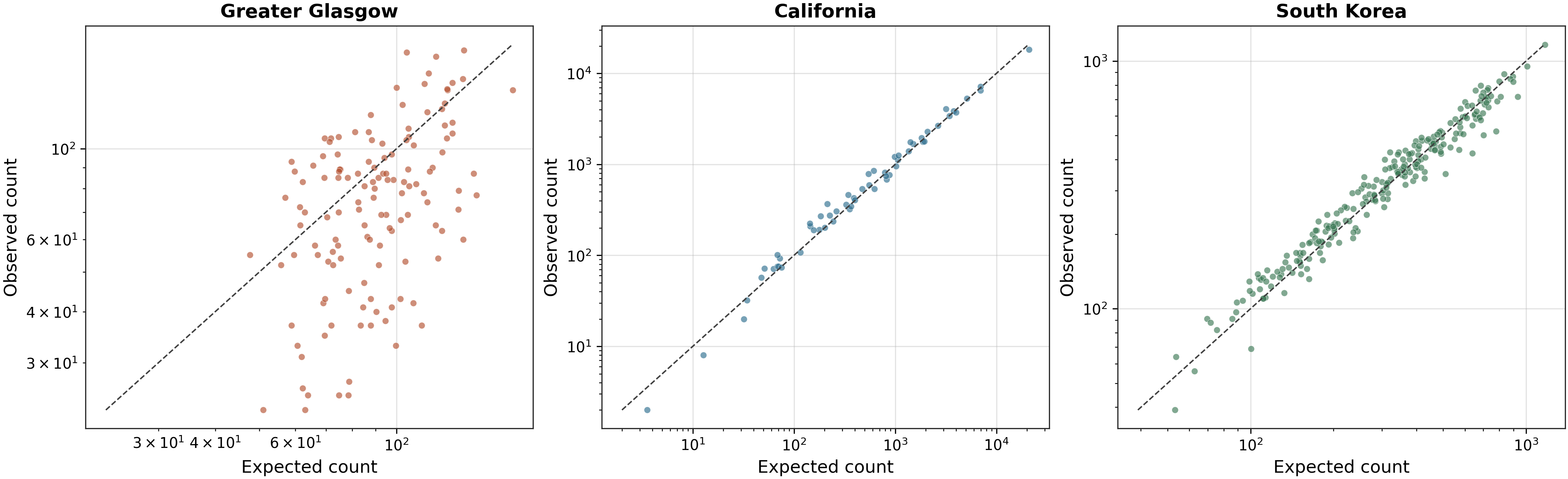}
\caption{Observed and expected counts in the three spatial health
applications. Both axes are logarithmic, and the dashed line is the equality
line. Left to right: Greater Glasgow, California, and South Korea.}
\label{supp:fig:observed_expected}
\end{figure}

\begin{figure}[htbp!]
\centering
\includegraphics[width=\textwidth]{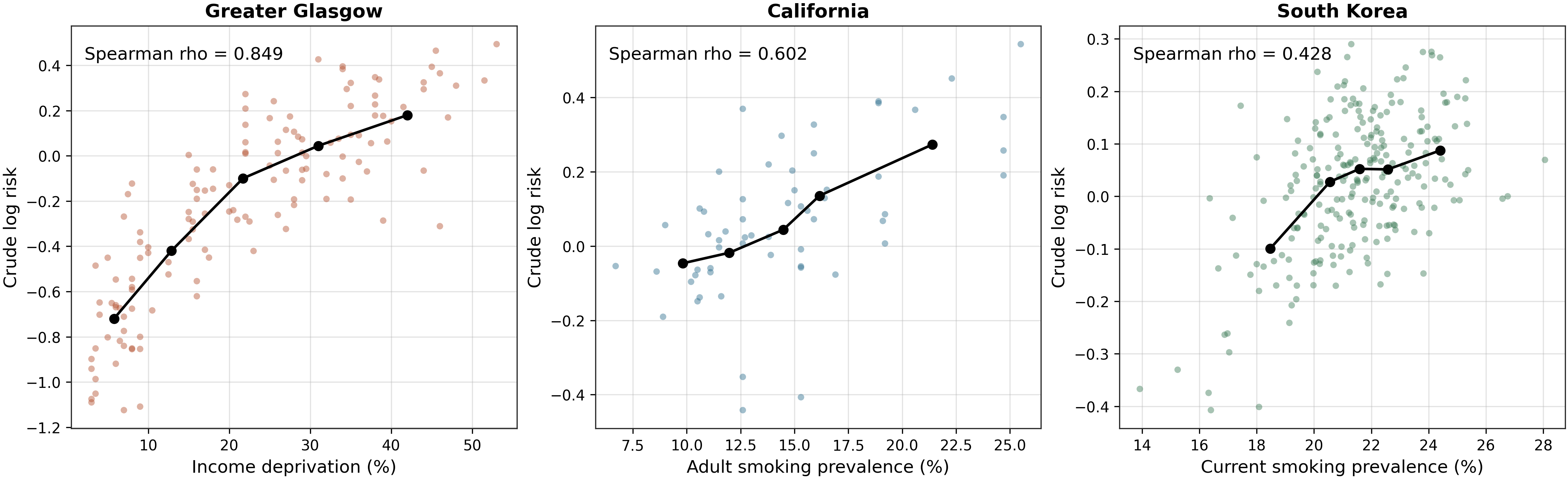}
\caption{Descriptive relationship between the boundary-driving covariate and
crude log risk. Points represent areas and the black line joins within-quintile
means. Left to right: Greater Glasgow, California, and South Korea.}
\label{supp:fig:covariate_risk}
\end{figure}

\FloatBarrier
\subsection{Geographic reference for selected boundary neighborhoods}
\label{supp:boundary_geography}

The application discussion in the main manuscript names several boundaries
with high ABI-DAGAR posterior probability. Fig.~\ref{supp:fig:boundary_geography}
places six high-probability boundary neighborhoods from each application in
their full geographic context, and Table~\ref{supp:tab:named_boundaries}
identifies the corresponding adjacent areas. To avoid presenting multiple
nearly identical edges around one area, the displayed examples were selected in
descending posterior-probability order subject to no area appearing in more
than two selected pairs. They are descriptive geographic references rather than
a separate inferential selection rule.

\begin{table}[htbp!]
\centering
\small
\setlength{\tabcolsep}{3pt}
\caption{Named high-probability ABI-DAGAR boundary neighborhoods displayed in
Fig.~\ref{supp:fig:boundary_geography}. $|\Delta x|$ is the raw neighboring
covariate difference in percentage points. The final column gives the crude
standardized risks on the two sides of the border: standardized hospitalization ratio (SHR) for Greater Glasgow, standardized incidence ratio (SIR)
for California, and standardized mortality ratio (SMR) for South Korea.}
\label{supp:tab:named_boundaries}
\begin{tabular}{llp{6.4cm}ccc}
\toprule
\textbf{Application} & \textbf{No.} & \textbf{Neighboring areas}
& $\bs{p_{ij}}$ & $\bs{|\Delta x|}$ & \textbf{Risk pair} \\
\midrule
Greater Glasgow & 1 & South Castlehill and Thorn--Drumchapel South & 0.956 & 44.000 & 0.331 / 1.178 \\
& 2 & Garrowhill East and Swinton--North Barlanark and Easterhouse South & 0.956 & 43.500 & 0.554 / 1.387 \\
& 3 & South Castlehill and Thorn--Drumry East & 0.955 & 42.500 & 0.331 / 1.582 \\
& 4 & Westerton West--Drumchapel South & 0.953 & 41.000 & 0.391 / 1.178 \\
& 5 & Kelvindale--Wyndford & 0.951 & 39.500 & 0.437 / 0.727 \\
& 6 & Garrowhill West--Barlanark & 0.945 & 36.000 & 0.420 / 0.931 \\
\midrule
California & 1 & Lake--Yolo & 0.743 & 14.000 & 1.719 / 1.014 \\
& 2 & Placer--Yuba & 0.731 & 13.300 & 1.057 / 1.566 \\
& 3 & Lake--Sonoma & 0.722 & 12.800 & 1.719 / 1.023 \\
& 4 & Del Norte--Humboldt & 0.669 & 10.300 & 1.286 / 1.344 \\
& 5 & Humboldt--Siskiyou & 0.669 & 10.300 & 1.344 / 1.206 \\
& 6 & Mendocino--Trinity & 0.652 & 9.700 & 1.160 / 1.405 \\
\midrule
South Korea & 1 & Sujeong-gu, Seongnam-si--Gwacheon-si & 0.826 & 11.160 & 0.991 / 0.688 \\
& 2 & Sujeong-gu, Seongnam-si--Bundang-gu, Seongnam-si & 0.804 & 9.840 & 0.991 / 0.718 \\
& 3 & Jungwon-gu, Seongnam-si--Bundang-gu, Seongnam-si & 0.796 & 9.400 & 1.195 / 0.718 \\
& 4 & Michuhol-gu--Yeonsu-gu & 0.781 & 8.660 & 1.071 / 0.842 \\
& 5 & Gwanak-gu--Gwacheon-si & 0.763 & 7.920 & 0.889 / 0.688 \\
& 6 & Cheoin-gu, Yongin-si--Suji-gu, Yongin-si & 0.755 & 7.640 & 1.109 / 0.687 \\
\bottomrule
\end{tabular}
\end{table}

Application-specific local zoom atlases for the six displayed boundary
neighborhoods are shown in
Figs.~\ref{supp:fig:glasgow_boundary_atlas}--\ref{supp:fig:south_korea_boundary_atlas}.
The atlases and crosswalks are generated by the exploratory-analysis notebook.

\begin{figure}[htbp!]
\centering
\includegraphics[width=\textwidth]{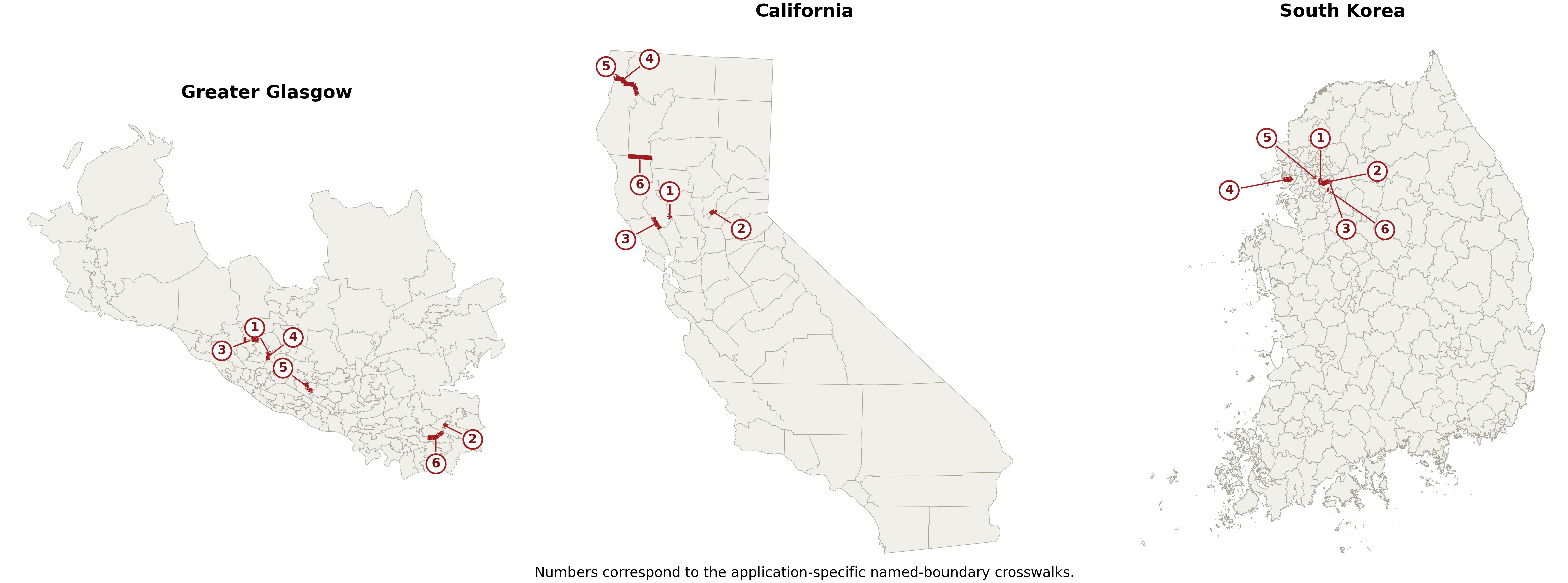}
\caption{Geographic locations of selected high-probability ABI-DAGAR boundary
neighborhoods. Red segments identify the selected borders and numbered labels
refer to Table~\ref{supp:tab:named_boundaries}. Left to right: Greater Glasgow, 
California, and South Korea.}
\label{supp:fig:boundary_geography}
\end{figure}

\begin{landscape}
\begin{figure}[p]
\centering
\includegraphics[width=0.96\linewidth]{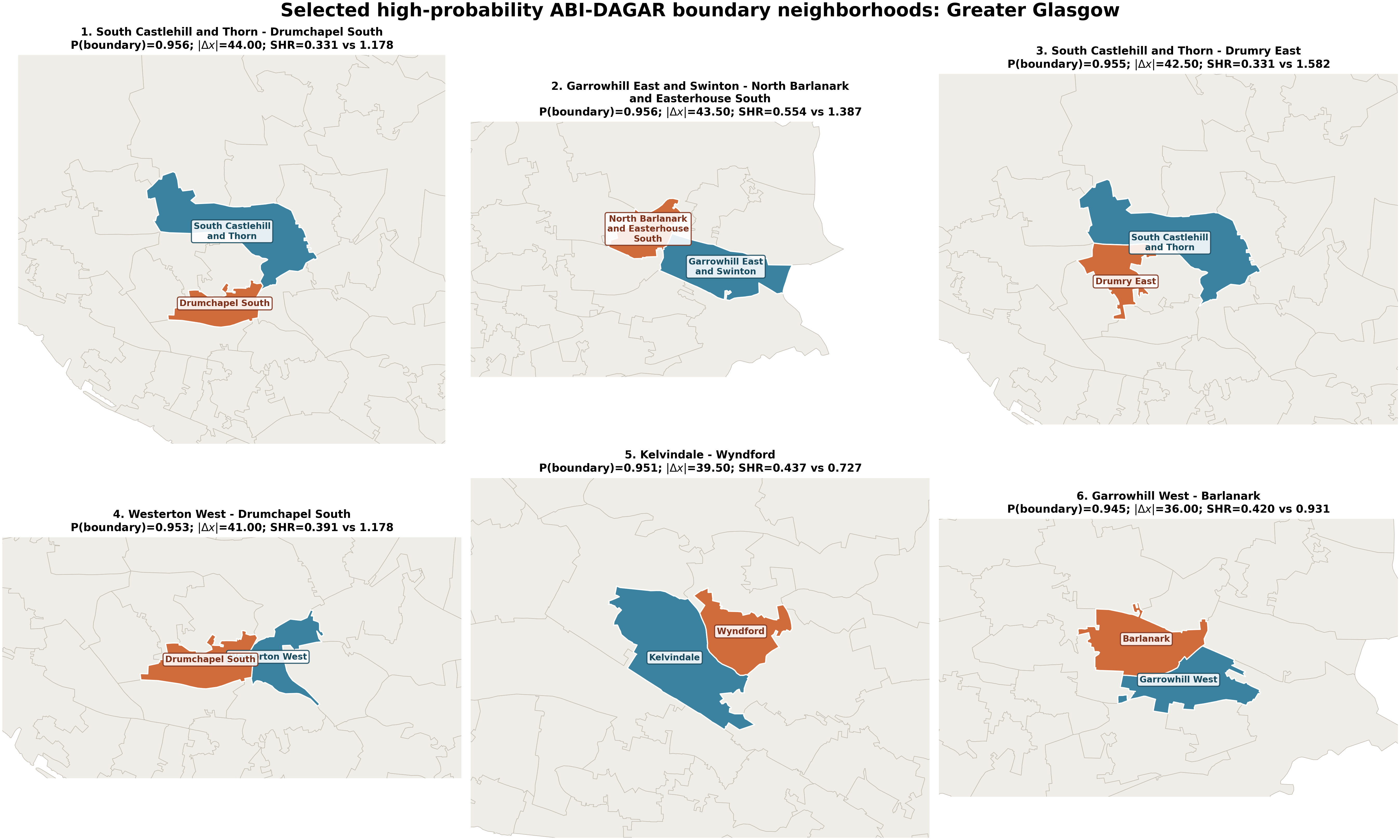}
\caption{Local zoom atlas for the six Greater Glasgow boundary neighborhoods
listed in Table~\ref{supp:tab:named_boundaries}. Blue and orange identify the
two neighboring areas.}
\label{supp:fig:glasgow_boundary_atlas}
\end{figure}
\end{landscape}

\begin{landscape}
\begin{figure}[p]
\centering
\includegraphics[width=0.87\linewidth]{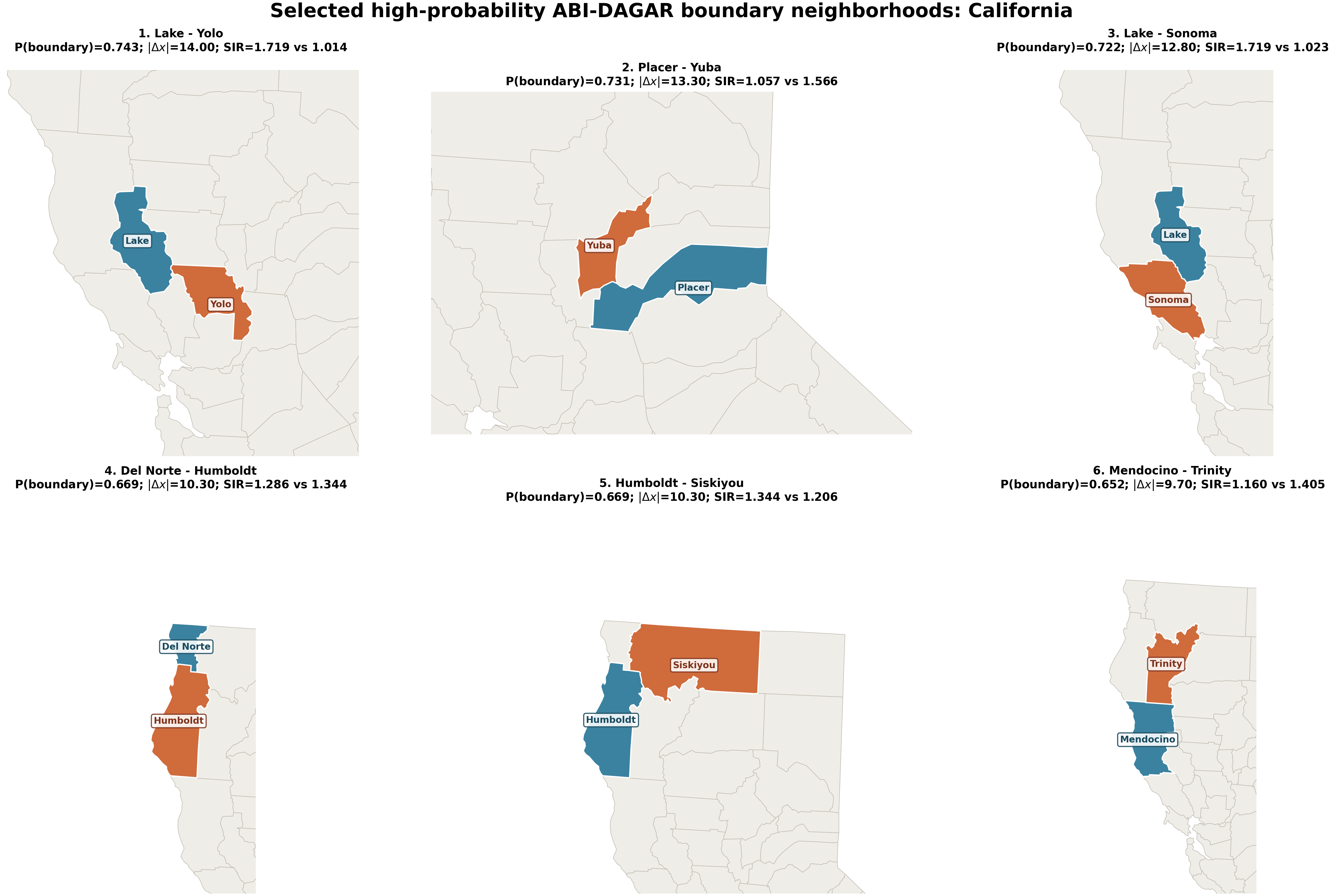}
\caption{Local zoom atlas for the six California boundary neighborhoods listed
in Table~\ref{supp:tab:named_boundaries}. Blue and orange identify the two
neighboring counties.}
\label{supp:fig:california_boundary_atlas}
\end{figure}
\end{landscape}

\begin{landscape}
\begin{figure}[p]
\centering
\includegraphics[width=0.95\linewidth]{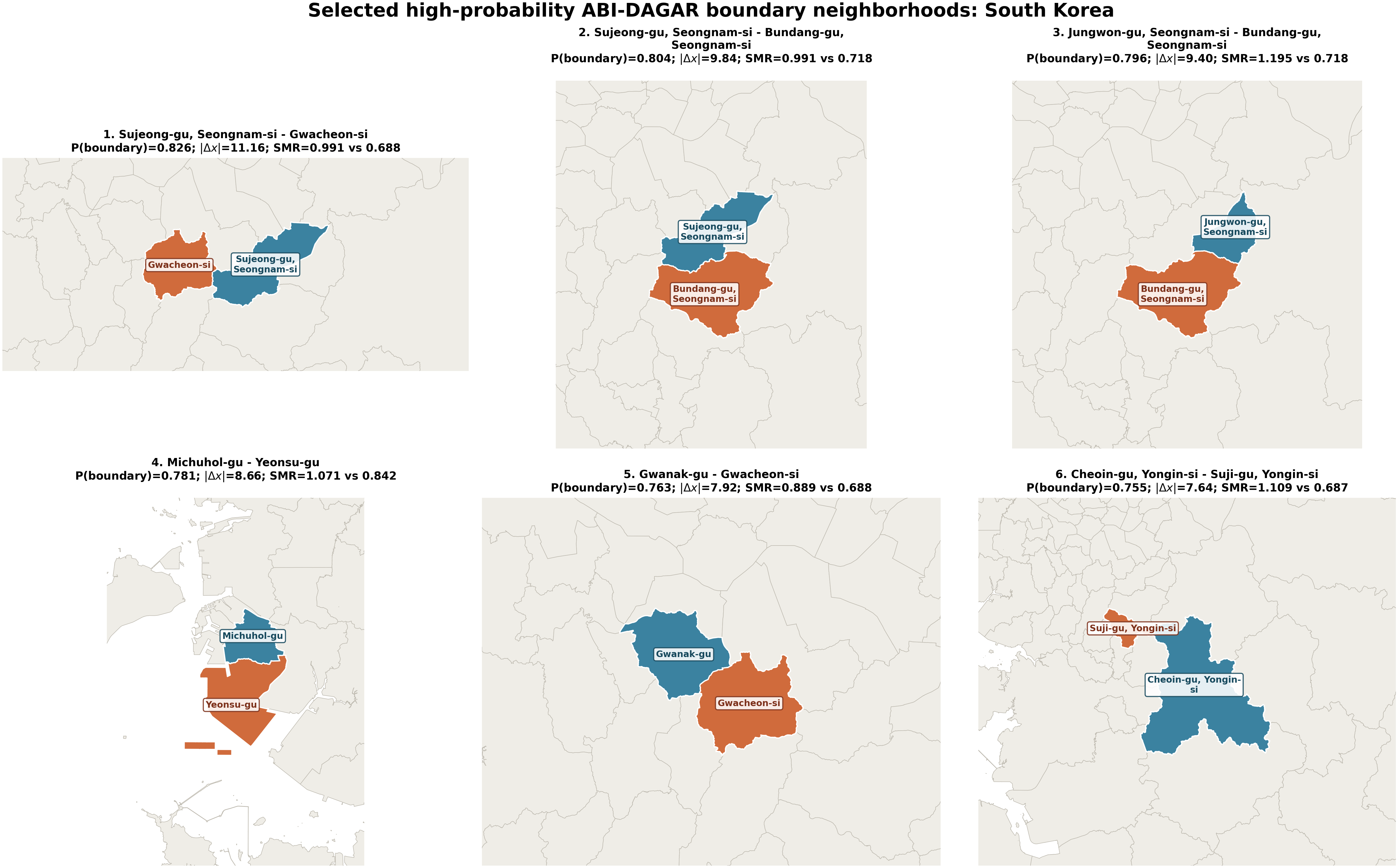}
\caption{Local zoom atlas for the six South Korean boundary neighborhoods
listed in Table~\ref{supp:tab:named_boundaries}. Blue and orange identify 
the two neighboring administrative units.}
\label{supp:fig:south_korea_boundary_atlas}
\end{figure}
\end{landscape}

\FloatBarrier
\subsection{Additional fitted-risk diagnostics}
\label{supp:risk}

Because ABI-DAGAR targets posterior inference for the scalar parameters rather
than direct posterior sampling of $\bs w$, the fitted-risk and count
comparisons below are secondary descriptive diagnostics. The fitted-risk table
and the displayed count reconstruction are produced by two related but
distinct reconstruction routines. Neither routine samples $\bs w$ from its
posterior conditional on the observed counts, and neither uses
$\sigma_w^2$ in reconstructing the latent field. They should therefore not be
interpreted as full posterior predictive summaries.
Fig.~\ref{supp:fig:real_ppc} displays the descriptive count reconstructions.

\begin{figure}[htbp!]
\centering
\includegraphics[width=0.32\linewidth]
{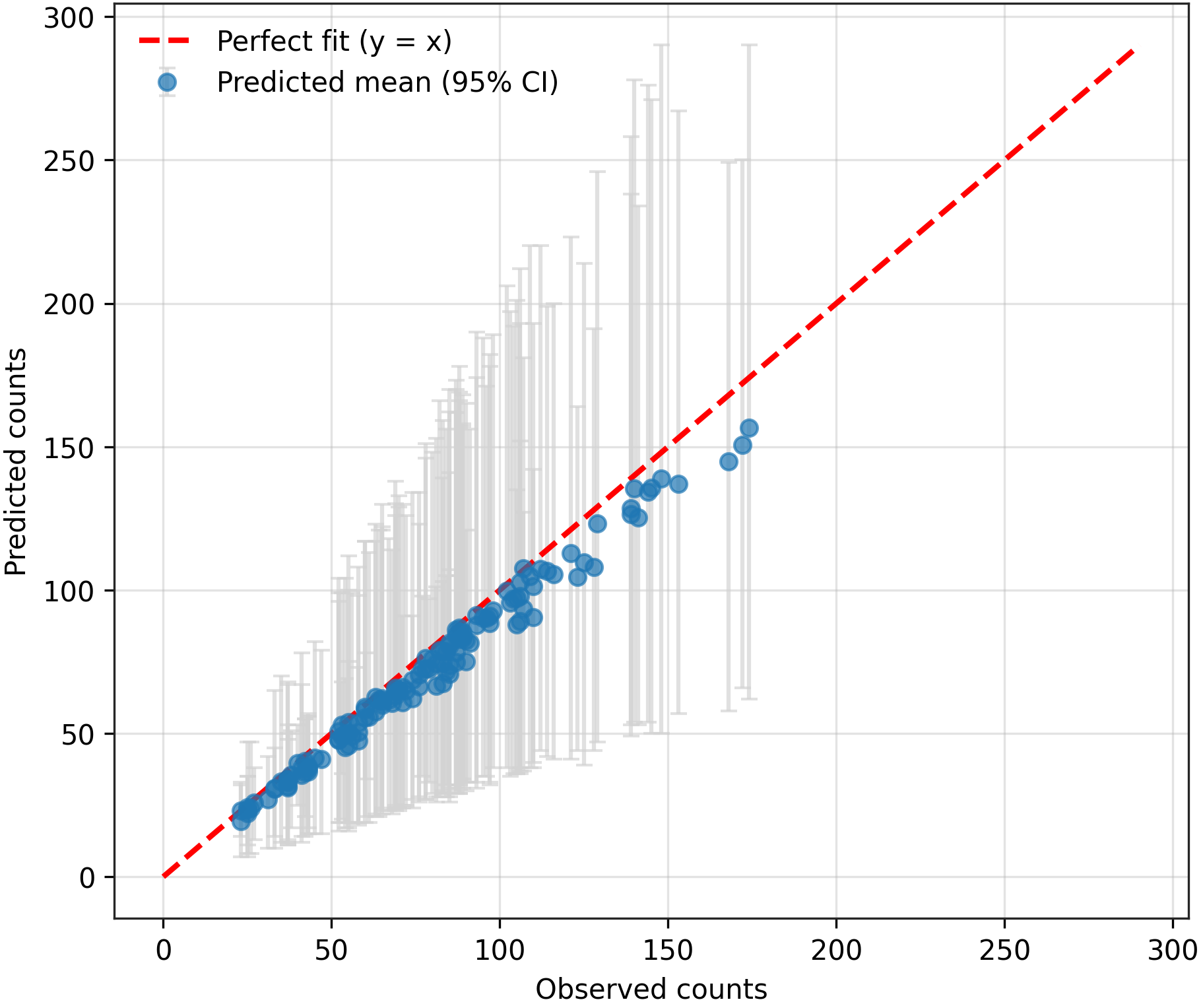}
\includegraphics[width=0.32\linewidth]
{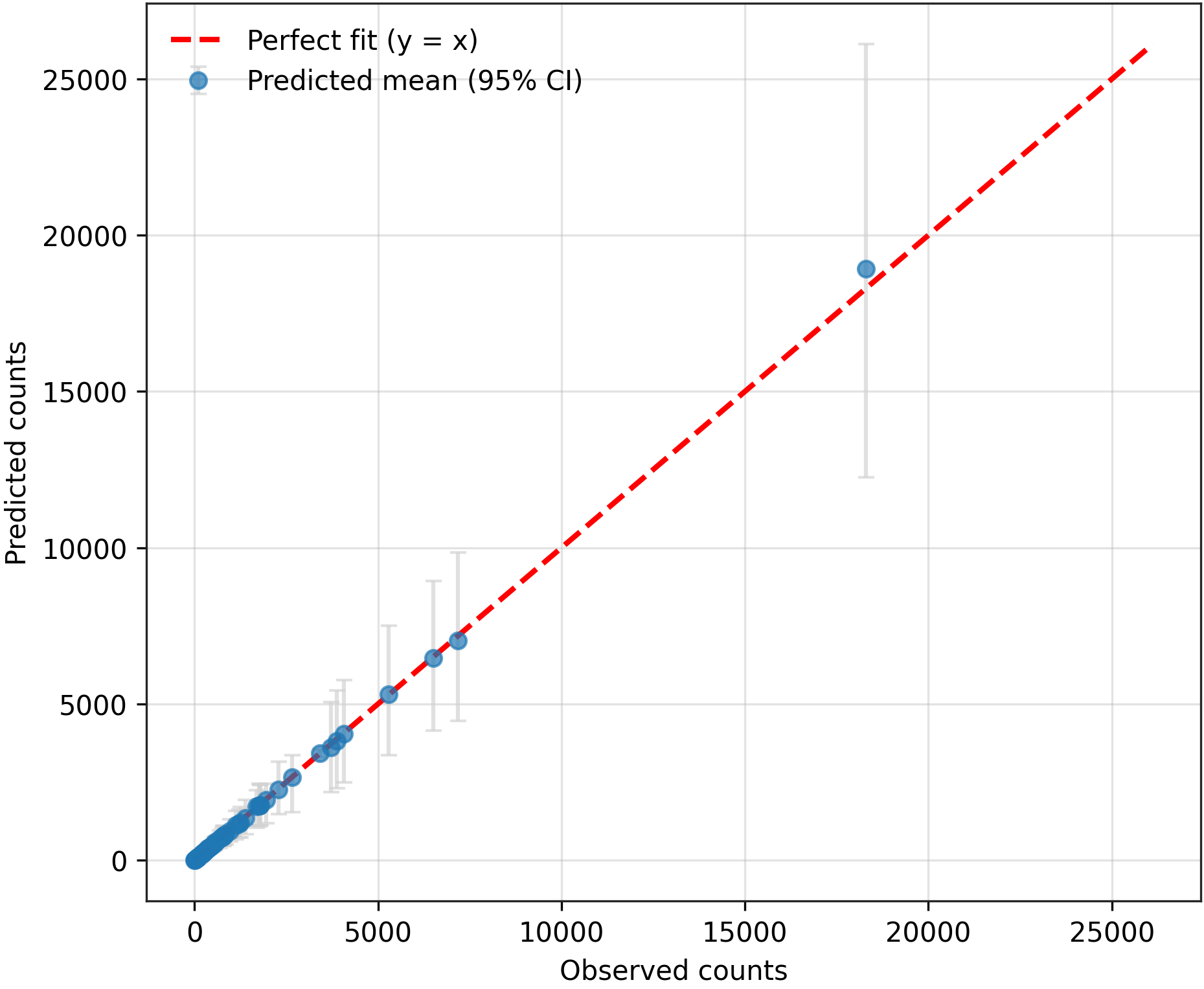}
\includegraphics[width=0.32\linewidth]
{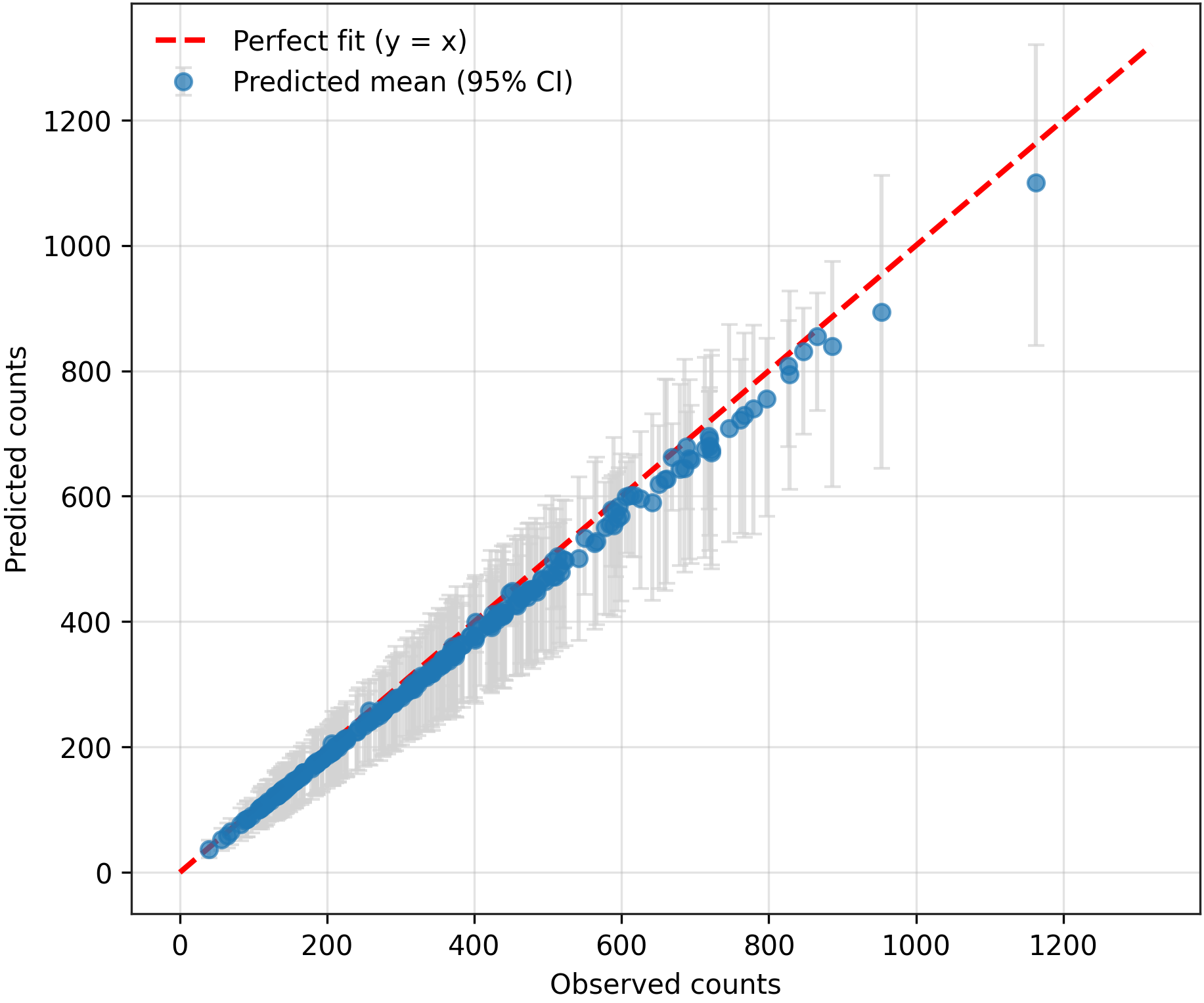}
\caption{Observed counts and descriptive count reconstructions.
Points are the predicted mean and vertical bars the corresponding 95\% credible intervals.
Left to right: Greater Glasgow, California, and South Korea.}
\label{supp:fig:real_ppc}
\end{figure}

For the fitted-risk comparison, let $S=1000$ and index the ABI posterior draws
by $s$. For each draw, the implementation constructs the filtered adjacency
matrix from $\eta^{(s)}$, and uses the
identity DAGAR ordering. It then forms
\[
u_i^{(s)}
=
\log(y_i+0.5)-\log(e_i)-\beta_0^{(s)}
\]
and centers $\bs u^{(s)}$ over areas. If
$L^{(s)}=I-B\{\widetilde A^{(s)},\rho^{(s)}\}$ is the resulting DAGAR
triangular factor, the reconstructed field solves
\[
L^{(s)}\widetilde{\bs w}^{(s)}
=
\bs u^{(s)}-\bar u^{(s)}\bs 1,
\]
after which $\widetilde{\bs w}^{(s)}$ is centered. The ABI reconstructed risk
for area $i$ is
\[
\widehat R_{A,i}
=
\frac{1}{S}\sum_{s=1}^{S}
\exp\{\beta_0^{(s)}+\widetilde w_i^{(s)}\}.
\]
The \texttt{CARBayes} fitted risk is
$\widehat R_{C,i}=\texttt{fitted.values}_i/e_i$, using the fitted counts
returned by \texttt{CARBayes}.

The reported correlation is the Pearson correlation between
$\{\widehat R_{A,i}\}_{i=1}^{N}$ and
$\{\widehat R_{C,i}\}_{i=1}^{N}$. The area-level discrepancies are
\[
\operatorname{MAE}
=
\frac{1}{N}\sum_{i=1}^{N}
|\widehat R_{A,i}-\widehat R_{C,i}|,
\qquad
\operatorname{RMSE}
=
\left[
\frac{1}{N}\sum_{i=1}^{N}
(\widehat R_{A,i}-\widehat R_{C,i})^2
\right]^{1/2}.
\]
The reported regression has ABI reconstructed risk as the response and
\texttt{CARBayes} fitted risk as the predictor,
$\widehat R_{A,i}=\alpha+\beta\widehat R_{C,i}+\varepsilon_i$; the table
reports $\beta$ as the slope and $\alpha$ as the intercept.
Table~\ref{supp:tab:risk} shows positive association in all three datasets,
with the strongest association in South Korea. Because these risks are
reconstructed rather than sampled directly by the ABI posterior, edge-level
boundary probabilities provide the more direct benchmark for the primary
inferential target. The area-level comparisons are displayed in
Fig.~\ref{supp:fig:risk}.

\begin{table}[htbp!]
\centering
\caption{Additional fitted-risk comparison metrics between ABI-DAGAR and
\texttt{CARBayes}.}
\label{supp:tab:risk}
\resizebox{\textwidth}{!}{%
\begin{tabular}{lccccc}
\hline
\textbf{Dataset}
& \textbf{Correlation}
& \textbf{MAE}
& \textbf{RMSE}
& \textbf{ABI on \texttt{CARBayes} slope}
& \textbf{ABI on \texttt{CARBayes} intercept} \\
\hline
Glasgow
& 0.769 & 0.729 & 1.094 & 3.046 & $-1.245$ \\
California
& 0.840 & 0.146 & 0.230 & 1.699 & $-0.751$ \\
South Korea
& 0.865 & 0.146 & 0.190 & 2.315 & $-1.310$ \\
\hline
\end{tabular}
}
\end{table}

\begin{figure}[htbp!]
\centering
\includegraphics[width=0.32\linewidth]{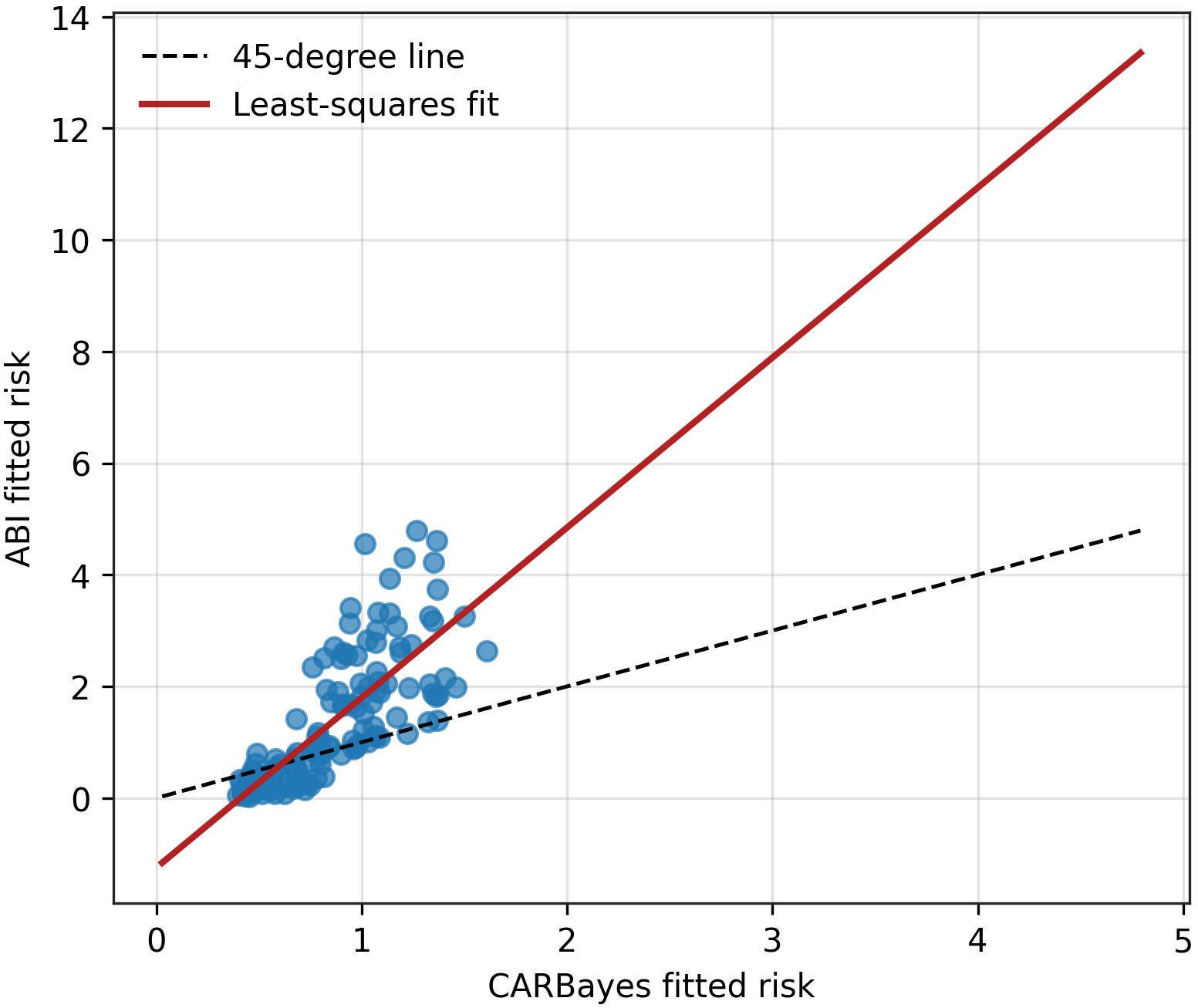}
\includegraphics[width=0.32\linewidth]{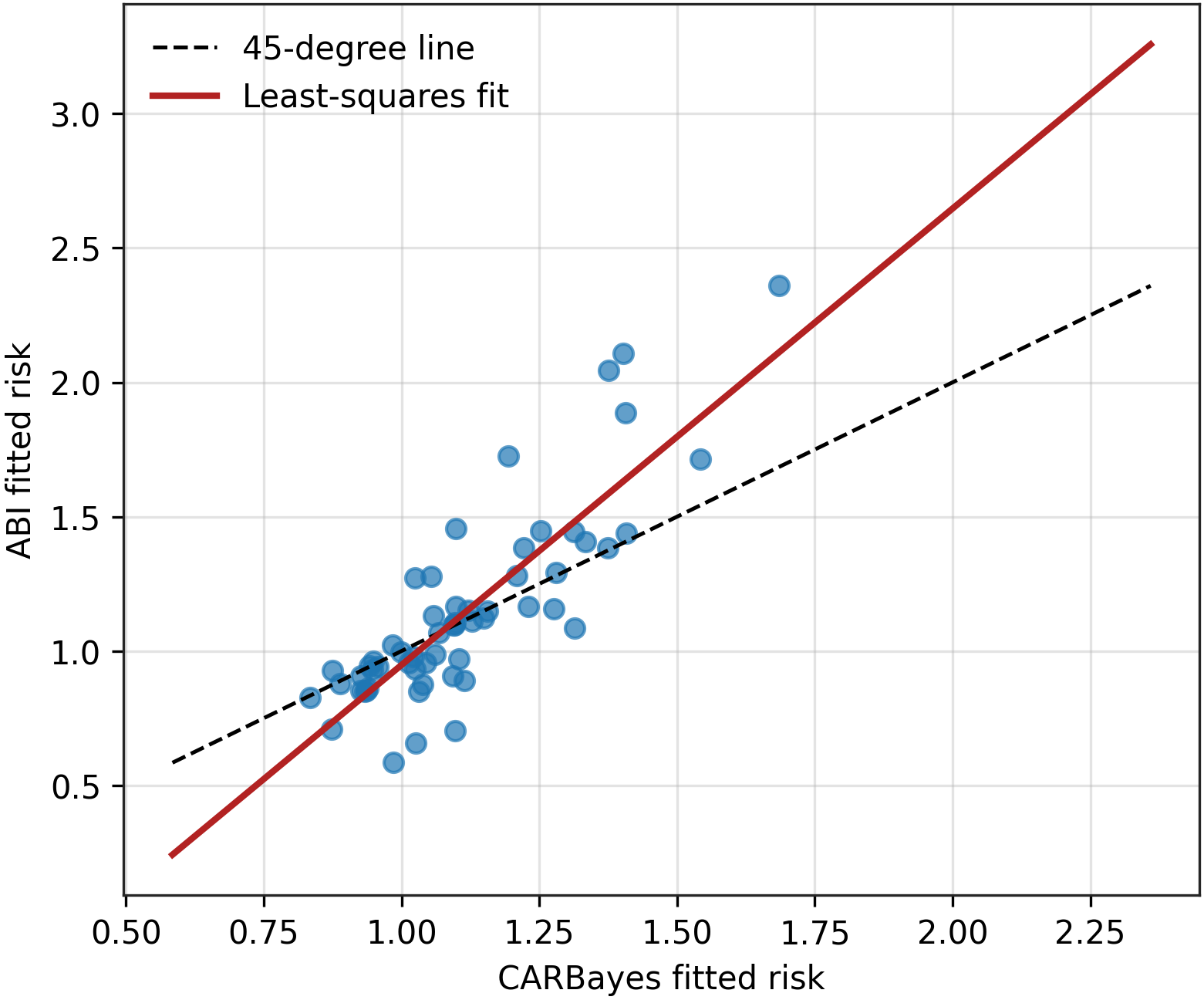}
\includegraphics[width=0.32\linewidth]{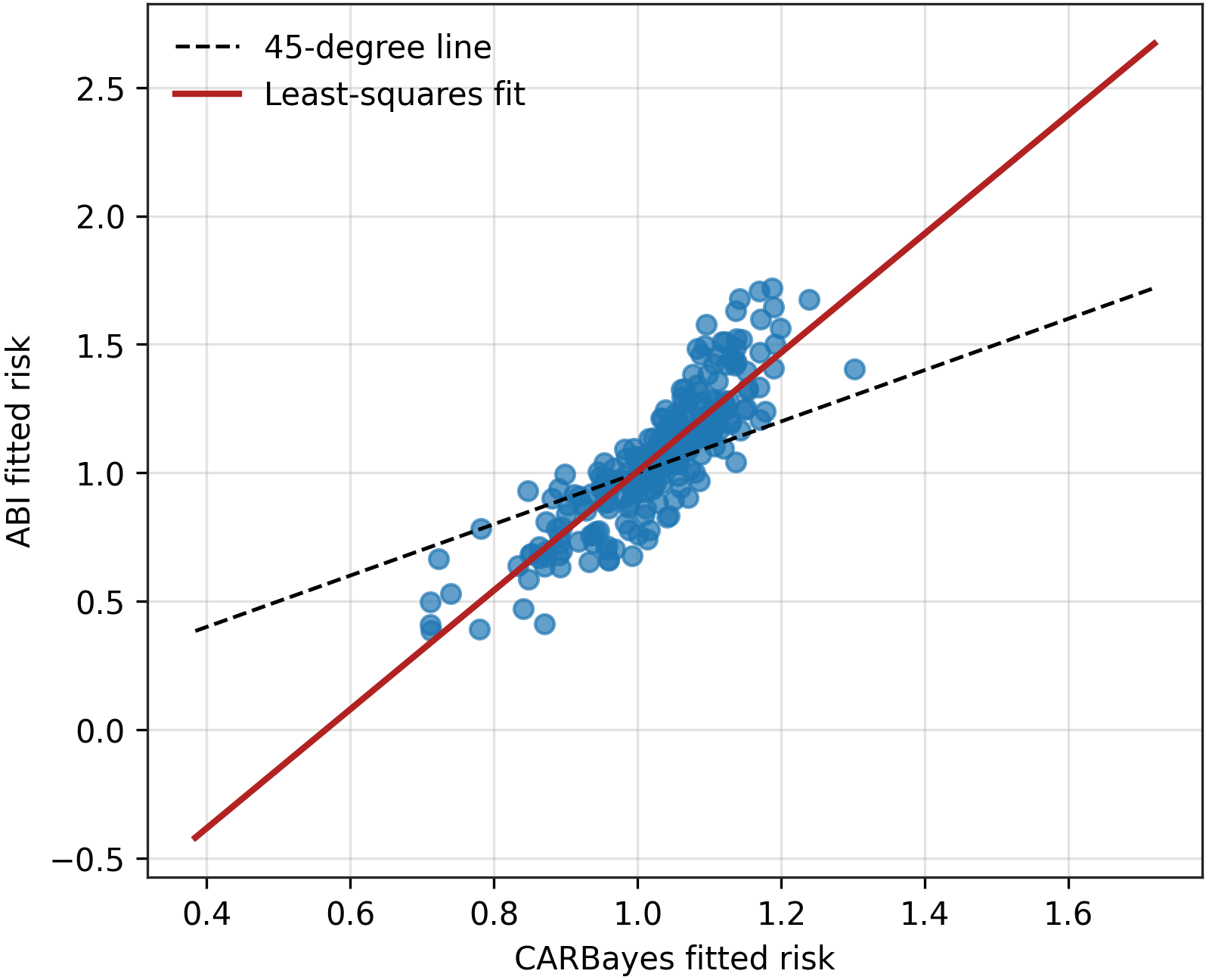}
\caption{Fitted-risk comparison between ABI-DAGAR and \texttt{CARBayes}.
Left to right: Greater Glasgow, California, and South Korea.}
\label{supp:fig:risk}
\end{figure}

\FloatBarrier
\subsection{Boundary probabilities and covariate dissimilarity}
\label{supp:real_dissimilarity}

Fig.~\ref{supp:fig:real_dissimilarity} complements the residual-contrast
diagnostic in the main manuscript by comparing how ABI-DAGAR and
\texttt{CARBayes} translate standardized covariate dissimilarity into posterior
boundary probability.

Under ABI-DAGAR, the Pearson correlation between posterior boundary probability
and standardized dissimilarity is $0.884$ in Greater Glasgow, $0.912$ in
California, and $0.906$ in South Korea, compared with $0.825$, $0.865$, and
$0.812$, respectively, under \texttt{CARBayes}. Thus, both analyses associate larger
deprivation or smoking-prevalence contrasts with stronger evidence for
interruption of spatial borrowing, although ABI-DAGAR assigns somewhat greater
boundary probability to some of the most dissimilar neighboring pairs.
Fig.~\ref{supp:fig:real_dissimilarity} displays these relationships for all
three applications.

\begin{figure}[htbp!]
\centering
\includegraphics[width=0.32\linewidth]
{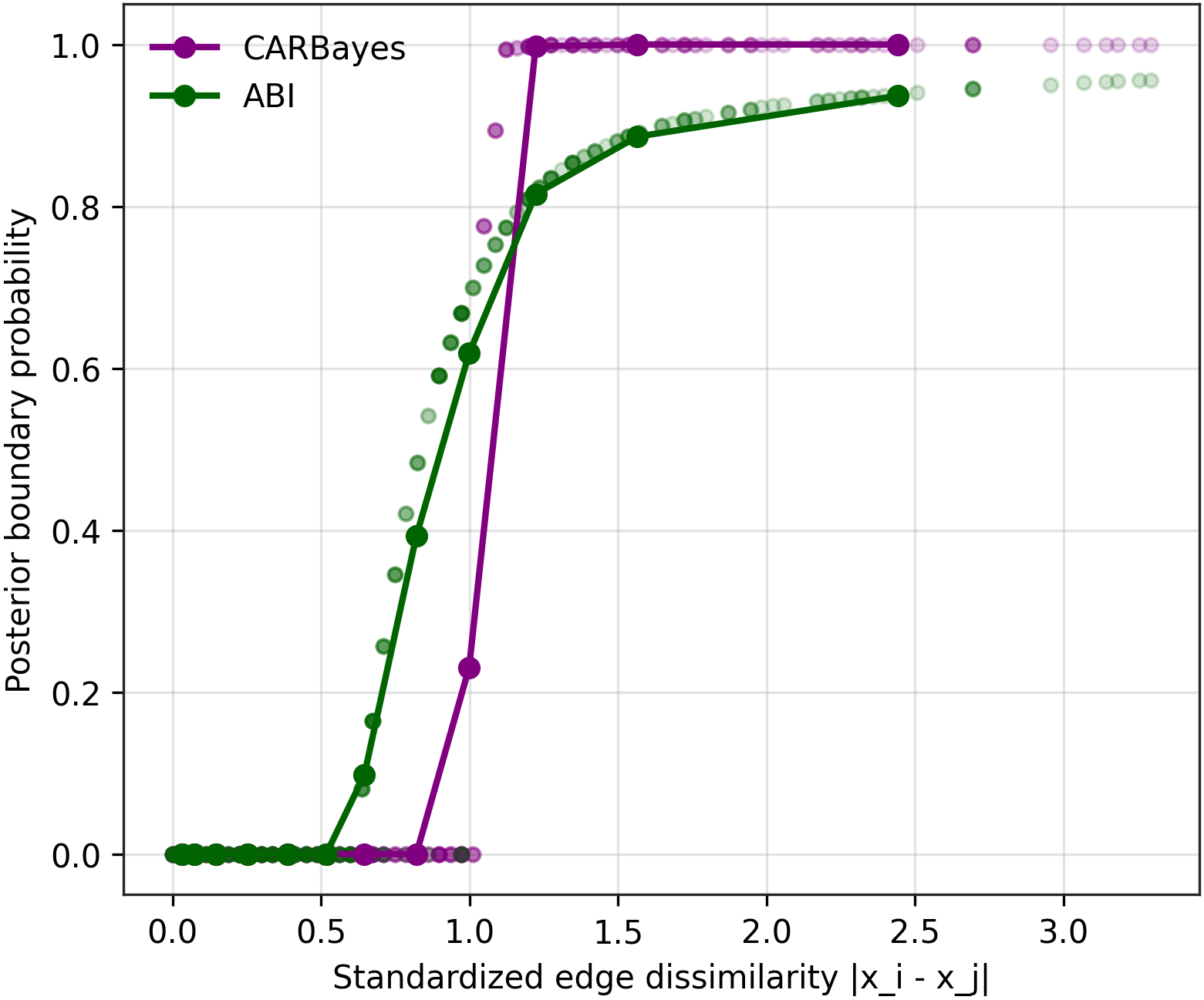}
\includegraphics[width=0.32\linewidth]
{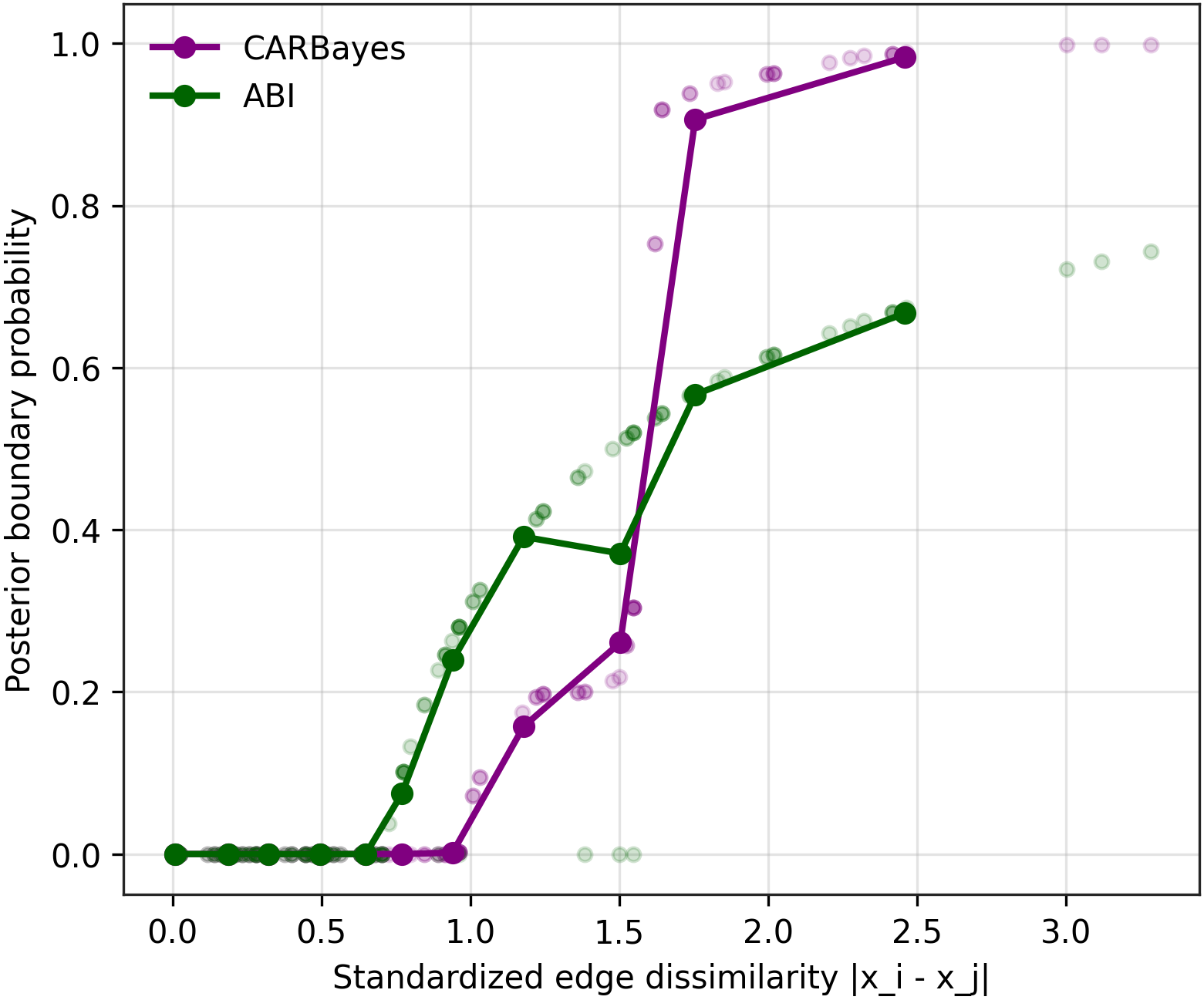}
\includegraphics[width=0.32\linewidth]
{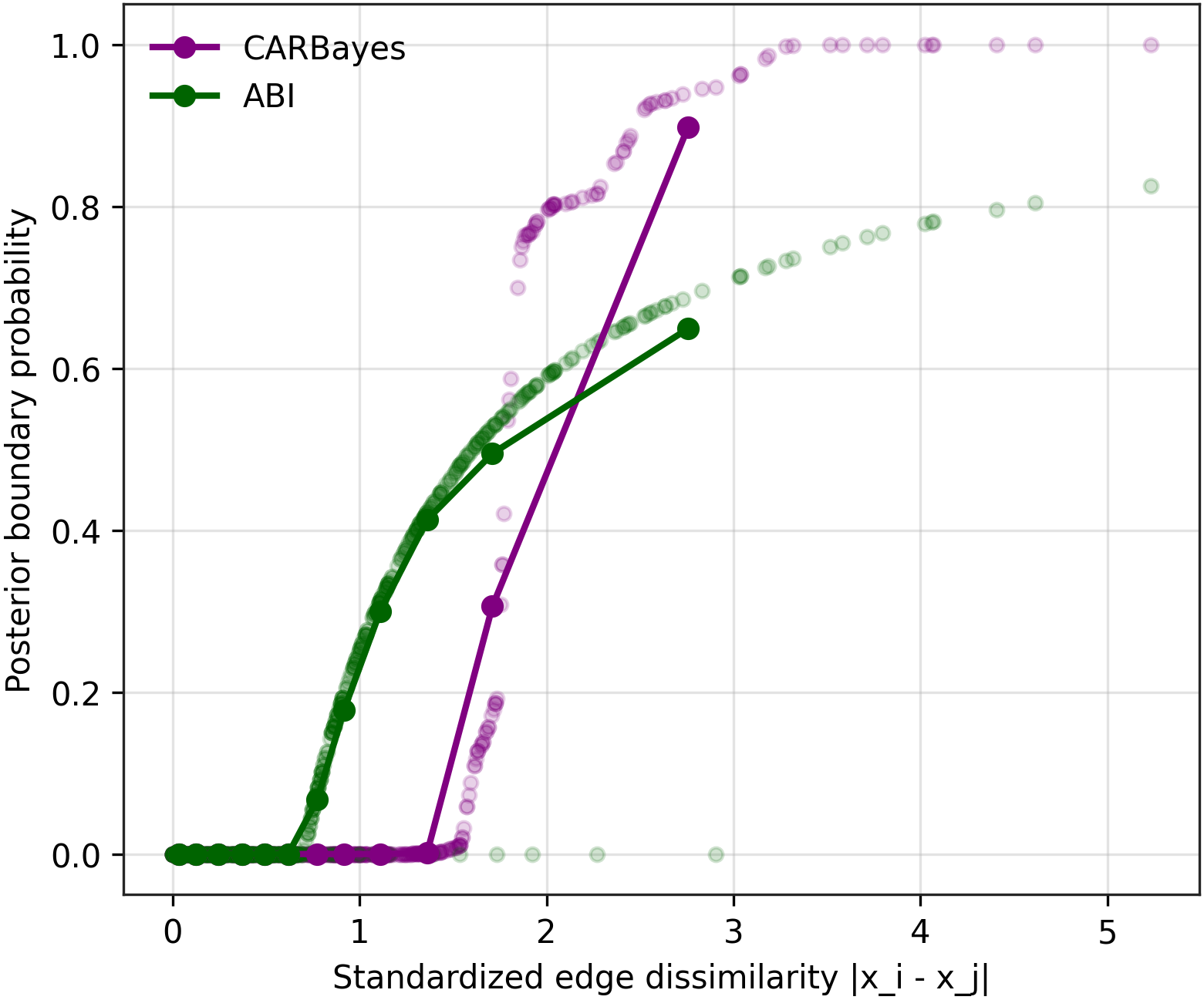}
\caption{Posterior boundary probability versus standardized edge
dissimilarity under ABI-DAGAR and \texttt{CARBayes}. Left to right: Greater Glasgow, California, and South Korea.}
\label{supp:fig:real_dissimilarity}
\end{figure}

\FloatBarrier
\subsection{DAGAR MCMC comparison}
\label{supp:real_mcmc}

The main manuscript summarizes the DAGAR MCMC comparison for the spatial
structure parameters $\eta$ and $\rho$ and reports the corresponding
edge-level probability agreement. Here we provide additional posterior
summaries, selected-boundary comparisons, and spatial maps for the same
benchmark.

\begin{table}[htbp!]
\centering
\caption{Posterior medians with 95\% credible intervals for ABI-DAGAR and
the DAGAR MCMC comparator.}
\label{supp:tab:real_mcmc_parameters}
\begin{tabular}{lccc}
\hline
\textbf{Dataset}
& \textbf{Parameter}
& \textbf{ABI-DAGAR}
& \textbf{MCMC-DAGAR} \\
\hline
Glasgow
& $\beta_0$
& $-0.239$ $(-0.309,-0.172)$
& $-0.220$ $(-0.242,-0.199)$ \\
Glasgow
& $\sigma_w^2$
& $0.336$ $(0.077,1.097)$
& $0.293$ $(0.125,1.683)$ \\
Glasgow
& $\eta$
& $0.831$ $(0.116,1.137)$
& $0.871$ $(0.702,1.083)$ \\
Glasgow
& $\rho$
& $0.878$ $(0.449,0.975)$
& $0.857$ $(0.643,0.976)$ \\
\hline
California
& $\beta_0$
& $0.086$ $(0.015,0.152)$
& $0.090$ $(0.070,0.108)$ \\
California
& $\sigma_w^2$
& $0.106$ $(0.011,0.839)$
& $0.031$ $(0.017,0.126)$ \\
California
& $\eta$
& $0.469$ $(0.017,0.964)$
& $0.440$ $(0.275,0.936)$ \\
California
& $\rho$
& $0.815$ $(0.051,0.990)$
& $0.578$ $(0.264,0.898)$ \\
\hline
South Korea
& $\beta_0$
& $0.022$ $(-0.021,0.062)$
& $0.019$ $(0.012,0.027)$ \\
South Korea
& $\sigma_w^2$
& $0.079$ $(0.007,0.812)$
& $0.015$ $(0.009,0.037)$ \\
South Korea
& $\eta$
& $0.433$ $(0.017,0.961)$
& $0.350$ $(0.214,0.511)$ \\
South Korea
& $\rho$
& $0.893$ $(0.143,0.994)$
& $0.634$ $(0.438,0.861)$ \\
\hline
\end{tabular}
\end{table}

Table~\ref{supp:tab:real_mcmc_parameters} shows that the MCMC-DAGAR intervals
are broader than the corresponding \texttt{CARBayes} intervals for the
boundary and spatial-variance parameters. Consequently,
wider intervals than \texttt{CARBayes} are not a feature unique to amortized
inference. Relative to MCMC-DAGAR, ABI-DAGAR has wider intervals for $\eta$
and $\rho$ in all three applications, but not uniformly for $\sigma_w^2$:
the Glasgow MCMC interval for the spatial variance is wider than its ABI
counterpart. Posterior locations for $\eta$ are comparatively close, whereas
the ABI posterior favors stronger residual dependence in California and South
Korea. Since the prior specifications also differ, these comparisons describe
the fitted analyses rather than isolating the contribution of amortization.

Let $\mathcal B_A$ and $\mathcal B_M$ denote the ABI-DAGAR and MCMC-DAGAR
median-probability boundary sets. In Table~\ref{supp:tab:real_mcmc_boundary},
``Shared'' is $|\mathcal B_A\cap\mathcal B_M|$,
``ABI in MCMC'' is
$|\mathcal B_A\cap\mathcal B_M|/|\mathcal B_A|$,
``MCMC in ABI'' is
$|\mathcal B_A\cap\mathcal B_M|/|\mathcal B_M|$, and
``Jaccard'' is
\[
\frac{
|\mathcal B_A\cap\mathcal B_M|
}{
|\mathcal B_A\cup\mathcal B_M|
}.
\]
The displayed percentages are the two directional containment proportions.

\begin{table}[htbp!]
\centering
\caption{Agreement between ABI-DAGAR and the DAGAR MCMC comparator. Risk and
boundary-probability correlations are Pearson correlations across areas and
geographic edges, respectively.}
\label{supp:tab:real_mcmc_boundary}
\resizebox{\textwidth}{!}{%
\begin{tabular}{lcccccccc}
\hline
\textbf{Dataset}
& \textbf{Risk corr.}
& \textbf{Bound. corr.}
& \textbf{ABI sel.}
& \textbf{MCMC sel.}
& \textbf{Shared}
& \textbf{ABI in MCMC}
& \textbf{MCMC in ABI}
& \textbf{Jaccard} \\
\hline
Glasgow
& 0.774 & 0.959 & 130 & 140 & 130 & 100.00\% & 92.86\% & 0.929 \\
California
& 0.847 & 0.922 & 31 & 24 & 24 & 77.42\% & 100.00\% & 0.774 \\
South Korea
& 0.881 & 0.777 & 89 & 52 & 50 & 56.18\% & 96.15\% & 0.549 \\
\hline
\end{tabular}%
}
\end{table}

Table~\ref{supp:tab:real_mcmc_boundary} shows the closest agreement in
Glasgow, where the MCMC-DAGAR set contains all 130 ABI-DAGAR selections and
10 additional edges. In California, all 24 MCMC-DAGAR selections are included
among the 31 ABI-DAGAR selections. In South Korea, 50 edges are shared, 39 are
selected only by ABI-DAGAR, and two are selected only by MCMC-DAGAR. Both
procedures use the same scalar-threshold mechanism and therefore share the
raw ordering induced by covariate dissimilarity. The 39 South Korean edges
selected only by ABI-DAGAR have probabilities concentrated near the decision
threshold (median $0.538$, range $0.503$--$0.580$), explaining why modest
probability differences produce appreciable differences in selected counts.
Fig.~\ref{supp:fig:real_mcmc} maps the corresponding boundary sets.

Both \texttt{CARBayes} and DAGAR MCMC used one chain of 300,000 iterations
per application, with 100,000 burn-in iterations and thinning by 20, yielding
10,000 retained draws. Table~\ref{supp:tab:real_mcmc_precision} reports ESS
and posterior-mean MCSE computed from these draws using the spectral-variance
method described in Section~\ref{supp:mcmc_sim}. The \texttt{CARBayes}
parameters $\alpha$ and $\tau^2$ denote its boundary and spatial-variance
parameters, respectively; they are not relabeled as DAGAR variance components.
The scripts save these quantities alongside trace, autocorrelation, and
running-mean diagnostics.

Table~\ref{supp:tab:real_mcmc_precision} reports ESS, posterior-mean MCSE, and
MCSE relative to posterior standard deviation for the scalar parameters in
both real-data MCMC benchmarks. Monte Carlo precision is strong for nearly all
reported posterior means. The main exception is the California
\texttt{CARBayes} boundary parameter, for which the diagnostics indicate
noticeably weaker Monte Carlo precision. The table therefore supports the
numerical stability of the posterior-mean comparisons overall while
identifying the estimate that warrants the greatest caution. These
single-chain summaries quantify Monte Carlo precision for posterior means; the
benchmark comparisons remain descriptive.

\begin{table}[htbp!]
\centering
\small
\caption{Single-chain posterior-mean precision for the real-data benchmarks.}
\label{supp:tab:real_mcmc_precision}
\begin{tabular}{lllrrr}
\hline
\textbf{Dataset} & \textbf{Sampler} & \textbf{Parameter}
& \textbf{ESS} & \shortstack{\textbf{MCSE}\\$(\times10^{-3})$}
& \shortstack{\textbf{MCSE / SD}\\(\%)} \\
\hline
Glasgow & \texttt{CARBayes} & $\beta_0$ & 9086.000 & 0.122 & 1.050 \\
& & $\alpha$ & 9269.300 & 0.184 & 1.040 \\
& & $\tau^2$ & 9847.800 & 0.244 & 1.010 \\
& DAGAR MCMC & $\beta_0$ & 9321.700 & 0.111 & 1.040 \\
& & $\sigma_w^2$ & 748.200 & 26.834 & 3.660 \\
& & $\eta$ & 7722.100 & 1.100 & 1.140 \\
& & $\rho$ & 1462.600 & 2.270 & 2.610 \\
\hline
California & \texttt{CARBayes} & $\beta_0$ & 4723.000 & 0.196 & 1.460 \\
& & $\alpha$ & 95.500 & 10.093 & 10.240 \\
& & $\tau^2$ & 428.200 & 0.506 & 4.830 \\
& DAGAR MCMC & $\beta_0$ & 2269.800 & 0.204 & 2.100 \\
& & $\sigma_w^2$ & 2252.600 & 0.903 & 2.110 \\
& & $\eta$ & 499.900 & 7.974 & 4.470 \\
& & $\rho$ & 3114.800 & 2.983 & 1.790 \\
\hline
South Korea & \texttt{CARBayes} & $\beta_0$ & 10000.000 & 0.041 & 1.000 \\
& & $\alpha$ & 566.300 & 2.374 & 4.200 \\
& & $\tau^2$ & 1150.900 & 0.093 & 2.950 \\
& DAGAR MCMC & $\beta_0$ & 10000.000 & 0.039 & 1.000 \\
& & $\sigma_w^2$ & 1079.600 & 0.310 & 3.040 \\
& & $\eta$ & 612.300 & 4.040 & 4.040 \\
& & $\rho$ & 1188.900 & 3.187 & 2.900 \\
\hline
\end{tabular}
\end{table}

\begin{figure}[htbp!]
\centering
\includegraphics[height=0.26\textwidth,keepaspectratio]{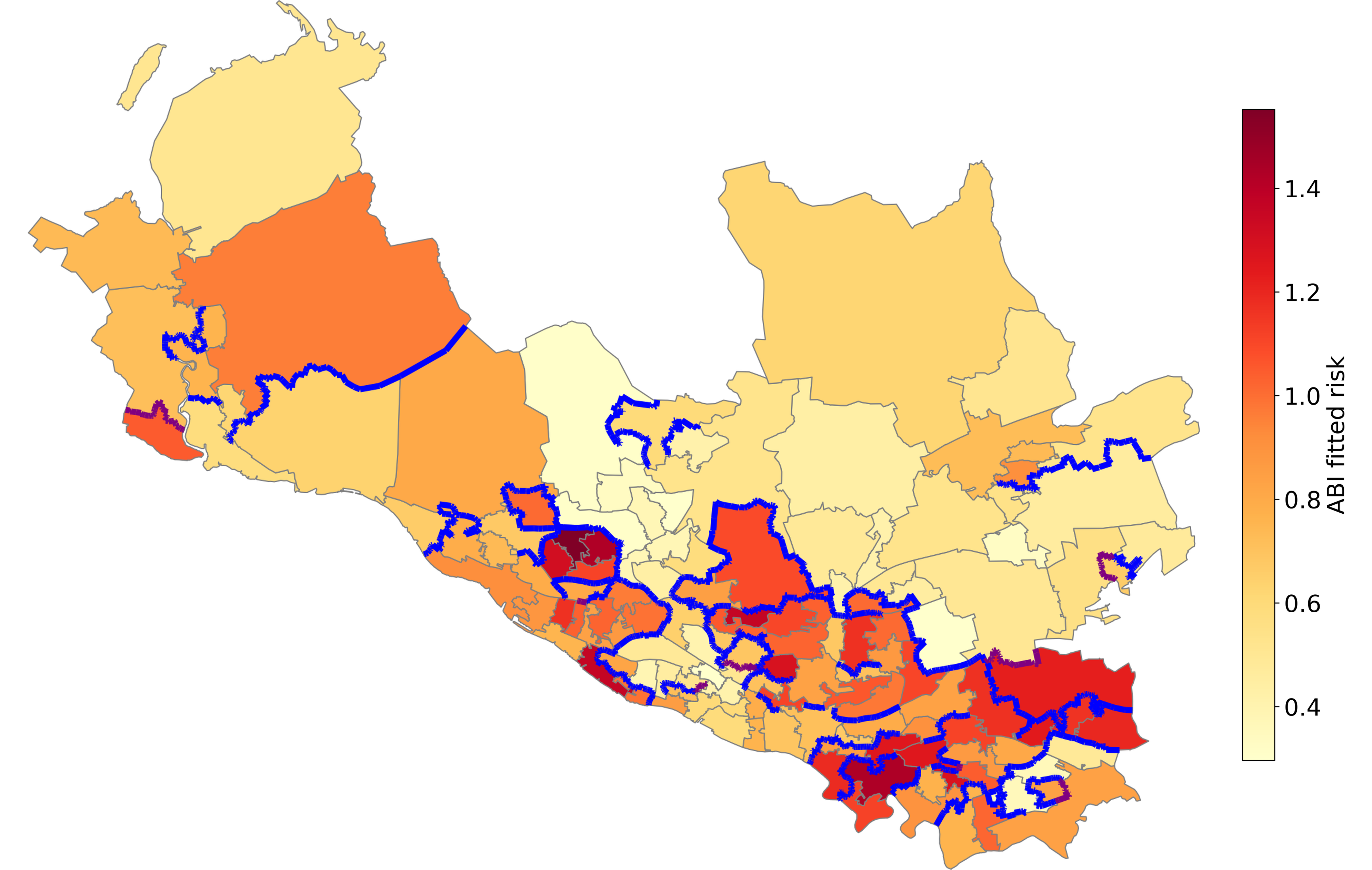}\hfill
\includegraphics[height=0.26\textwidth,keepaspectratio]{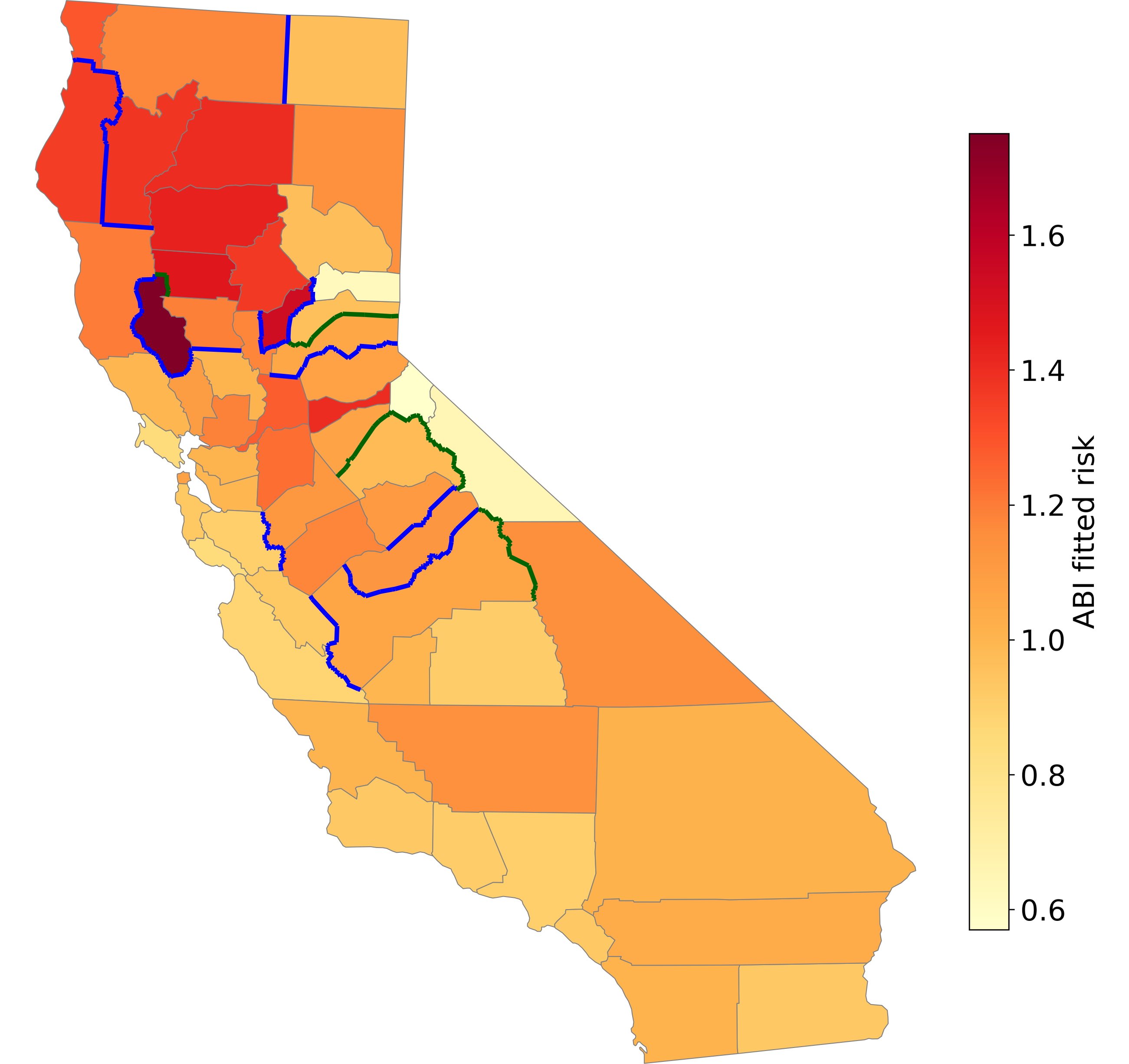}\hfill
\includegraphics[height=0.26\textwidth,keepaspectratio]{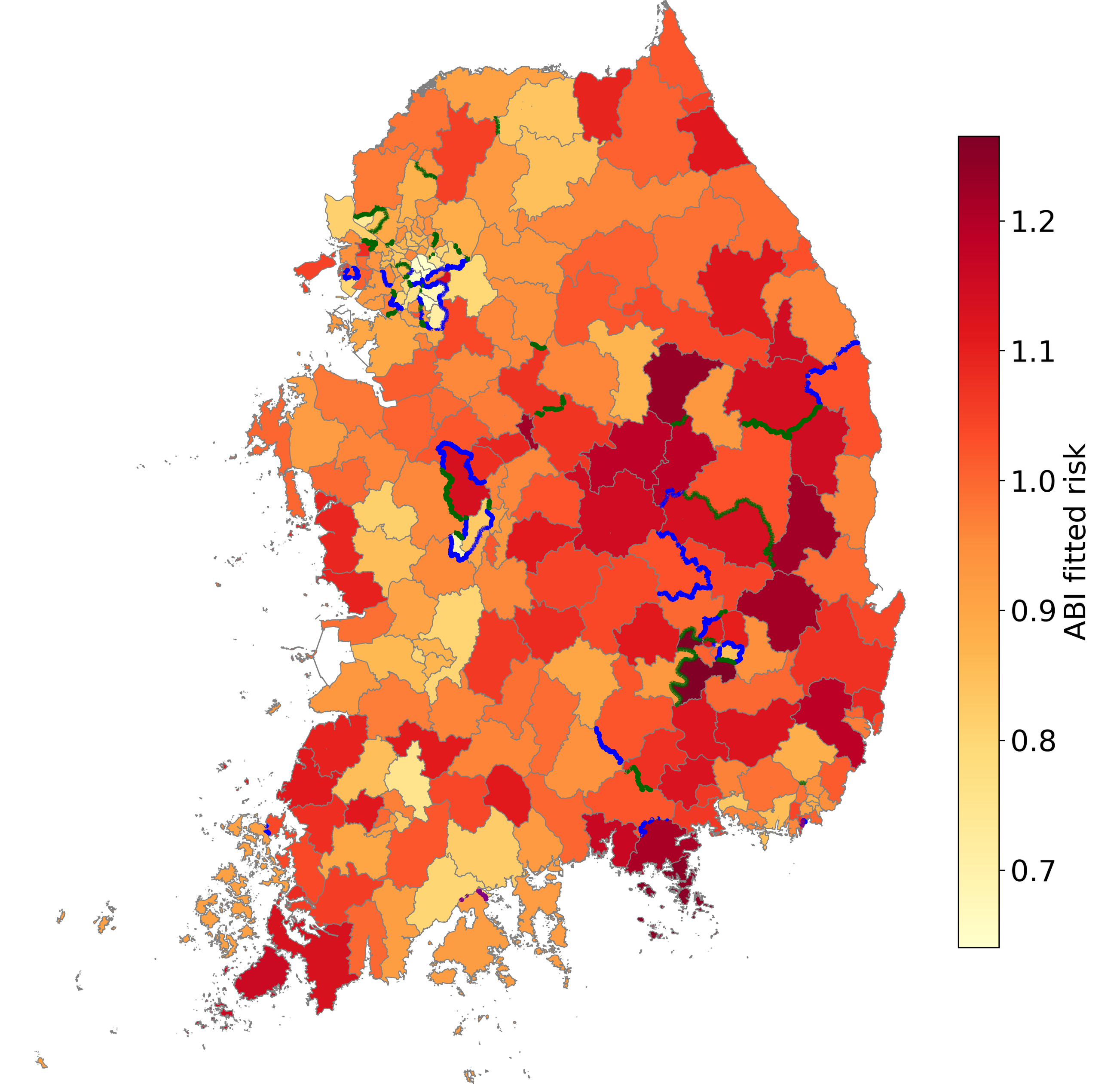}
\caption{Median-probability boundary agreement with DAGAR MCMC. Blue: both
methods; green: ABI only; purple: MCMC only. Area shading gives the
reconstructed ABI fitted relative risk. Left to right: Greater Glasgow,
California, and South Korea.}
\label{supp:fig:real_mcmc}
\end{figure}

\FloatBarrier
\subsection{Runtime comparison}
\label{supp:real_runtime}

Table~\ref{supp:tab:runtime} reports the wall-clock timing references for
ABI-DAGAR and the longer and shorter \texttt{CARBayes} configurations across
the three applications. The posterior summaries reported in the main
manuscript use the longer \texttt{CARBayes} configuration; the shorter runs are
included only as timing references.

\begin{table}[htbp!]
\centering
\caption{Real-data posterior sampling runtimes for ABI-DAGAR and
\texttt{CARBayes}.}
\label{supp:tab:runtime}
\begin{tabular}{llcc}
\hline
\textbf{Dataset}
& \textbf{Method}
& \textbf{Retained draws}
& \textbf{Runtime} \\
\hline
Glasgow
& ABI-DAGAR
& 10,000
& $\approx 1$ s \\
Glasgow
& \texttt{CARBayes}, 300k iter., 100k burn-in, thin 20
& 10,000
& 615 s \\
Glasgow
& \texttt{CARBayes}, 20k iter., 10k burn-in, thin 1
& 10,000
& 45 s \\
\hline
California
& ABI-DAGAR
& 10,000
& $\approx 1$ s \\
California
& \texttt{CARBayes}, 300k iter., 100k burn-in, thin 20
& 10,000
& 395 s \\
California
& \texttt{CARBayes}, 20k iter., 10k burn-in, thin 1
& 10,000
& 30 s \\
\hline
South Korea
& ABI-DAGAR
& 10,000
& $\approx 1$ s \\
South Korea
& \texttt{CARBayes}, 300k iter., 100k burn-in, thin 20
& 10,000
& 1,198 s \\
South Korea
& \texttt{CARBayes}, 20k iter., 10k burn-in, thin 1
& 10,000
& 69 s \\
\hline
\end{tabular}
\end{table}

The shorter \texttt{CARBayes} runs are timing references only; reported
posterior results use the longer configuration. Including the one-time
ABI-DAGAR training cost of 5 hours and 45 minutes, the break-even points are
approximately 34, 53, and 17 analyses under the longer Glasgow, California,
and South Korea runs, respectively, and 471, 714, and 304 under the shorter
runs. The computational advantage of amortization is therefore most relevant
for repeated deployment across related maps, outcomes, or sensitivity analyses.

\clearpage


\bibliographystyle{apalike} 
\bibliography{bibliography}       

@article{datta2019spatial,
  title={{Spatial disease mapping using directed acyclic graph auto-regressive (DAGAR) models}},
  author={Datta, Abhirup and Banerjee, Sudipto and Hodges, James S and Gao, Leiwen},
  journal={Bayesian Analysis},
  volume={14},
  number={4},
  pages={1221--1244},
  year={2019},
  doi={10.1214/19-BA1177},
  url={https://doi.org/10.1214/19-BA1177}
}

@article{gao2023spatial,
  title={Spatial difference boundary detection for multiple outcomes using {B}ayesian disease mapping},
  author={Gao, Leiwen and Banerjee, Sudipto and Ritz, Beate},
  journal={Biostatistics},
  volume={24},
  number={4},
  pages={922--944},
  year={2023},
  doi={10.1093/biostatistics/kxac013},
  url={https://doi.org/10.1093/biostatistics/kxac013}
}

@article{lee2012boundary,
  title={Boundary detection in disease mapping studies},
  author={Lee, Duncan and Mitchell, Richard},
  journal={Biostatistics},
  volume={13},
  number={3},
  pages={415--426},
  year={2012},
  doi={10.1093/biostatistics/kxr036},
  url={https://doi.org/10.1093/biostatistics/kxr036}
}

@article{li2015bayesian,
  title={{B}ayesian models for detecting difference boundaries in areal data},
  author={Li, Pei and Banerjee, Sudipto and Hanson, Timothy A and McBean, Alexander M},
  journal={Statistica Sinica},
  volume={25},
  number={1},
  pages={385--402},
  year={2015},
  doi={10.5705/ss.2013.238w},
  url={https://doi.org/10.5705/ss.2013.238w}
}

@article{lu2007bayesian,
  title={{B}ayesian areal wombling via adjacency modeling},
  author={Lu, Haolan and Reilly, Cavan S and Banerjee, Sudipto and Carlin, Bradley P},
  journal={Environmental and Ecological Statistics},
  volume={14},
  number={4},
  pages={433--452},
  year={2007},
  doi={10.1007/s10651-007-0029-9},
  url={https://doi.org/10.1007/s10651-007-0029-9}
}

@article{ma2007bayesian,
author = {Bradley P. Carlin and Haijun Ma},
title = {{Bayesian multivariate areal wombling for multiple disease boundary analysis}},
volume = {2},
journal = {Bayesian Analysis},
number = {2},
publisher = {International Society for Bayesian Analysis},
pages = {281 -- 302},
year = {2007},
doi = {10.1214/07-BA211},
URL = {https://doi.org/10.1214/07-BA211}
}

@article{ma2010hierarchical,
  title={Hierarchical and joint site-edge methods for Medicare hospice service region boundary analysis},
  author={Ma, Haijun and Carlin, Bradley P and Banerjee, Sudipto},
  journal={Biometrics},
  volume={66},
  number={2},
  pages={355--364},
  year={2010},
  doi={10.1111/j.1541-0420.2009.01291.x},
  url={https://doi.org/10.1111/j.1541-0420.2009.01291.x}
}

@book{koch2005cartographies,
title={Cartographies of Disease: Maps, Mapping, and Medicine},
author={Koch, Tom},
year={2005},
publisher={Esri Press},
address={Redlands, CA},
isbn={9781589481206}
}

@book{lawson2016handbook,
	title={Handbook of Spatial Epidemiology},
	editor={Lawson, Andrew B. and Banerjee, Sudipto and Haining, Robert and Ugarte, Mar{\'i}a D.},
	year={2016},
	publisher={CRC Press},
	address={Boca Raton, FL},
	isbn={9781482253016},
	url={https://www.routledge.com/Handbook-of-Spatial-Epidemiology/Lawson-Banerjee-Haining-Ugarte/p/book/9780367570385}
}

@article{lu2005bayesian,
	title={{B}ayesian Areal Wombling for Geographical Boundary Analysis},
	author={Lu, Haolan and Carlin, Bradley P},
	journal={Geographical Analysis},
	volume={37},
	number={3},
	pages={265--285},
	year={2005},
	doi={10.1111/j.1538-4632.2005.00624.x},
	url={https://doi.org/10.1111/j.1538-4632.2005.00624.x}
}

@manual{seer,
  author       = {{Surveillance Research Program, National Cancer Institute}},
  title        = {{SEER*Stat} Software},
  organization = {National Cancer Institute},
  address      = {Bethesda, MD},
  year         = {2019},
  note         = {Version 8.3.6, released August 8, 2019},
  url          = {https://seer.cancer.gov/seerstat/}
}

@article{pavani2025bayesian,
author = {Pavani, Jessica and Quintana, Fernando Andrés},
title = {A Bayesian Multivariate Model With Temporal Dependence on Random Partition of Areal Data for Mosquito-Borne Diseases},
journal = {Statistics in Medicine},
volume = {44},
number = {3-4},
pages = {e10325},
doi = {https://doi.org/10.1002/sim.10325},
url = {https://onlinelibrary.wiley.com/doi/abs/10.1002/sim.10325},
eprint = {https://onlinelibrary.wiley.com/doi/pdf/10.1002/sim.10325},
year = {2025}
}

@article{pavani2026modeling,
author = {Jessica Pavani and Rosangela H. Loschi and Fernando A. Quintana},
title = {{Modeling temporal dependence in a sequence of spatial random partitions driven by spanning tree: An application to mosquito-borne diseases}},
volume = {20},
journal = {The Annals of Applied Statistics},
number = {2},
publisher = {Institute of Mathematical Statistics},
pages = {1388 -- 1408},
year = {2026},
doi = {10.1214/26-AOAS2172},
URL = {https://doi.org/10.1214/26-AOAS2172}
}

@article{gianella2026bayesian,
author = {Matteo Gianella and Mario Beraha and Alessandra Guglielmi},
title = {{Bayesian Nonparametric Boundary Detection for Multiple Areal Data}},
journal = {Bayesian Analysis},
publisher = {International Society for Bayesian Analysis},
pages = {1 -- 25},
year = {2026},
doi = {10.1214/26-BA1605},
URL = {https://doi.org/10.1214/26-BA1605}
}

@Article{yao2017using,
    title = {Using stacking to average {B}ayesian predictive distributions},
    author = {Yuling Yao and Aki Vehtari and Daniel Simpson and Andrew Gelman},
    year = {2017},
    journal = {Bayesian Analysis},
    doi = {10.1214/17-BA1091},
  }

@article{zhang2025jasa,
author = {Lu Zhang and Wenpin Tang and Sudipto Banerjee},
title = {Bayesian Geostatistics Using Predictive Stacking},
journal = {Journal of the American Statistical Association},
volume = {121},
number = {554},
pages = {1549--1561},
year = {2026},
publisher = {Taylor \& Francis},
doi = {10.1080/01621459.2025.2566449},
URL = {https://doi.org/10.1080/01621459.2025.2566449},
eprint = {https://doi.org/10.1080/01621459.2025.2566449}
}

@article{besag1991bayesian,
  title={Bayesian image restoration, with two applications in spatial statistics},
  author={Besag, Julian and York, Jeremy and Molli{\'e}, Annie},
  journal={Annals of the Institute of Statistical Mathematics},
  volume={43},
  number={1},
  pages={1--20},
  year={1991},
  doi={10.1007/BF00116466},
  url={https://doi.org/10.1007/BF00116466}
}

@article{radev2020bayesflow,
  title={Bayes{F}low: Learning complex stochastic models with invertible neural networks},
  author={Radev, Stefan T and Mertens, Ulf K and Voss, Andreas and Ardizzone, Lynton and K{\"o}the, Ullrich},
  journal={IEEE Transactions on Neural Networks and Learning Systems},
  volume={33},
  number={4},
  pages={1452--1466},
  year={2022},
  doi={10.1109/TNNLS.2020.3042395},
  url={https://doi.org/10.1109/TNNLS.2020.3042395}
}

@article{zammit2025neural,
  title={Neural methods for amortized inference},
  author={Zammit-Mangion, Andrew and Sainsbury-Dale, Matthew and Huser, Rapha{\"e}l},
  journal={Annual Review of Statistics and Its Application},
  volume={12},
  number={1},
  pages={311--335},
  year={2025},
  doi={10.1146/annurev-statistics-112723-034123},
  url={https://doi.org/10.1146/annurev-statistics-112723-034123}
}

@article{sainsbury2024likelihood,
  title={Likelihood-free parameter estimation with neural {B}ayes estimators},
  author={Sainsbury-Dale, Matthew and Zammit-Mangion, Andrew and Huser, Rapha{\"e}l},
  journal={The American Statistician},
  volume={78},
  number={1},
  pages={1--14},
  year={2024},
  doi={10.1080/00031305.2023.2249522},
  url={https://doi.org/10.1080/00031305.2023.2249522}
}

@article{aiello2023detecting,
  title={Detecting spatial health disparities using disease maps},
  author={Aiello, Luca and Banerjee, Sudipto},
  journal={arXiv preprint arXiv:2309.02086},
  year={2023}
}

@InProceedings{rezende2015variational,
  title = 	 {Variational Inference with Normalizing Flows},
  author = 	 {Rezende, Danilo and Mohamed, Shakir},
  booktitle = 	 {Proceedings of the 32nd International Conference on Machine Learning},
  pages = 	 {1530--1538},
  year = 	 {2015},
  editor = 	 {Bach, Francis and Blei, David},
  volume = 	 {37},
  series = 	 {Proceedings of Machine Learning Research},
  address = 	 {Lille, France},
  month = 	 {07--09 Jul},
  publisher =    {PMLR},
  url = 	 {https://proceedings.mlr.press/v37/rezende15.html}
}

@inproceedings{papamakarios2017masked,
 author = {Papamakarios, George and Pavlakou, Theo and Murray, Iain},
 booktitle = {Advances in Neural Information Processing Systems},
 editor = {I. Guyon and U. Von Luxburg and S. Bengio and H. Wallach and R. Fergus and S. Vishwanathan and R. Garnett},
 pages = {},
 publisher = {Curran Associates, Inc.},
 title = {Masked Autoregressive Flow for Density Estimation},
 url = {https://proceedings.neurips.cc/paper_files/paper/2017/file/6c1da886822c67822bcf3679d04369fa-Paper.pdf},
 volume = {30},
 year = {2017}
}

@inproceedings{durkan2019neural,
author = {Durkan, Conor and Bekasov, Artur and Murray, Iain and Papamakarios, George},
title = {Neural spline flows},
year = {2019},
publisher = {Curran Associates Inc.},
address = {Red Hook, NY, USA},
booktitle = {Proceedings of the 33rd International Conference on Neural Information Processing Systems},
articleno = {675},
numpages = {12}
}

@article{papamakarios2021normalizing,
  author  = {George Papamakarios and Eric Nalisnick and Danilo Jimenez Rezende and Shakir Mohamed and Balaji Lakshminarayanan},
  title   = {Normalizing Flows for Probabilistic Modeling and Inference},
  journal = {Journal of Machine Learning Research},
  year    = {2021},
  volume  = {22},
  number  = {57},
  pages   = {1--64},
  url     = {http://jmlr.org/papers/v22/19-1028.html}
}

@InProceedings{lee2019set,
  title = 	 {Set {T}ransformer: A Framework for Attention-based Permutation-Invariant Neural Networks},
  author =       {Lee, Juho and Lee, Yoonho and Kim, Jungtaek and Kosiorek, Adam and Choi, Seungjin and Teh, Yee Whye},
  booktitle = 	 {Proceedings of the 36th International Conference on Machine Learning},
  pages = 	 {3744--3753},
  year = 	 {2019},
  editor = 	 {Chaudhuri, Kamalika and Salakhutdinov, Ruslan},
  volume = 	 {97},
  series = 	 {Proceedings of Machine Learning Research},
  month = 	 {09--15 Jun},
  publisher =    {PMLR},
  url = 	 {https://proceedings.mlr.press/v97/lee19d.html}
}

@book{lawson2018bayesian,
  title={Bayesian Disease Mapping: Hierarchical Modeling in Spatial Epidemiology},
  author={Lawson, Andrew B},
  year={2018},
  edition={3},
  publisher={Chapman and Hall/CRC},
  url={https://www.routledge.com/Bayesian-Disease-Mapping-Hierarchical-Modeling-in-Spatial-Epidemiology-Third-Edition/Lawson/p/book/9780367781224}
}

@article{wakefield2007disease,
  title={Disease mapping and spatial regression with count data},
  author={Wakefield, Jon},
  journal={Biostatistics},
  volume={8},
  number={2},
  pages={158--183},
  year={2007},
  doi={10.1093/biostatistics/kxl008},
  url={https://doi.org/10.1093/biostatistics/kxl008}
}

@incollection{leroux2000estimation,
  title={Estimation of disease rates in small areas: a new mixed model for spatial dependence},
  author={Leroux, Brian G and Lei, Xingye and Breslow, Norman},
  booktitle={Statistical Models in Epidemiology, the Environment, and Clinical Trials},
  editor={Halloran, M. Elizabeth and Berry, Donald},
  pages={179--191},
  year={2000},
  publisher={Springer},
  doi={10.1007/978-1-4612-1284-3_4},
  url={https://doi.org/10.1007/978-1-4612-1284-3_4}
}

@article{lee2014bayesian,
  title   = {A {B}ayesian Localized Conditional Autoregressive Model for Estimating the Health Effects of Air Pollution},
  author  = {Lee, Duncan and Mitchell, Richard},
  journal = {Biometrics},
  volume  = {70},
  number  = {2},
  pages   = {419--429},
  year    = {2014},
  doi     = {10.1111/biom.12156}
}

@article{rushworth2017adaptive,
  title={An adaptive spatiotemporal smoothing model for estimating trends and step changes in disease risk},
  author={Rushworth, Alastair and Lee, Duncan and Sarran, Christophe},
  journal={Journal of the Royal Statistical Society Series C: Applied Statistics},
  volume={66},
  number={1},
  pages={141--157},
  year={2017},
  doi={10.1111/rssc.12155},
  url={https://doi.org/10.1111/rssc.12155}
}

@article{lee2021improved,
  title   = {Improved Inference for Areal Unit Count Data Using Graph-Based Optimisation},
  author  = {Lee, Duncan and Meeks, Kitty and Pettersson, William},
  journal = {Statistics and Computing},
  volume  = {31},
  pages   = {51},
  year    = {2021},
  doi     = {10.1007/s11222-021-10025-7}
}

@article{wu2025assessing,
  title={Assessing spatial disparities: a {B}ayesian linear regression approach},
  author={Wu, Kyle and Banerjee, Sudipto},
  journal={Biostatistics},
  volume={26},
  number={1},
  pages={kxaf048},
  year={2025},
  doi={10.1093/biostatistics/kxaf048},
  url={https://doi.org/10.1093/biostatistics/kxaf048}
}

@Article{CARBayes2013,
    author = {Duncan Lee},
    title = {{CARBayes}: An {R} Package for {B}ayesian Spatial Modeling
      with Conditional Autoregressive Priors},
    year = {2013},
    journal = {{Journal of Statistical Software}},
    doi = {10.18637/jss.v055.i13},
    url = {https://doi.org/10.18637/jss.v055.i13},
    pages = {1--24},
    volume = {55},
    number = {13},
  }

@article{besag1974spatial,
  title={Spatial interaction and the statistical analysis of lattice systems},
  author={Besag, Julian},
  journal={Journal of the Royal Statistical Society: Series B (Methodological)},
  volume={36},
  number={2},
  pages={192--236},
  year={1974},
  doi={10.1111/j.2517-6161.1974.tb00999.x},
  url={https://doi.org/10.1111/j.2517-6161.1974.tb00999.x}
}

@inproceedings{zaheer2017deep,
 author = {Zaheer, Manzil and Kottur, Satwik and Ravanbakhsh, Siamak and Poczos, Barnabas and Salakhutdinov, Russ R and Smola, Alexander},
 booktitle = {Advances in Neural Information Processing Systems},
 editor = {I. Guyon and U. Von Luxburg and S. Bengio and H. Wallach and R. Fergus and S. Vishwanathan and R. Garnett},
 pages = {},
 publisher = {Curran Associates, Inc.},
 title = {Deep Sets},
 url = {https://proceedings.neurips.cc/paper_files/paper/2017/file/f22e4747da1aa27e363d86d40ff442fe-Paper.pdf},
 volume = {30},
 year = {2017}
}

@article{wikle2023statistical,
   author = "Wikle, Christopher K. and Zammit-Mangion, Andrew",
   title = "Statistical Deep Learning for Spatial and Spatiotemporal Data", 
   journal= "Annual Review of Statistics and Its Application",
   year = "2023",
   volume = "10",
   number = "Volume 10, 2023",
   pages = "247-270",
   doi = "https://doi.org/10.1146/annurev-statistics-033021-112628",
   url = "https://www.annualreviews.org/content/journals/10.1146/annurev-statistics-033021-112628",
   publisher = "Annual Reviews",
   issn = "2326-831X",
   type = "Journal Article",
  }

@inproceedings{
navott2026deeprv,
title={Deep{RV}: Accelerating Spatiotemporal Inference with Pre-trained Neural Priors},
author={Jhonathan Navott and Daniel Jenson and Seth Flaxman and Elizaveta Semenova},
booktitle={The 29th International Conference on Artificial Intelligence and Statistics},
year={2026},
url={https://openreview.net/forum?id=5lM1So6mQN}
}

@article{mishra2022pi,
  title={$\pi$ VAE: a stochastic process prior for Bayesian deep learning with MCMC},
  author={Mishra, Swapnil and Flaxman, Seth and Berah, Tresnia and Zhu, Harrison and Pakkanen, Mikko and Bhatt, Samir},
  journal={Statistics and Computing},
  volume={32},
  number={6},
  pages={96},
  year={2022},
  publisher={Springer},
  doi={https://doi.org/10.1007/s11222-022-10151-w}
}

@article{semenova2022priorvae,
    author = {Semenova, Elizaveta and Xu, Yidan and Howes, Adam and Rashid, Theo and Bhatt, Samir and Mishra, Swapnil and Flaxman, Seth},
    title = {PriorVAE: encoding spatial priors with variational autoencoders for small-area estimation},
    journal = {Journal of The Royal Society Interface},
    volume = {19},
    number = {191},
    pages = {20220094},
    year = {2022},
    month = {06},
    issn = {1742-5689},
    doi = {10.1098/rsif.2022.0094},
    url = {https://doi.org/10.1098/rsif.2022.0094},
    eprint = {https://royalsocietypublishing.org/rsif/article-pdf/doi/10.1098/rsif.2022.0094/926596/rsif.2022.0094.pdf},
}

@article{semenova2023priorcvae,
  title={PriorCVAE: scalable MCMC parameter inference with Bayesian deep generative modelling},
  author={Semenova, Elizaveta and Verma, Prakhar and Cairney-Leeming, Max and Solin, Arno and Bhatt, Samir and Flaxman, Seth},
  journal={arXiv preprint arXiv:2304.04307},
  year={2023}
}

@article{gelfand2005spatialdp,
author = {Alan E Gelfand and Athanasios Kottas and Steven N MacEachern},
title = {Bayesian Nonparametric Spatial Modeling With Dirichlet Process Mixing},
journal = {Journal of the American Statistical Association},
volume = {100},
number = {471},
pages = {1021--1035},
year = {2005},
publisher = {Taylor \& Francis},
doi = {10.1198/016214504000002078},
URL = {https://doi.org/10.1198/016214504000002078},
eprint = {https://doi.org/10.1198/016214504000002078}
}

@article{duan2007generalized,
    author = {Duan, Jason A. and Guindani, Michele and Gelfand, Alan E.},
    title = {Generalized Spatial Dirichlet Process Models},
    journal = {Biometrika},
    volume = {94},
    number = {4},
    pages = {809-825},
    year = {2007},
    month = {12},
    issn = {0006-3444},
    doi = {10.1093/biomet/asm071},
    url = {https://doi.org/10.1093/biomet/asm071},
    eprint = {https://academic.oup.com/biomet/article-pdf/94/4/809/681380/asm071.pdf},
}

@article{presicce_bayesian_2024,
  author  = {Luca Presicce and Sudipto Banerjee},
  title   = {Bayesian Transfer Learning for Artificially Intelligent Geospatial Systems: A Predictive Stacking Approach},
  journal = {Journal of Machine Learning Research},
  year    = {2026},
  volume  = {27},
  number  = {196},
  pages   = {1--60},
  url     = {http://jmlr.org/papers/v27/26-0307.html}
}

@article{panEtAl2025ba,
author = {Soumyakanti Pan and Lu Zhang and Jonathan R. Bradley and Sudipto Banerjee},
title = {{Bayesian Inference for Spatial-Temporal Non-Gaussian Data Using Predictive Stacking}},
journal = {Bayesian Analysis},
publisher = {International Society for Bayesian Analysis},
volume = {(In press)},
pages = {},
year = {2025},
doi = {10.1214/25-BA1582},
URL = {https://doi.org/10.1214/25-BA1582}
}

@misc{kosis2026,
  author       = {{Korean Statistical Information Service}},
  title        = {{KOSIS} Statistical Database},
  year         = {2026},
  note         = {Tables DT\_1B34E13 and DT\_1B34E11 for municipal and
    national cause-specific mortality, and table DT\_1B040M5 for municipal
    midyear population; accessed August 2026},
  url          = {https://kosis.kr/eng/}
}

@misc{kdca2026,
  author       = {{Korea Disease Control and Prevention Agency}},
  title        = {Community Health Outcomes and Health Determinants Database,
    Version 1.7},
  year         = {2026},
  note         = {Released April 1, 2026; accessed August 2026},
  url          = {https://chs.kdca.go.kr/chs/recsRoom/dataBaseDownloadPop.do}
}

@misc{koreaBoundaries2019,
  author       = {{Ministry of the Interior and Safety, Republic of Korea}},
  title        = {Road Name Address Electronic Map: May 2019 Municipal
    Boundary Layer
    (\texttt{TL\_SCCO\_SIG})},
  year         = {2019},
  note         = {Legal Si/Gun/Gu boundary layer; the May 2019 archive used by
    the reproducibility scripts was retrieved through the GIS Developer
    mirror},
  url          = {https://eng.juso.go.kr/addrlink/adresInfoProvd/guidance/provdAdresInfo.do}
}

@techreport{californiaTobacco2018,
  author = {{California Department of Public Health, California Tobacco
    Control Program}},
  title = {California Tobacco Facts and Figures 2018},
  institution = {California Department of Public Health},
  address = {Sacramento, CA},
  year = {2018},
  url = {https://www.cdph.ca.gov/Programs/CCDPHP/DCDIC/CTCB/CDPH\%20Document\%20Library/ResearchandEvaluation/FactsandFigures/CATobaccoFactsFigures2018.pdf}
}

@Manual{CARBayesdata2022,
  title = {{CARBayesdata}: Data Used in the Vignettes Accompanying the
    {CARBayes} and {CARBayesST} Packages},
  author = {Duncan Lee},
  year = {2022},
  note = {R package documentation for \texttt{respiratorydata} and
    \texttt{GGHB.IZ}; version 3.0},
  url = {https://cran.r-project.org/web/packages/CARBayesdata/CARBayesdata.pdf}
}

@article{wang2026inference,
title = {Inference for stationary Log-Gaussian Cox point processes using Bayesian deep learning: Application to human oral microbiome image data},
journal = {Spatial Statistics},
volume = {73},
pages = {100973},
year = {2026},
issn = {2211-6753},
doi = {https://doi.org/10.1016/j.spasta.2026.100973},
url = {https://www.sciencedirect.com/science/article/pii/S2211675326000217},
author = {Shuwan Wang and Christopher K. Wikle and Athanasios C. Micheas and Jessica L. {Mark Welch} and Jacqueline R. Starr and Kyu Ha Lee}
}

\end{document}